\documentclass{osudissert96}

\makeatletter \@addtoreset{footnote}{chapter} \makeatother

\usepackage[hyperindex,pdfborder={0 0 0}]{hyperref}

\usepackage{dsfont}
\usepackage{ifthen}
\usepackage{bm} 
\usepackage{amsmath}
\usepackage{calc}

\usepackage{twoopt}

\usepackage[withbib,all]{authorindex}

\usepackage{tabls}
\usepackage{booktabs}

\usepackage{paralist}

\makeatletter

\renewcommand{\p@enumii}{\theenumi.}
\makeatother

\usepackage{amscd,mathrsfs}

\usepackage{subfigure}

\usepackage{graphicx}

\usepackage{hyperref}
\usepackage{pict2e}
\usepackage{amsmath}
\usepackage{bbm}
\usepackage{graphicx}
\usepackage{epsfig}
\usepackage{array}
\usepackage{psfrag}
\usepackage{amssymb}
\usepackage{mathdots}
\usepackage{color}
\usepackage{pstricks,pst-node,pst-text,pst-3d,pst-plot}
\usepackage{psfrag}
\usepackage{enumerate}
\usepackage{url,cite}
\usepackage{appendix}
\usepackage{amsfonts}%
\usepackage{amsthm}
\usepackage{dsfont}%
\usepackage{verbatim}%
\usepackage{setspace}
\usepackage{float}
\usepackage{url}
\usepackage{fancyhdr}
\usepackage[hmargin=0.6in , vmargin=1in]{geometry}
\usepackage{algorithm}
\usepackage{algorithmic}
\newtheorem{thm}{Theorem}
\newtheorem{lma}{Lemma}

\newcommand{\gammath}[1]{\gamma_{\rm th}\left( #1 \right)}
\newcommand{\gammathst}[1]{\gamma_{\rm th}^*\left( #1 \right)}
\newcommand{\bfGamma}{{\bf \Gamma}_{\rm th}}
\newcommand{\Gammast}[1]{\bfGamma^{\rm *} \lb #1 \rb}
\newcommand{\bfP}[1]{{\bf P}_{#1}}
\newcommand{\Pst}[1]{{\bf P}_{#1}^*}
\newcommand{\dgamdlamp}[1]{\frac{\partial \gammathst{#1}}{\partial \lambdapst}}
\newcommand{\dUdlamp}[1]{\frac{\partial U_{#1}^*}{\partial \lambdapst}}
\newcommand{\dSdlamp}[1]{\frac{\partial S_{#1}^*}{\partial \lambdapst}}
\newcommand{\dpdlamp}[1]{\frac{\partial p_{#1}^*}{\partial \lambdapst}}

\newcommand{\dUSpdlamp}[1]{\dUdlamp{#1} - \lambdap \dSdlamp{#1} + \lambdad \dpdlamp{#1}}

\newcommand{\Pist}[1]{P_{#1}^* \lb \gammathst{#1} \rb}
\newcommand{\pdf}{f_{\gamma}}
\newcommand{\ccdf}{\bar{F}_{\gamma}}
\newcommand{\lb}{\left (}
\newcommand{\rb}{\right )}
\newcommand{\script}[1]{{\mathcal {#1}}}

\newcommand{\fzb}{f_{z \vert b} \lb z \vert b_i=1 \rb}
\newcommand{\fzf}{f_{z \vert b} \lb z \vert b_i=0 \rb}
\newcommand{\Ibessel}[1]{I_{#1}^{\rm Bes}}

\newcommand{\Pavg}{P_{\rm avg}}
\newcommand{\Iavg}{I_{\rm avg}}
\newcommand{\rmin}{r_{\rm min}}
\newcommand{\pframe}{P_{\rm frame}}
\newcommand{\Dmax}{\bar{D}_{\rm max}}
\newcommand{\invDmaxline}{1/\Dmax}
\newcommand{\invDmax}{\frac{1}{\Dmax}}
\newcommand{\zb}{z_{\rm b}}
\newcommand{\zf}{z_{\rm f}}

\newcommand{\pmd}{P_{\rm MD}}
\newcommand{\pfa}{P_{\rm FA}}
\newcommand{\pskip}[1]{p_{#1}^{\rm skip}}
\newcommand{\Utwomax}{U_2^{\rm max}}
\newcommand{\ptwomax}{p_2^{\rm max}}
\newcommand{\Ts}{T}

\newcommand{\trho}{\tilde{\rho}}
\newcommand{\tS}{{\bf \tilde{S}}}
\newcommand{\sS}{\script{S}}

\newcommand{\tinset}{t \in \{1,...t_f\}}

\newcommand{\lambdad}{\lambda_{\rm D}}
\newcommand{\lambdadst}{\lambda_{\rm D}^*}
\newcommand{\lambdap}{\lambda_{\rm P}}
\newcommand{\lambdapst}{\lambda_{\rm P}^*}
\newcommand{\lambdapmin}{\lambdap^{\rm min}}
\newcommand{\lambdapmax}{\lambdap^{\rm max}}
\newcommand{\lambdadmax}{\lambdad^{\rm max}}

\newcommand{\Pmax}{P_{\rm max}}

\newcommand{\lambdai}{\lambda_{\rm I}}

\newcommand{\lambdaist}{\lambda_{\rm I}^*}

\newcommand{\PS}[1]{P_{#1}\lb \gamma \rb}
\newcommand{\PSst}[1]{P^*_{#1}\lb \gamma \rb}
\newcommand{\bfPS}[1]{{\bf P}_{#1}}

\newcommand{\gammaz}[1]{\gamma_{\rm th} \left(#1,z \right)}
\newcommand{\gammazst}[1]{\gamma_{\rm th}^* \left(#1,z \right)}
\newcommand{\GammaS}[1]{\bfGamma \lb #1,z\rb}
\newcommand{\Gammasst}[1]{\bfGamma^{\rm *} \lb #1,z\rb}
\newcommand{\Usoft}{U}
\newcommand{\Ssoft}{S}
\newcommand{\psoft}{p}
\newcommand{\Isoft}{I}
\newcommand{\Usoftst}{U^*}
\newcommand{\Ssoftst}{S^*}
\newcommand{\psoftst}{p^*}
\newcommand{\Isoftst}{I^*}

\newtheorem{Def}{Definition}
\DeclareMathOperator{\E}{\mathbb{E}}
\newcommand{\bP}{\bar{P}}
\newcommand{\Pmin}{P_{\rm min}}
\newcommand{\bfpi}{{\pmb \pi}}
\newcommand{\Iinst}{I_{\rm inst}}
\newcommand{\fgammai}{f_{\gamma_i}}
\newcommand{\fgi}{f_{g_i}}
\newcommand{\EE}[1]{\E \left[ #1 \right]}

\newcommand{\bgi}{\bar{g}_i}
\newcommand{\bW}{\bar{W}}

\newcommand{\bfPst}{{\bf P}^*}

\newcommand{\bfpist}{{\bm{\pi}}^*}

\newcommand{\bgamma}{\bar{\gamma}}

\newcommand{\Xvq}{\{X(k)\}_{k=0}^\infty}
\newcommand{\Yivq}{\{Y_i(k)\}_{k=0}^\infty}

\newcommand{\bfY}{{\bf Y}}

\newcommand{\pardef}[1]{\triangleq [#1_1^{(t)},\cdots,#1_N^{(t)}]^T}
\newcommand{\parFdef}[1]{\triangleq [#1_1(k),\cdots,#1_N(k)]^T}
\newcommand{\DOIC}{\emph{DOIC}}
\newcommand{\DOAC}{\emph{DOAC}}
\newcommand{\DOACopt}{\emph{DOAC-Pow-Alloc}}
\newcommand{\brho}{\bar{\rho}}
\newcommand{\Rit}{R_i^{(t)}}
\newcommand{\Rt}{R^{(t)}}
\newcommand{\FDurK}{\vert \script{F}_k\vert}
\newcommand{\Wup}{W^{\rm up}_{\pi_j}}

\newcommand{\gammaerr}{\gamma_i^{\rm err}(t)}
\newcommand{\gerr}{g_i^{\rm err}(t)}
\newcommand{\dOne}{29}
\newcommand{\dTwo}{40}
\newcommand{\Obj}{O_{\rm obj}}

\newcommand{\Pminn}{P^{\rm min}}
\newcommand{\gmax}{g_{\rm max}}
\newcommand{\EEU}[1]{\E_{\bfU(k)} \left[ #1 \right]}
\newcommand{\EEY}[1]{\E_{\bfY(k)} \left[ #1 \right]}
\newcommand{\Prob}[1]{{\bf Pr}\left[ #1 \right]}
\newcommand{\Pit}{P_i^{(t)}}
\newcommand{\git}{g_i^{(t)}}
\newcommand{\bfU}{{\bf U}}
\newcommand{\bfr}{{\bf r}}

\newcommand{\psiI}{\psi_{\pi_j}^{\rm I}}
\newcommand{\psiD}{\psi_{\pi_j}^{\rm D}}
\newcommand{\psiInt}{\psi^{\rm I}}
\newcommand{\psiDel}{\psi^{\rm D}}

\newcommand{\gammait}{\gamma_i^{(t)}}
\newcommand{\tP}{{\bf \tilde{P}}}
\newcommand{\tPsi}{\tilde{\Psi}}

\newcommand{\brhomax}{\brho^{\rm max}}

\newcommand{\sB}{s_{\rm B}}
\newcommand{\sNB}{s_{\rm NB}}

\newcommand{\gammamini}{\gamma_i^{\rm min}}
\newcommand{\RBt}{R_{\rm B}^{(t)}}

\newcommand{\HOL}{L^{\rm rem}}
\newcommand{\widthn}{0.8}
\newcommand{\Rmax}{R_{\rm max}}
\newcommand{\gammamax}{\gamma_{\rm max}}
\newcommand{\DF}{D_{\rm F}}
\begin{document}


\author{Ahmed Emad Ewaisha}
\title{Scheduling and Power Allocation to Optimize Service and Queue-Waiting Times in Cognitive Radio Uplinks}
\graduationyear{January $13^{\rm th}$, 2016}


\maketitle


\startsinglespace

\startdoublespace



\tableofcontents



\begin{abstract}
%
%
%

Cognitive Radio (CR) is an emerging wireless communication paradigm to improve the spectrum utilization. Cognitive Radio users, also known as secondary users (SUs), are allowed to transmit over the channels as long as they do not cause harmful interference to the primary users (PUs) who have licensed access to those channels. The quality of service (QoS) associated with this transmission scheme might deteriorate if those channels are not utilized efficiently by these SUs. In addition, SUs located physically close to PUs might cause more harmful interference than those who are far. This might degrade the QoS of those SUs since they will be allocated the channel less frequently to protect the PUs.

In this report, we study the packet delay as a QoS metric in CR systems. The packet delay is defined as the average time spent by a packet in the queue waiting for transmission as well as that spent during the transmission process. The former is referred to as the queue waiting time while the latter is the service time. In real-time applications, the average delay of packets needs to be below a prespecified threshold to guarantee an acceptable QoS. In this work, we study the effect of both the scheduling and the power allocation algorithms on the delay performance of the SUs. We study how these two parameters affect both the service time as well as the queue waiting time.

To study the delay due to the service time we study the effect of multiple channels on a single SU and, thus, ignore the scheduling problem. Specifically, in a multichannel system where the channels are sensed sequentially, we study the tradeoff between throughput and delay. The problem is formulated as an optimal stopping rule problem where it is required to decide at which channel the SU should stop sensing and begin transmission. We provide a closed-form solution for this optimal stopping problem and specify the optimal amount of power that this SU should be transmitting with over this channel. The algorithm trades off the service time versus the throughput to guarantee a maximum throughput performance subject to a bound on the average service time. This tradeoff results from skipping low-quality channels to seek the possibility of finding high-quality ones in the future at the expense of a higher probability of being blocked from transmission since these future channels might be busy.

On the other hand, the queue waiting time is studied by considering a multi-SU single channel system. Specifically, we study the effect of scheduling and power allocation on the delay performance of all SUs in the system. We propose a delay optimal algorithm to this problem that schedules the SUs to minimize the delay while protecting the PUs from harmful interference. One of the contributions of this algorithm is that it can provide differentiated service to the users even if their channels are statistically heterogeneous. In heterogeneous-channels system, users with statistically low channel quality are expected to have worse delay performances. However, the proposed algorithm guarantees a prespecified delay performance to each SU without violating the PU's interference constraint. Existing scheduling algorithms do not provide such guarantees if the interference channels are heterogeneous. This is because they are developed for conventional non-CR wireless systems that neglect interference since channels are orthogonal.

Finally, we present two potential extensions to these studied problems. In the first one, instead of imposing a constraint on the average delay as assumed in this report, we impose strict deadlines by which the packets need to be transmitted. Problems with strict deadlines have their applications in live streaming and online gaming where packets are expected to reach their destination before a prespecified deadline expires. If a packet misses its deadline it is dropped from the system and does not count towards the throughput. However, these applications can tolerate a small percentage out of the total packets missing their deadlines. We have solved the single-SU version of this problem in Chapter \ref{Ch2_Sequential_Sensing}. Extensions to the multi-user case is an interesting problem.

The second extension is to study the scheduling problem at hand aiming at finding scheduling algorithms that are throughput optimal and delay optimal at the same time. Except for special cases of our problem, the proposed scheduling and power allocation policy does not achieve the capacity region. However, our preliminary results show that their exist throughput-optimal scheduling algorithms that are well studied in the literature that can be developed to be delay optimal as well.
\end{abstract}

\chapter{Introduction}
\label{Ch1_Introduction}
Cognitive Radio (CR) systems are emerging wireless communication systems that allow efficient spectrum utilization \cite{Survey_CR_1st_2006_Akyildiz}. CRs refer to devices that coexist with the licensed spectrum owners called the primary users (PUs), and that are capable of detecting their presence. Once PU's activity is detected on some frequency channel, the CR user refrains from any further transmission on this channel. This may result in service disconnection for the CR user, thus degrading the quality of service (QoS). If the CR users have access to other channels, the QoS can be improved by switching to another frequency channel instead of completely stopping transmission. If not, then they should control their transmission power to avoid harmful interference to the PUs. Hence, CR users are required to adjust their transmission power levels, and -thus- their rates, according to the interference level the PUs can tolerate. This adjustment could lead to severe degradation in the QoS provided for the CR users, if not designed carefully.

\section{Cognitive Radio Transmission Schemes}
\begin{figure}%
\centering
\includegraphics[width=\widthn\columnwidth]{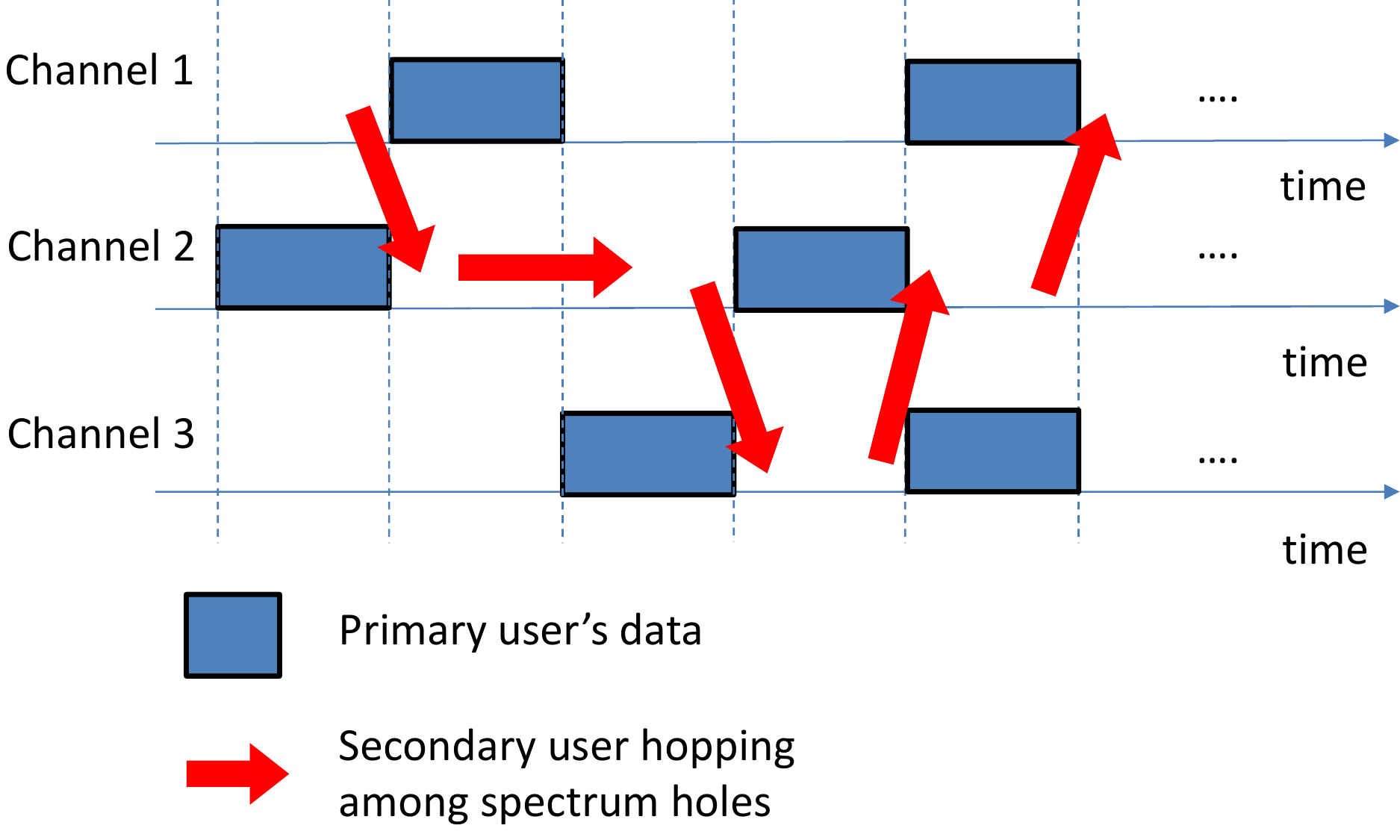}%
\caption{Spectrum holes are the locations of the unused spectrum in time and frequency.}
\label{Spectrum_Holes}%
\end{figure}
There are two main transmission schemes that CR systems may follow to coexist with the PUs; the overlay and the underlay. In the overlay, CR users, also referred to as the secondary users (SUs), transmit their signal only when the PUs are not using the channel. In other words, the SUs look for the spectrum holes to transmit their data as in Fig. \ref{Spectrum_Holes}. Hence, unlike conventional radios, SUs's radios are equipped with a spectrum sensor that is used to sense the spectrum before beginning the transmission phase. In this sensing phase, the SUs listen to all frequency channels to overhear the PUs' transmission so as to decide which channels are free from PUs and which are not. Upon this detection process, the SU picks up a channel, or more, out of the detected-free channels to transmit its data over for a limited amount of time. Once the channel is occupied again by the PU, the SU is expected to refrain from transmission over this channel but allowed to use a different channel after performing the sensing phase again. A practical spectrum sensor might yield wrong decisions, namely, it might detect the presence of a PU on some channel although this channel is actually free, or might miss-detect the PU when it is using the channel. These events are referred to as the false-alarm and miss-detection events, respectively. The higher the false-alarm probability the higher the SU misses transmission opportunities and, thus, the lower the SU's throughput is. Similarly, the higher the probability of miss-detection the more the SU's packet collides with the PU's and leading to a lower throughput since collided packets are lost. While the false-alarm probability affects the SU's throughput alone, the miss-detection probability affects both the SU and the PU. As the sensing phase duration increases, these two probabilities decrease simultaneously. However, increasing the sensing phase duration comes at the expense of the transmission phase duration thus decreasing the throughput. This tradeoff has been studied extensively in the literature \cite{Ewaisha2011Optimization}.

\begin{figure}%
\centering
\includegraphics[width=\widthn\columnwidth]{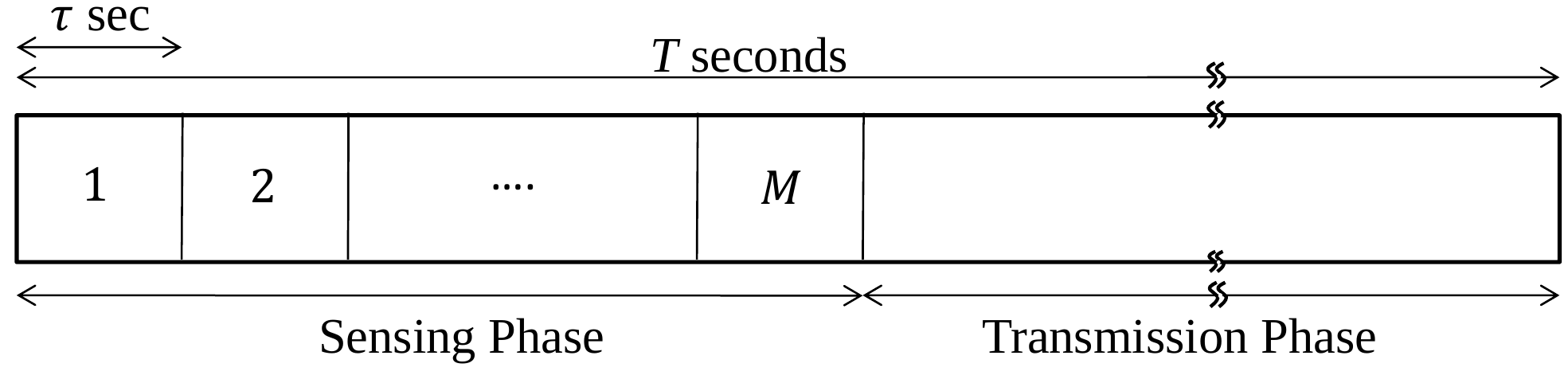}%
\caption{The sensing phase is used to sense $M$ channels to detect the presence of the PU. The SU starts transmitting its data in the transmission phase on one of the free channels.}
\label{Sensing_Phase}%
\end{figure}

In the underlay scheme, the SU is allowed to transmit over any frequency channel at any time as long as the PU can tolerate the interference caused by this transmission. This tolerable level is referred to as the interference temperature as dictated by the Federal Communications Commission (FCC) \cite{FCC_Report_SEWG}. In order to guarantee this protection for the PU, the SU has to adjust its transmission power according to the gain of the channel to the primary receiver referred to as the interference channel. The knowledge of this gain instantaneously is essential at the SU's transmitter. While this channel knowledge might be infeasible in CR systems that assume no cooperation between the PU and the SU, in some scenarios the SU might be able to overhear the pilots sent by the primary receiver when it is acting as a transmitter if the PU is using a time division duplex scheme.

In both cases, the overlay and the underlay, the SU might interfere with the PU. This in turn dictates that the SU should adopt its channel access scheme in such a way that this interference is tolerable so that the PU's quality of service (QoS) is not degraded. With that being said, we might expect that the SUs located physically closer to the PUs might suffer larger degradation in their QoS compared to those that are far because closer SUs transmit with smaller amounts of power. This problem does not appear in conventional non CR cellular systems since frequency channels tend to be orthogonal in non CR systems. In other words, in non CR systems, all users are allocated the channels via some scheduler that guarantees those users do not interfere with each other. While in CR systems, since SUs interfere with PUs, we need to develop scheduling and power control algorithms that prevent harmful interference to PUs, as well as guaranteeing acceptable QoS for the SUs.

\section{Guaranteeing Quality of Service in Cognitive Radio Systems}
Since CR users operate in an interference limited environment, they are expected to experience lower QoS than in conventional systems. However, the QoS provided needs to fall within the acceptable level that varies with the application. For example, the average delay of a packet in online streaming is required to not more than 300ms while that in online gaming should not exceed 50ms. However, these two applications might tolerate small losses in their transmitted packets which is not the case with some other applications as file sharing and email applications that, on the other hand, might tolerate packet delays.

The QoS can include, but is not limited to, throughput, delay, bit-error-rate, interference caused to the PU. Out of these metrics the most two major ones are the throughput and the delay that have gained strong attention in the literature recently \cite{asadi2013survey}. The throughput metric is defined as the average amount of packets (or bits) per channel-use that can be delivered in the SU's network without violating the PU's interference constraints. On the other hand, the delay refers to as the amount of time elapsed from the instant a packet joins the SU's buffer until it is successfully and fully transmitted to its intended receiver. A higher throughput is usually achieved by the efficient power allocation algorithms while better delay performances are usually achieved by efficient scheduling of users.

The problem of scheduling and/or power control has been widely studied in the literature (see \cite{Adaptive_Rate_Power_CR_Sonia,Letaief_PU_Known_Location,NEP_Distributed,Iter_Bit_Allocation_OFDM,6464638,shakkottai2002scheduling,kang2013performance}, and references therein). These works aim at optimizing the throughput, providing delay guarantees and/or guaranteeing protection from interference. In real-time applications, such as audio/video conference calls, one of the most important QoS metrics is the delay metric. The delay is defined as the average amount of time a packet spends in the system starting from the instant it arrives to the buffer until it is completely transmitted. In real-time applications, packets are expected to arrive at the destination before a prespecified deadline \cite{IEEE802p11}. Thus, the average packet delay needs to be as small as possible to prevent jitter and to guarantee acceptable QoS for these applications \cite{shakkottai2002scheduling,kang2013performance}.

There are two different factors that cause delay in data networks. The first is the service time which is the amount of time required to transmit this packet. The second is the queue-waiting time which is the time spent by a packet in the queue waiting for its transmission to begin. The sum of both yields the delay. Thus, in order to optimize over the delay we should study both factors.

\section{Service Time}
The service time is affected by the amount of resources allocated to the packet at the time of transmission. Resources might include power, channel bandwidth, coding rate and transmission time. Several works have been studied to address how to optimally allocate these resources over time and users. However, from a practical implementation point of view, the most challenging resource is channel bandwidth. This is because increasing the bandwidth requires allocating multiple channels to a user which might require the user to be equipped with high cost transmitters (receivers) capable of transmitting (receiving) over multiple channels simultaneously. On the other hand, allocating a single fixed channel to a user is not optimal.

The problem of channel allocation in multi-channel CR systems has gained attention in recent works due to the challenges associated with the sensing and access mechanisms in a multichannel CR system. Practical hardware constraints on the SUs' transceivers may prevent them from sensing multiple channels simultaneously to detect the state of these channels (free/busy). This leads the SU to sense the channels sequentially, then decide which channel should be used for transmission \cite{POMDP_Qing_Zhao,Sensing_Order_Poor}. In a time slotted system if sequential channel sensing is employed, the SU senses the channels one at a time and stops sensing when a channel is found free. But due to the independent fading among channels, the SU is allowed to skip a free channel if its quality, measured by its power gain, is low and sense another channel seeking the possibility of a higher future gain. Otherwise, if the gain is high, the SU stops at this free channel to begin transmission. The question of when to stop sensing can be formulated as an optimal stopping rule problem \cite{Sabharwal_NonCR_Multiband,Sensing_Order_Poor,Jia_Multichannel_Tx_Opt_Stop_Rule, Cheng_Simple_Chan_Sensing}. In \cite{Sabharwal_NonCR_Multiband} the authors present the optimal stopping rule for this problem in a non-CR system. The work in \cite{Sensing_Order_Poor} develops an algorithm to find the optimal order by which channels are to be sequentially sensed in a CR scenario, whereas \cite{Jia_Multichannel_Tx_Opt_Stop_Rule} studies the case where the SUs are allowed to transmit on multiple contiguous channels simultaneously. The authors presented the optimal stopping rule for this problem in a non-fading wireless channel. Transmissions on multiple channels simultaneously may be a strong assumption for low-cost transceivers especially when they cannot sense multiple channels simultaneously.

In general, if a perfect sensing mechanism is adopted, the SU will not cause interference to the PU since the former transmits only on spectrum holes (referred to as an overlay system). Nevertheless, if the sensing mechanism is imperfect, or if the SU's system is an underlay one (where the SU uses the channels as long as the interference to the PU is tolerable), the transmitted power needs to be controlled to prevent harmful interference to the PU. References \cite{Pei2013Sensing} and \cite{Adaptive_Rate_Power_CR_Sonia} consider power control and show that the optimal power control strategy is a water-filling approach under some interference constrain imposed on the SU transmitter. Yet, all of the above work studies single channel systems which cannot be extended to multiple channels in a straightforward manner. A multiuser CR system was considered in \cite{Hu_MultiCR_Contention} in a time slotted system. To allocate the frequency channel to one of the SUs, the authors proposed a contention mechanism that does not depend on the SUs' channel gains, thus neglecting the advantage of multiuser diversity. A major challenge in a multichannel system is the sequential nature of the sensing where the SU needs to take a decision to stop and begin transmission, or continue sensing based on the information it has so far. This decision needs to trade-off between waiting for a potentially higher throughput and taking advantage of the current free channel. Moreover, if transmission takes place on a given channel, the SU needs to decide the amount of power transmitted to maximize its throughput given some average interference and average power constraints.


In Chapter \ref{Ch2_Sequential_Sensing}, we model the overlay and underlay scenarios of a multi-channel CR system that are sensed sequentially. The problem is solved for a single SU first then we discuss extensions to a multi-SU scenario. For the single SU case, the problem is formulated as a joint optimal-stopping-rule and power-control problem with the goal of maximizing the SU's throughput subject to average power and average interference constraints. This formulation results in increasing the expected service time of the SU's packets. The expected service time is the average number of time slots that pass while the SU attempts to find a free channel, before successfully transmitting a packet. The increase in the service time is due to skipping free channels, due to their poor gain, hoping to find a future channel of sufficiently high gain. If no channels having a satisfactory gain were found, the SU will not be able to transmit its packet, and will have to wait for longer time to find a satisfactory channel. This increase in service time increases the queuing delay. Thus, we solve the problem subject to a bound on the expected service time which controls the delay. In the multiple SUs case, we show that the solution to the single SU problem can be applied directly to the multi-SU system with a minor modification. We also show that the complexity of the solution decreases when the system has a large number of SUs.

To the best of our knowledge, this is the first work to study the joint power-control and optimal-stopping-rule problem in a multi channel CR system. Our contribution in this work is the formulation of a joint power-control and optimal-stopping-rule problem that also incorporates a delay constraint and present a low complexity solution in the presence of interference/collision constraint from the SU to the PU due to the imperfect sensing mechanism. The preliminary results in \cite{Ewaisha_Throughput_Maximization} consider an overlay framework for single user case while neglecting sensing errors. But in this work, we also study the problem in the underlay scenario where interference is allowed from the secondary transmitter (ST) to the primary receiver (PR) and extend to multiple SU case. We also generalize the solution to the multi-SU case when the number of SUs is large. We discuss the applicability of our formulation in typical delay-constrained scenarios where packets arrive simultaneously and have a same deadline. We show that the proposed algorithm can be used to solve this problem offline, to maximize the throughput and meet the deadline constraint at the same time. Moreover, we propose an online algorithm that gives higher throughput compared to the offline approach while meeting the deadline constraint.

\section{Queue-Waiting Time}
Unlike the service time, the delay due to queue-waiting time is affected by the scheduling algorithm. The more frequently a user is allocated the channel for transmission, the less its queue-waiting time is, but the more the queue-waiting times for the other users are. Delay due to the queue-waiting time is also well studied recently in the literature and scheduling algorithms have been proposed to guarantee small delay for users in conventional systems \cite{li2011delay,Two_Q_Light_Hvy,neely2003power}. In \cite{li2011delay}, the authors study the joint scheduling-and-power-allocation problem in the presence of an average power constraint. Although in \cite{li2011delay} the proposed algorithm offers an acceptable delay performance, all users are assumed to transmit with the same power. A power allocation and routing algorithm is proposed in \cite{neely2003power} to maximize the capacity region under an instantaneous power constraint. While the authors show an upper bound on the average delay, this delay performance is not guaranteed to be optimal.

Although queuing theory, that was originally developed to model packets at a server, can be applied to wireless channels, the challenges are different. One of the main challenges is the fading nature of the wireless channel that changes from a slot to another. Fading requires adapting the user's power and/or rate according to the channel's fading coefficient. The idea of power and/or rate adaptation based on the channel condition does not have an analogy in server problems and, thus, is absent in the aforementioned references. Instead, existing works treat wireless channels as on-off fading channels and do not consider multiple fading levels. Among the relevant references that consider a more general fading channel model are \cite{neely2003power}, which was discussed above, \cite{Fading_No_Scheduling,E_Hossain_CR_Delay_Analysis} where the optimization over the scheduling algorithm was out of the scope of their work, and \cite{Min_Pow_4_Delay_NonCR} that neglects the interference constraint since it considers a non CR system.

In contrast with \cite{Letaief_PU_Known_Location,NEP_Distributed,Ewaisha_TVT2015,Iter_Bit_Allocation_OFDM,6464638} that do not optimize the queuing delay, the problem of minimizing the sum of SUs' average delays is considered in this work. The proposed algorithm guarantees a bound on the instantaneous interference to the PUs, a guarantee that is absent in \cite{li2011delay,neely2003power}. Based on Lyapunov optimization techniques \cite{li2011delay}, an algorithm that dynamically schedules the SUs as well as optimally controlling their transmission power is presented. The contributions in this work are: i) Proposing a joint power-control and scheduling algorithm that is optimal with respect to the average delay of the SUs in an interference-limited system; ii) Showing that the proposed algorithm can provide differentiated service to the different SUs based on their heterogeneous QoS requirements. Moreover, the complexity of the algorithm is shown to be polynomial in time since it is equivalent to that of sorting a vector of $N$ numbers, where $N$ is the number of SUs in the system.

\chapter{Delay Due To Service Time}
\label{Ch2_Sequential_Sensing}

In this chapter we study the delay resulting from the service time of packets and neglect the delay resulting from the waiting time in the queues. We treat the cognitive radio system as a single secondary user (SU) accessing a multi-channel system. The main problem studied in this chapter is the tradeoff between the service time and the throughput. We assume the SU senses the channels sequentially to detect the presence of the primary user (PU), and stops its search to access a channel if it offers a significantly high throughput. The tradeoff exists because stopping at early-sensed channels gives low average service time but, at the same time, gives low throughput since early channels might have low gains. The joint optimal stopping rule and power control problem is formulated as a throughput maximization problem with an average service time and power constraint. We note that in this chapter we use the word delay to refer to the service time.

\section{Overlay System Model}
\label{Overlay_System_Model}
Consider a PU network that has a licensed access to $M$ orthogonal frequency channels. Time is slotted with a time slot duration of $\Ts$ seconds.
The SU's network consists of a single ST (SU and ST will be used interchangeably) attempting to send real-time data to its intended secondary receiver (SR) through one of the channels licensed to the PU. Before a time slot begins, the SU is assumed to have ordered the channels according to some sequence (we note that the method of ordering the channels is outside the scope of this work. The reader is referred to \cite{Sensing_Order_Poor} for further details about channel ordering), labeled $1,...,M$. The set of channels is denoted by $\script{M}=\{1,...,M\}$. Before the SU attempts to transmit its packet over channel $i$, it senses this channel to determine its availability ``state'' which is described by a Bernoulli random variable $b_i$ with parameter $\theta_i$ ($\theta_i$ is called the availability probability of channel $i$). If $b_i=0$ (which happens with probability $\theta_i$), then channel $i$ is free and the SU may transmit over it until the on-going time slot ends. If $b_i=1$, channel $i$ is busy, and the SU proceeds to sense channel $i+1$. Channel availabilities are statistically independent across frequency channels and across time slots.

We assume that the SU has limited capabilities in the sense that no two channels can be sensed simultaneously. This may be the case when considering radios having a single sensing module with a fixed bandwidth, so that it can be tuned to only one frequency channel at a time. The reader is referred to \cite{Spectrum_Sensing_via_Kolmogorov_Smirnov_Test}, \cite{Zou2011_CognitiveRelaySelection} and \cite{Robust_Spectrum_Sensing_With_Crowd_Sensors} for detailed information on advanced spectrum sensing techniques. Therefore, at the beginning of a given time slot, the SU selects a channel, say channel $1$, senses it for $\tau$ seconds ($\tau \ll \Ts/M$), and if it is free, the SU transmits on this channel if its channel gain is high enough\footnote{How ``high'' is ``high'' is going to be explained later}. Otherwise, the SU skips this channel and senses channel $2$, and so on until it finds a free channel. If all channels are busy (i.e. the PU has transmission activities on all $M$ channels) then this time slot will be considered as ``blocked''. In this case, the SU waits for the following time slot and begins sensing following the same channel sensing sequence. As the sensing duration increases, the transmission phase duration decreases which decreases the throughput. But we cannot arbitrarily decrease the value of $\tau$ since this decreases the reliability of the sensing outcome. This trade-off has been studied extensively in the literature, e.g. \cite{Liang2008_SensingThroughputTradeoff}, \cite{Zou2010_AsymptoticOutageProb}. In this work we study the impact of sequential channel sensing on the throughput rather than the sensing duration on the throughput. Hence we assume that $\tau$ is a fixed parameter and is not optimized over. For details on the trade-off between throughput and sensing duration in this sequential sensing problem the reader is referred to \cite{Ewaisha2011Optimization}.

\begin{figure}
	\centering
		\includegraphics[width=\widthn\columnwidth]{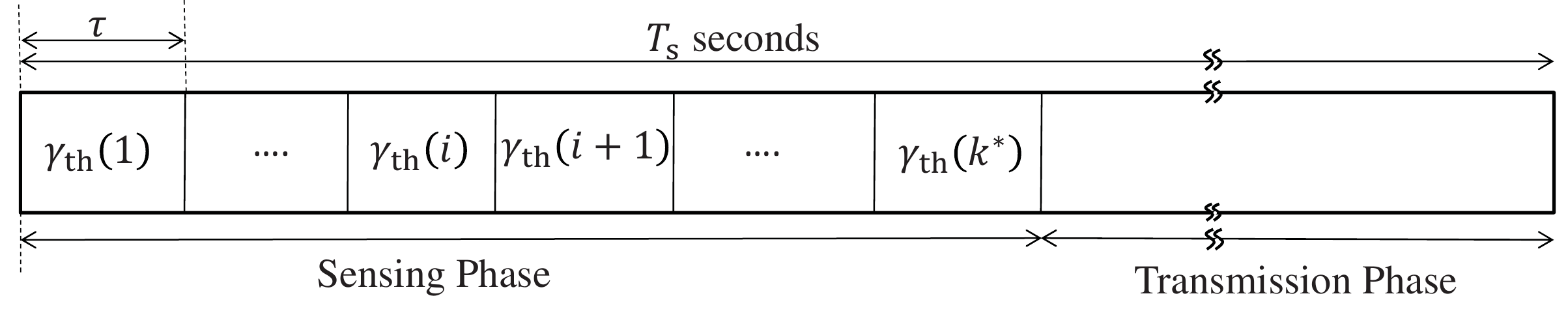}
	\caption{Sensing and transmission phases in one time slot. The SU senses each channel for $\tau$ seconds, determines its state, then probes the gain if the channel is found free. The sensing phase ends if the probed gain $\gamma_i>\gammath{i}$, in which case $k^*=i$. Hence, $k^*$ is a random variable that depends on the channel states and gains.}
	\label{Time_Slot_Fig}
\end{figure}

The fading channel between ST and SR is assumed to be flat fading with independent, identically distributed (i.i.d.) channel gains across the $M$ channels. To achieve higher data rates, the SU adapts its data rate according to the instantaneous power gain of the channel before beginning transmission on this channel. To do this, once the SU finds a free channel, say channel $i$, the gain $\gamma_i$ is probed. The data rate will be proportional to $\log(1+P_{1,i}(\gamma_i){\gamma}_i)$, where $P_{1,i}(\gamma_i)$ is the power transmitted by the SU at channel $i$ as a function of the instantaneous gain \cite{Goldsmith_Wireless_Comm}. Fig. \ref{Time_Slot_Fig} shows a potential scenario where the SU senses $k^*$ channels, skips the first $k^*-1$, and uses the $k^*$th channel for transmission until the end of this on-going time slot. In this scenario the SU ``stops'' at the $k^*$th channel, for some $k^*  \in \script{M}$. Stopping at channel $i$ depends on two factors: 1) the availability of channel $b_i$, and 2) the instantaneous channel gain $\gamma_i$. Clearly, $b_i$ and $\gamma_i$  are random variables that change from one time slot to another. Hence, $k^*$, that depends on these two factors, is a random variable. More specifically, it depends on the states $[b_1,...,b_M]$ along with the gains of each channel $[\gamma_1,...,\gamma_M]$. To understand why, consider that the SU senses channel $i$, finds it free and probes its gain $\gamma_i$. If ${\gamma}_i$ is found to be low, then the SU skips channel $i$ (although free) and senses channel $i+1$. This is to take advantage of the possibility that $\gamma_j \gg \gamma_i$ for $j>i$. On the other hand, if $\gamma_i$ is sufficiently large, the SU stops at channel $i$ and begins transmission. In that latter case $k^*=i$. Defining the two random vectors $\underline{b}=[b_1,...,b_M]^T$ and $\underline{\gamma}=[\gamma_1,...,\gamma_M]^T$, $k^*$ is a deterministic function of $\underline{b}$ and $\underline{\gamma}$.

We define the stopping rule by defining a threshold $\gammath{i}$ to which each $\gamma_i$ is compared when the $i$th channel is found free. If $\gamma_i \geq \gammath{i}$, channel $i$ is considered to have a ``high'' gain and hence the SU ``stops'' and transmits at channel $i$. Otherwise, channel $i$ is skipped and channel $i+1$ sensed. In the extreme case when $\gammath{i}=0$, the SU will not skip channel $i$ if it is found free. Increasing $\gammath{i}$ allows the SU to skip channel $i$ whenever ${\gamma}_i<\gammath{i}$, to search for a better channel, thus potentially increasing the throughput. Setting $\gammath{i}$ too large allows channel $i$ to be skipped even if $\gamma_i$ is high. This constitutes the trade-off in choosing the thresholds $\gammath{i}$. The optimal values of $\gammath{i}$ $i\in \script{M}$, determine the optimal stopping rule.

Let $P_{1,i}(\gamma)$ denote the power transmitted at the $i$th channel when the instantaneous channel gain is $\gamma$, if channel $i$ was chosen for transmission. Since the SU can transmit on one channel at a time, the power transmitted at any time slot at channel $i$ is $P_{1,i}(\gamma_i) \mathds{1} \left(i=k^* \right)$, where $\mathds{1} \left(i=k^* \right)=1$ if $i=k^*$ and $0$ otherwise.
Define $c_i \triangleq 1-\frac{i\tau}{\Ts}$ as the fraction of the time slot remaining for the SU's transmission if the SU transmits on the $i$th channel in the sensing sequence. The average power constraint is $\mathbb{E}_{\underline{\gamma},\underline{b}} [c_{k^*} P_{k^*}(\gamma_{k^*})] \leq P_{\rm avg}$, where the expectation is with respect to the random vectors $\underline{\gamma}$ and $\underline{b}$. We will henceforth drop the subscript from the expected value operator ${\mathbb E}$. This expectation can be calculated recursively from
\begin{equation}
S_i(\bfGamma(i),\bfP{1,i})=\theta_i c_i \int_{\gammath{i}}^\infty{P_{1,i}(\gamma) f_{\gamma_i}(\gamma) \,d\gamma}+ \left[1-\theta_i \bar{F}_{\gamma_i}(\gammath{i}) \right]S_{i+1}(\bfGamma(i+1),\bfP{i+1}),
\label{Average_Power}
\end{equation}
$i \in \script{M}$, where $\bfP{1,i} \triangleq [P_{1,i}(\gamma),...,P_{1,M}(\gamma)]^T$ and $\bfGamma(i) \triangleq [\gammath{i},...,\gammath{M}]^T$ are the vectors of the power functions and thresholds respectively, with $S_{M+1}(\bfGamma(M+1),\bfP{M+1})  \triangleq  0$, $f_{\gamma_i}(\gamma)$ is the Probability Density Function (PDF) of the gain $\gamma_i$ of channel $i$, and $\bar{F}_{\gamma_i}(x) \triangleq \int_x^{\infty}{f_{\gamma_i}(\gamma) \, d\gamma}$ is the complementary cumulative distribution function. The first term in (\ref{Average_Power}) is the average power transmitted at channel $i$ given that channel is chosen for transmission (i.e. given that $k^*=i$). The second term represents the case where channel $i$ is skipped and channel $i+1$ is sensed. It can be shown that $S_1(\bfGamma(1),\bfP{1,1})=\mathbb{E} \left[c_{k^*} P_{k^*}(\gamma) \right]$. Moreover, we will also drop the index $i$ from the subscript of $f_{\gamma_i}(\gamma)$ and $\bar{F}_{\gamma_i}(\gamma)$ since channels suffer i.i.d. fading. Although we have only included an average power constraint in our problem, we will modify, after solving the problem, the solution to include an instantaneous power constraint as well.

The SU's average throughput is defined as $\mathbb{E} [c_{k^*} \log(1+P_{k^*}(\gamma_{k^*})\gamma_{k^*})]$. Similar to the average power, we denote the expected throughput as $U_1(\bfGamma(1),\bfP{1,1})$ which can be derived using the following recursive formula
\begin{align}
\nonumber U_i(\bfGamma(i),&\bfP{1,i})=\theta_i c_i \int_{\gammath{i}}^{\infty}{\log \lb 1+P_{1,i}(\gamma)\gamma \rb f_{\gamma}(\gamma)} \, d\gamma +\\
& \left[1-\theta_i \bar{F}_{\gamma}(\gammath{i}) \right] U_{i+1} \lb \bfGamma(i+1),\bfP{i+1} \rb
\label{Reward}
\end{align}
$i\in \script{M}$, with $U_{M+1}(\cdot,\cdot) \triangleq 0$. $U_1(\bfGamma(1),\bfP{1,1})$ represents the expected data rate of the SU as a function of the threshold vector $\bfGamma(1)$ and the power function vector $\bfP{1,1}$.

If the SU skips all channels, either due to being busy, due to their low gain or due to a combination of both, then the current time slot is said to be blocked. The SU has to wait for the following time slot to begin searching for a free channel again. This results in a delay in serving (transmitting) the SU's packet. Define the delay $D$ as the number of time slots the SU consumes before successfully transmitting a packet. That is, $D-1$ is a random variable that represents the number of consecutively blocked time slots. In real-time applications, there may exist some average delay requirement $\bar{D}_{\rm{max}}$ on the packets that must not be exceeded. Since the availability of each channel is independent across time slots, $D$ follows a geometric distribution having $\mathbb{E}[D]= \left(\rm{Pr}[\rm{Success}]\right)^{-1}$ where $\rm{Pr}[\rm{Success}]=1-\rm{Pr}[\rm{Blocking}]$. In other words, $\rm{Pr}[\rm{Success}]$ is the probability that the SU finds a free channel with high enough gain so that it does not skip all $M$ channels in a time slot. It is given by $\rm{Pr}[\rm{Success}]  \triangleq  p_1(\bfGamma(1))$ which can be calculated recursively using the following equation
\begin{equation}
\label{Prob_recursive}
p_i(\bfGamma(i))=\theta_i \bar{F}_\gamma(\gammath{i})+ \left[1-\theta_i \bar{F}_\gamma(\gammath{i}) \right]p_{i+1}(\bfGamma(i+1)),
\end{equation}
$i\in \script{M}$, where $p_{M+1} \triangleq 0$. Here, $p_i(\bfGamma(i))$ is the probability of transmission on channel $i$, $i+1$,..., or $M$.

\section{Problem Statement and Proposed Solution}
\label{Problem_Formulation}
From equation (\ref{Reward}) we see that the SU's expected throughput $U_1$ depends on the threshold vector $\bfGamma(1)$ and the power vector $\bfP{1,1}$. The goal is to find the optimum values of $\bfGamma(1) \in {\mathbb R}^M$ and functions $\bfP{1,1}$ that maximize $U_1$ subject to an average power constraint and an expected packet delay constraint. The delay constraint can be written as $\mathbb{E}[D] \leq \Dmax$ or, equivalently, $p_1(\bfGamma(1)) \geq \invDmaxline$. Mathematically, the problem becomes

\begin{equation}
\begin{array}{ll}
\rm{maximize}& U_1(\bfGamma(1),\bfP{1,1})\\
\label{Prob_Opt_Pow_Control}
\rm{subject \; to} &S_1(\bfGamma(1),\bfP{1,1}) \leq \Pavg\\
& p_1(\bfGamma(1)) \geq \frac{1}{\bar{D}_{\rm{max}}}\\
\rm{variables} & \bfGamma(1),\bfP{1,1},
\end{array}
\end{equation}
where the first constraint represents the average power constraint, while the second is a bound on the average packet delay. 
We allow the power $P_{1,i}$ to be an arbitrary function of $\gamma_i$ and optimize over this function to maximize the throughput subject to average power and delay constraints. Even though (\ref{Prob_Opt_Pow_Control}) is not proven to be convex, we provide closed-form expressions for the optimal thresholds and power-functions vector. To this end, we first calculate the Lagrangian associated with (\ref{Prob_Opt_Pow_Control}). Let $\lambda_{\rm P}$ and $\lambda_{\rm D}$ be the dual variables associated with the constraints in problem (\ref{Prob_Opt_Pow_Control}). The Lagrangian for (\ref{Prob_Opt_Pow_Control}) becomes
\begin{align}
\nonumber L&\left(\bfGamma(1),\bfP{1,1}, \lambda_{\rm P}, \lambda_{\rm D} \right)=U_1 \left( \bfGamma(1),\bfP{1,1} \right)-\\
&\lambda_{\rm P} \left( S_1(\bfGamma(1),\bfP{1,1}) -P_{\rm{avg}}\right)+ \lambda_{\rm D} \left( p_1(\bfGamma(1)) - \frac{1}{\bar{D}_{\rm{max}}} \right).
\label{Lagrange_Optimum_Pow_Control}
\end{align}
Differentiating (\ref{Lagrange_Optimum_Pow_Control}) with respect to each of the primal variables $P_{1,i}(\gamma)$ and $\gammath{i}$ and equating the resulting derivatives to zero, we obtain the KKT equations below which are necessary conditions for optimality \cite{Cvx_Boyd}, \cite{Miersemann_Calc_Var}:
\begin{align}
	\label{Water_Filling}
	&P_{1,i}^*(\gamma)=\left(\frac{1}{\lambdapst} - \frac{1}{\gamma}\right)^+, \hspace{0.5cm} \gamma> \gammathst{i},\\
	\nonumber &\log \left(1+\left(\frac{1}{\lambdapst}-\frac{1}{\gammathst{i}}\right)^+\gammathst{i} \right) - \lambdapst \left(\frac{1}{\lambdapst} - \frac{1}{\gammathst{i}} \right)^+\\
	&= \frac{U_{i+1}^* - \lambdapst S_{i+1}^*-\lambdadst \cdot \left( 1-p_{i+1}^* \right)}{c_i},
	\label{gamma_i_Equation}\\
\label{Primal_Dual_Feasible}
	&S_1^* \leq P_{\rm{avg}} \hspace{0.2cm}, \hspace{0.2cm} p_1^* \geq \frac{1}{\bar{D}_{\rm{max}}}  \hspace{0.2cm}, \hspace{0.2cm} \lambdapst \geq 0 \hspace{0.2cm}, \hspace{0.2cm} \lambdadst \geq 0,\\
	\label{Comp_Slackness_Power}
	&\lambdapst \cdot \left( S_1^* -P_{\rm{avg}}\right)=0,\\
	\label{Comp_Slackness_Delay}
	&\lambdadst \cdot \left( p_1^* - \frac{1}{\bar{D}_{\rm{max}}} \right)=0,
\end{align}
$i\in \script{M}$. We use $U_{i+1}^* \triangleq U_{i+1}\left(\Gammast{i+1},\Pst{i+1} \right)$ while $S_{i+1}^* \triangleq S_{i+1}\left(\Gammast{i+1},\Pst{i+1} \right)$ and $p_{i+1}^* \triangleq p_{i+1}\left(\Gammast{i+1}\right)$ for brevity in the sequel. We note that $U_{M+1} \lb \cdot , \cdot \rb = S_{M+1} \lb \cdot , \cdot \rb = p_{M+1} \lb \cdot \rb \triangleq 0$ by definition. 
We observe that these KKT equations involve the primal ($\Gammast{1}$ and $\Pst{1}$) and the dual ($\lambdapst$ and $\lambdadst$) variables. Our approach is to find a closed-form expression for the primal variables in terms of the dual variables, then propose a low-complexity algorithm to obtain the solution for the dual variables. The optimality of this approach is discussed at the end of this section (in Section \ref{Optimality_of_Approach}) where we show that, loosely speaking, the KKT equations provide a unique solution to the primal-dual variables. Hence, based on this unique solution, and on the fact that the KKT equations are necessary conditions for the optimal solution, then this solution is not only necessary but sufficient as well, and hence optimal.


\subsection{Solving for Primal Variables}
\label{Primal_Variables}
Equation (\ref{Water_Filling}) is a water-filling strategy with a slight modification due to having the condition $\gamma>\gammath{i}$. This condition comes from the sequential sensing of the channels which is absent in the classic water-filling strategy \cite{Goldsmith_Wireless_Comm}. Equation (\ref{Water_Filling}) gives a closed-form solution for $\bfP{1,1}$. On the other hand, the entries of the vector $\Gammast{1}$ are found via the set of equations (\ref{gamma_i_Equation}). Note that equation (\ref{gamma_i_Equation}) indeed forms a set of $M$ equations, each solves for one of the $\gammathst{i}$, $i\in \script{M}$. We refer to this set as the ``threshold-finding'' equations. For a given value of $i$, solving for $\gammathst{i}$ requires the knowledge of only $\gammathst{i+1}$ through $\gammathst{M}$, and does not require knowing $\gammathst{1}$ through $\gammathst{i-1}$. Thus, these $M$ equations can be solved using back-substitution starting from $\gammathst{M}$. To solve for $\gammathst{i}$, we use the fact that $\gammathst{i} \geq \lambdapst$ that is proven in the following lemma.
\begin{lma}
\label{Lma_Stochastic_Dominance}
The optimal solution of problem (\ref{Prob_Opt_Pow_Control}) satisfies $\gamma_{\rm th}^*(i) \geq \lambdapst$ $\forall i\in \script{M}$.
\end{lma}
\begin{proof}
See Appendix \ref{Apdx_Stochastic_Dominance} for proof.
\end{proof}
The intuition behind Lemma \ref{Lma_Stochastic_Dominance} is as follows. If, for some channel $i$, $\gammathst{i}<\lambdapst$ was possible, and the instantaneous gain $\gamma_i$ happened to fall in the range $[\gammathst{i},\lambdapst)$ at a given time slot, then the SU will not skip channel $i$ since $\gamma_i>\gammathst{i}$. But the power transmitted on channel $i$ is $P_{1,i}(\gamma_i)=\left(1/\lambdapst-1/\gamma_i \right)^+=0$ since $\gamma_i<\lambdapst$. This means that the SU will neither skip nor transmit on channel $i$, which does not make sense from the SU's throughput perspective. To overcome this event, the SU needs to set $\gammathst{i}$ at least as large as $\lambdapst$ so that whenever $\gamma_i<\lambdapst$, the SU skips channel $i$ rather than transmitting with zero power.

Lemma \ref{Lma_Stochastic_Dominance} allows us to remove the $( \cdot )^+$ sign in equation (\ref{gamma_i_Equation}) when solving for $\gammathst{i}$. 
Rewriting (\ref{gamma_i_Equation}) we get
\begin{align}
&\nonumber \frac{-\lambdapst}{\gammathst{i}} \exp \lb \frac{-\lambdapst}{\gammathst{i}} \rb =\\
& -\exp \lb -\frac{U_{i+1}^* - \lambdapst S_{i+1}^*-\lambdadst \cdot \left( 1-p_{i+1}^* \right)}{c_i} - 1 \rb, \hspace{0.1cm} i \in \script{M},
\label{gamma_i_2}
\end{align}
Equation (\ref{gamma_i_2}) is now on the form $W\exp(W)=c$, whose solution is $W=W_0(c)$, where $W_0(x)$ is the principle branch of the Lambert W function \cite{Lambert_W_Function} and is given by $W_0(x)=\sum_{n=1}^{\infty} \frac{\left( -n \right)^{n-1}}{n!}x^n$. The only solution to \eqref{gamma_i_2} which satisfies Lemma \ref{Lma_Stochastic_Dominance} is given for $i\in \script{M}$ by
\begin{equation}
\gammathst{i}=\frac{-\lambdapst}{W_0 \left( -\exp \left(-\frac{\left(U_{i+1}^* - \lambdapst S_{i+1}^*-\lambdadst \left( 1-p_{i+1}^* \right)\right)^+}{c_i}-1\right) \right)}.
\label{Gamma_Solution_Lambert_W}
\end{equation}

Hence, $\Gammast{1}$ and $\Pst{1}$ are found via equations (\ref{Gamma_Solution_Lambert_W}) and (\ref{Water_Filling}) respectively which are one-to-one mappings from the dual variables $(\lambdapst,\lambdadst)$. And if we had an instantaneous power constraint $P_{1,i}(\gamma)\leq \Pmax$, we could write down the Lagrangian and solve for $P_{1,i}(\gamma)$. The details are similar to the case without an instantaneous power constraint and are, thus, omitted for brevity. The reader is referred to \cite{Adaptive_Rate_Power_CR_Sonia} for a similar proof. The expression for $P_{1,i}^*(\gamma)$ is given by
	\begin{equation}
P_{1,i}^*(\gamma)=\left \{
\begin{array}{lll}
	\left(\frac{1}{\lambda_{\rm P}^*}-\frac{1}{\gamma} \right)^+ & \mbox{if } \frac{1}{\lambda_{\rm P}^*}-\frac{1}{\gamma} < P_{\rm max} \\
	P_{\rm max} & \mbox{otherwise.}
\end{array}
\right.
\end{equation}

Since the optimal primal variables are explicit functions of the optimal dual variables, once the optimal dual variables are found, the optimal primal variables are found and the optimization problem is solved. We now discuss how to solve for these dual variables.

\subsection{Solving for Dual Variables}
\label{Dual_Variables}
The optimum dual variable $\lambdapst$ must satisfy equation (\ref{Comp_Slackness_Power}). Thus if $\lambdapst>0$, then we need $S_1^* -P_{\rm{avg}}=0$. This equation can be solved using any suitable root-finding algorithm. Hence, we propose Algorithm \ref{Alg_lambda_P_lambda_D} that uses bisection \cite{Numerical_Recipes_Ch9}. In each iteration $n$, the algorithm calculates $S_1^*$ given that $\lambdap=\lambdap^{(n)}$, and given some fixed $\lambdad$, compares it to $\Pavg$ to update $\lambdap^{(n+1)}$ accordingly. The algorithm terminates when $S_1^*=\Pavg$, i.e. $\lambdap^{(n)}=\lambdapst$. The superiority of this algorithm over the exhaustive search is due to the use of the bisection algorithm that does not go over all the search space of $\lambdap$. In order for the bisection to converge, there must exist a single solution for equation $S_1^*=\Pavg$. This is proven in Theorem \ref{Thm_Unique_Solution_S}.
\begin{thm}
\label{Thm_Unique_Solution_S}
$S_1^*$ is decreasing in $\lambdapst \in [0,\infty)$ given some fixed $\lambdadst \geq 0$. Moreover, the optimal value $\lambdapst$ satisfying $S_1^*=\Pavg$ is upper bounded by $\lambdapmax \triangleq \sum_{i=1}^M \theta_i c_i /\Pavg$.
\end{thm}
\begin{proof}
See Appendix \ref{Apdx_Unique_Solution_S} for the proof.
\end{proof}
We note that Algorithm \ref{Alg_lambda_P_lambda_D} can be systematically modified to call any other root-finding algorithm (e.g. the secant algorithm \cite{Numerical_Recipes_Ch9} that converges faster than the bisection algorithm).

\begin{algorithm}
\caption{Finding $\lambdapst$ given some $\lambdad$}
\begin{algorithmic}[1]
\label{Alg_lambda_P_lambda_D}
\STATE Initialize $n \leftarrow 1$, $\lambdapmin \leftarrow 0$, $\lambdapmax \leftarrow \sum_{i=1}^M \theta_i c_i /\Pavg$, $\lambdap^{(1)} \leftarrow \lb \lambdapmin + \lambdapmax \rb /2$
\WHILE{$\vert S_1^*-P_{\rm avg} \vert > \epsilon$}
\STATE Calculate $S_1^*$ given that $\lambdapst=\lambdap^{(n)}$. Call it $S^{(n)}$.
\IF{$S^{(n)}-\Pavg>0$}
\STATE $\lambdapmin=\lambdap^{(n)}$
\ELSE
\STATE $\lambdapmax=\lambdap^{(n)}$
\ENDIF
\STATE $\lambdap^{(n+1)} \leftarrow \lb \lambdapmin + \lambdapmax \rb /2$
\STATE $n \leftarrow n+1$
\ENDWHILE
\STATE $\lambdapst \leftarrow \lambdap^{(n)}$
\end{algorithmic}
\end{algorithm}

Now, to search for $\lambdadst$, we state the following lemma.

\begin{lma}
\label{Lma_lambdadst_Bound}
The optimum value $\lambdadst$ that solves problem (\ref{Prob_Opt_Pow_Control}) satisfies $0 \leq \lambdadst <\lambdadmax$, where
\begin{equation}
\lambdadmax \triangleq \frac{c_1 \left[ \log \lb t\rb - t +1 \right] + \Utwomax}{1-\ptwomax}
\label{lambdadst_Bound}
\end{equation}
with $t \triangleq \lb\min \lb \lambdapmax,\ccdf^{-1} \lb \frac{1}{\theta_1 \Dmax}\rb\rb \rb /\lb\ccdf^{-1} \lb \frac{1}{\theta_1 \Dmax}\rb\rb$ and $\Utwomax$ is an upper bound on $U_2^*$ and is given by $\lb\int_{\lambdapmax}^\infty \log \lb \gamma/\lambdapmax\rb \pdf(\gamma) \, d\gamma \rb \lb \sum_{i=2}^M \theta_i c_i\rb$, while $\ptwomax$ is an upper bound on $p_2^*$ and is given by $\sum_{i=2}^M \prod_{j=2}^{i-1} \lb 1-\theta_j \rb \theta_i$.
\end{lma}
\begin{proof}
See Appendix \ref{Apdx_Lma_lambdadst_Bound}.
\end{proof}
Lemma \ref{Lma_lambdadst_Bound} gives an upper bound on $\lambdadst$. This bound decreases the search space of $\lambdadst$ drastically instead of searching over $\mathbb{R}$. Thus the solution of problem (\ref{Prob_Opt_Pow_Control}) can be summarized on 3 steps: 1) Fix $\lambdadst \in [0,\lambdadmax)$ and find the corresponding optimum $\lambdapst$ using Algorithm \ref{Alg_lambda_P_lambda_D}. 2) Substitute the pair $(\lambdapst,\lambdadst)$ in equations (\ref{Water_Filling}) and (\ref{Gamma_Solution_Lambert_W}) to get the power and threshold functions, then evaluate $U_1^*$ from (\ref{Reward}). 3) Repeat steps 1 and 2 for other values of $\lambdadst$ until reaching the optimum $\lambdadst$ that satisfies $p_1^*=\invDmaxline$. If there are multiple $\lambdadst$'s satisfying $p_1^*=\invDmaxline$, then the optimum one is the one that gives the highest $U_1^*$.

Although the order by which the channels are sensed is assumed fixed, the proposed algorithm can be modified to optimize over the sensing order by a relatively low complexity sorting algorithm. Particularly, the dynamic programming proposed in \cite{Sensing_Order_Poor} can be called by Algorithm \ref{Alg_lambda_P_lambda_D} to order the channels. The complexity of the sorting algorithm alone is $O(2^M)$ compared to the $O(M!)$ of the exhaustive search to sort the $M$ channels. The modification to our proposed algorithm would be in step 3 of Algorithm \ref{Alg_lambda_P_lambda_D}, where $S_1^*$ would be optimized over the number of channels (as well as $\Gammast{1}$).

\subsection{Optimality of the Proposed Solution}
\label{Optimality_of_Approach}
Since the problem in (\ref{Prob_Opt_Pow_Control}) is not proven to be convex, the KKT conditions provide only necessary conditions for optimality and need not be sufficient \cite{LinNonlinProg_Luenberger}. This means that there might exist multiple solutions (i.e. multiple solutions for the primal and/or dual variables) satisfying the KKT conditions, at least one of which is optimal. But since Theorem \ref{Thm_Unique_Solution_S} proves that there exists one unique solution to $\lambdapst$ given $\lambdadst$, then $\Gammast{1}$ and $\Pst{1}$ are unique as well (from equations (\ref{Water_Filling}) and (\ref{Gamma_Solution_Lambert_W})) given some $\lambdadst$. Hence, by sweeping $\lambdadst$ over $[0,\lambdadmax)$, we have a unique solution satisfying the KKT conditions, which means that the KKT conditions are sufficient as well and our approach is optimal for problem (\ref{Prob_Opt_Pow_Control}).

\section{Generalization of Deadline Constraints}
\label{General_Delay}
In the overlay and underlay schemes discussed thus far, we were assuming that each packet has a hard deadline of one time slot. If a packet is not delivered as soon as it arrives at the ST, then it is dropped from the system. But in real-time applications, data arrives at the ST's buffer on a frame-by-frame structure. Meaning multiple packets (constituting the same frame) arrive simultaneously rather than one at a time. A frame consists of a fixed number of packets, and each packet fits into exactly one time slot of duration $\Ts$. Each frame has its own deadline and thus, packets belonging to the same frame have the same deadline \cite{Adaptive_NC_Deadline}. This deadline represents the maximum number of time slots that the packets, belonging to the same frame, need to be transmitted by, on average.

In this section we solve this problem for the overlay scenario. The solution presented in Section \ref{Problem_Formulation} can be thought of as a special case of the problem presented in this section where the deadline was equal to $1$ time slot and each frame consists of one packet. We show that the solution presented in Section \ref{Problem_Formulation} can be used to solve this generalized problem in an offline fashion (i.e. before attempting to transmit any packet of the frame). Moreover, we propose an online update algorithm that updates the thresholds and power functions each time slot and show that this outperforms the offline solution.


\subsection{Offline Solution}
Assume that each frame consists of $K$ packets and that the entire frame has a deadline of $t_f$ time slots ($t_f>K$). If the SU does not succeed in transmitting the $K$ packets before the $t_f$ time slots, then the whole frame is considered wasted. Since instantaneous channel gains and PU's activities are independent across time slots, the probability that the SU succeeds in transmitting the frame in $t_f$ time slots or less is given by
\begin{equation}
\pframe \lb K,t_f \rb =\sum_{n=K}^{t_f} { \binom{t_f}{n} p^n \lb 1-p \rb^{t_f-n}}
\label{Binomial}
\end{equation}
where $p$ is the probability of transmitting a packet on some channel in a single time slot and is given by (\ref{Prob_recursive}) or (\ref{Prob_recursive_Soft}) if the SU's system was overlay or underlay respectively. $\pframe \lb K,t_f \rb$ represents the probability of finding $K$ or more free time slots out of a total of $t_f$ time slots.



In order to guarantee some QoS for the real-time data the SU needs to keep the probability of successful frame transmission above a minimum value denoted $\rmin$, that is $\pframe \geq \rmin$. Hence the problem becomes a throughput maximization problem subject to some average power and QoS constraints as follows
\begin{equation}
\begin{array}{ll}
\rm{maximize}& U_1(\bfGamma(1),\bfP{1,1})\\
\label{Optim_Prob_Offline}
\rm{subject \; to} &S_1(\bfGamma(1),\bfP{1,1}) \leq \Pavg\\
& \pframe (K,t_f) \geq \rmin\\
\rm{variables} & \bfGamma(1),\bfP{1,1}.
\end{array}
\end{equation}
This is the optimization problem assuming an overlay system since we used equations (\ref{Reward}) and (\ref{Average_Power}) for the throughput and power, respectively. It can also be modified systematically to the case of an underlay system. Since there exists a one-to-one mapping between $\pframe(K,t_f)$ and $p$, then there exists a value for $\Dmax$ such that the inequality $p \geq \invDmaxline$ is equivalent to the QoS inequality $\pframe \lb K,t_f \rb \geq \rmin$. That is, we can replace inequality $\pframe(K,t_f) \geq \rmin$ by $p \geq \invDmaxline$ for some $\Dmax$ that depends on $\rmin$, $K$ and $t_f$ that are known a priori. Consequently, problem (\ref{Optim_Prob_Offline}) is reduced to the simpler, yet equivalent, single-time-slot problem (\ref{Prob_Opt_Pow_Control}) and the SU can solve for $\Pst{1}$ and $\Gammast{1}$ vectors following the approach proposed in Section \ref{Problem_Formulation}. The SU solves this problem offline (i.e. before the beginning of the frame transmission) and uses this solution each time slot of the $t_f$ time slots. With this offline scheme, the SU will be able to meet the QoS and the average power constraint requirements as well as maximizing its throughput.

\subsection{Online Power-and-Threshold Adaptation}
In problem (\ref{Prob_Opt_Pow_Control}), we have seen that as $\invDmaxline$ decreases, the system becomes less stringent in terms of the delay constraint. This results in an increase in the average throughput $U_1^*$. With this in mind, let us assume, in the generalized delay model, that at time slot $1$ the SU succeeds in transmitting a packet. Thus, at time slot $2$ the SU has $K-1$ remaining packets to be transmitted in $t_f-1$ time slots. And from the properties of equation (\ref{Binomial}), $\pframe (K-1,t_f-1) > \pframe (K,t_f)$. This means that the system becomes less stringent in terms of the QoS constraint after a successful packet transmission. This advantage appears in the form of higher throughput. To see how we can make use of this advantage, define $\pframe(K(t),t_f-t+1)$ as
\begin{align}
\nonumber &\pframe \lb K(t),t_f-t+1 \rb =\\
&\sum_{n=K(t)}^{t_f-t+1} { \binom{t_f-t+1}{n} \lb p(t) \rb^n \lb 1-p(t) \rb^{t_f-t+1-n}},
\label{Binomial_t}
\end{align}
where $K(t)$ is the remaining number of packets before time slot $t \in \{1,...,t_f\}$ and $p(t)$ is the probability of successful transmission at time slot $t$. At each time slot $\tinset$, the SU modifies the QoS constraint to be $\pframe(K(t),t_f-t+1) \geq \rmin$ instead of $\pframe(K,t_f) \geq \rmin$, that was used in the offline adaptation, and solve the following problem
\begin{equation}
\begin{array}{ll}
\rm{maximize}& U_1(\bfGamma(1),\bfP{1,1})\\
\label{Optim_Prob_Online}
\rm{subject \; to} &S_1(\bfGamma(1),\bfP{1,1}) \leq \Pavg\\
& \pframe (K(t),t_f-t+1) \geq \rmin\\
\rm{variables} & \bfGamma(1),\bfP{1,1},
\end{array}
\end{equation}
to obtain the power and threshold vectors. When the delay constraint in (\ref{Optim_Prob_Online}) is replaced by its equivalent constraint $p \geq \invDmaxline$, the resulting problem can be solved using the overlay approach proposed in Section \ref{Problem_Formulation} without much increase in computational complexity since the power functions and thresholds are given in closed-form expressions. With this online adaptation, the average throughput $U_1^*$ increases while still satisfying the QoS constraint.

\section{Underlay System}
\label{Underlay}
In the overlay system, the SU tries to locate the free channels at each time slot to access these spectrum holes without interfering with the PUs. Recently, the FCC has allowed the SUs to interfere with the PU's network as long as this interference does not harm the PUs \cite{FCC_Interference_Threshold}. If the interference from the SU measured at the PU's receiver is below the tolerable level, then the interference is deemed acceptable.

In order to model the interference at the PR, we assume that the SU uses a channel sensing technique that produces the sufficient statistic $z_i$ at channel $i$ \cite{Spectrum_Sensing_Survey_Poor, Spectrum_Sensing_Survey}. The SU is assumed to know the distribution of $z_i$ given channel $i$ is free and busy, namely $f_{z \vert b} \lb z_i \vert b_i=0 \rb$ and $f_{z \vert b} \lb z_i \vert b_i=1 \rb$ respectively. For brevity, we omit the subscript $i$ from $b_i$ whenever it is clear from the context.
The value of $z_i$ indicates how confident the SU is in the presence of the PU at channel $i$. Thus the SU stops at channel $i$ according to how likely busy it is and how much data rate it will gain from this channel (i.e. according to $z_i$ and $\gamma_i$ respectively). Hence, when the SU senses channel $i$ to acquire $z_i$, the channel gain $\gamma_i$ is probed and compared to some function $\gamma_{\rm th}(i,z_i)$; if $\gamma_i\geq\gamma_{\rm th}(i,z_i)$ transmission occurs on channel $i$, otherwise, channel $i$ is skipped and $i+1$ is sensed. Potentially, $\gamma_{\rm th}(i,z_i)$ is a function in the statistic $z_i$. This means that, at channel $i$, for each possible value that $z_i$ might take, there is a corresponding threshold $\gammaz{i}$. Before formulating the problem we note that this model can capture the overlay with sensing errors model as a special case where $\fzb=(1-\pmd)\delta(z-\zb)+\pmd\delta(z-\zf)$ while $\fzf=\pfa\delta(z-\zb)+(1-\pfa)\delta(z-\zf)$, where $\pmd$ and $\pfa$ are the probabilities of missed-detection and false-alarm respectively, while $\delta(z)$ is the Dirac delta function, and $\zb$ and $\zf$ that represent the values that $z$ takes when the channel is busy and free, respectively. Hence, the interference constraint, which will soon be described, can be modified to a detection probability constraint and/or a false alarm probability constraint.

The SU's expected throughput is given by $\Usoft_1(\GammaS{1},\bfPS{1})$ which can be calculated recursively from
\begin{equation}
\begin{array}{ll}
\Usoft_i&(\GammaS{i},\bfPS{i})=\\
&c_i \int_{-\infty}^{\infty}{\int_{\gammaz{i}}^{\infty}{\log(1+\PS{i}\gamma) f_{\gamma}(\gamma)} \, d\gamma f_z(z)} \, dz +
\\ &\pskip{i} \Usoft_{i+1}(\GammaS{i+1},\bfPS{i+1}), \hspace{0.2in} i\in \script{M},
\label{Reward_Soft}
\end{array}
\end{equation}
where $\Usoft_{M+1}(\GammaS{M+1},\bfPS{M+1}) \triangleq 0$, $\GammaS{i}\triangleq [\gammaz{i},...,\gammaz{M}]^T$, $f_z(z) \triangleq \theta_i \fzf + (1-\theta_i) \fzb$ is the PDF of the random variable $z_i$ and $\pskip{i} \triangleq \int_{-\infty}^{\infty}{\int_0^{\gammaz{i}}{f_{\gamma}(\gamma)} \, d\gamma f_z(z)} \, dz$. The first term in (\ref{Reward_Soft}) is the SU's throughput at channel $i$ averaged over all realizations of $z_i$ and that of $\gamma_i \geq \gammaz{i}$. The second term is the average throughput when the SU skips channel $i$ due to its low gain. Also, let the average interference from the SU's transmitter to the PU's receiver, aggregated over all $M$ channels, be $\Isoft_1(\GammaS{1},\bfPS{1})$. This represents the total interference affecting the PU's network due to the existence of the SU. The SU is responsible for guaranteeing that this interference does not exceed a threshold $\Iavg$ dictated by the PU's network. $\Isoft_1(\GammaS{1},\bfPS{1})$ can be derived using the following recursive formula
\begin{equation}
\begin{array}{ll}
&\Isoft_i(\GammaS{i},\bfPS{i})=\\
&\left( 1- \theta_i \right) c_i \int_{-\infty}^{\infty}{\int_{\gammaz{i}}^{\infty}{\PS{i} f_{\gamma}(\gamma)} \, d\gamma f_{z \vert b} \lb z \vert b_i=1 \rb} \, dz \\
&+ \pskip{i}\Isoft_{i+1}(\GammaS{i+1},\bfPS{i+1}), \hspace{0.3in} i\in \script{M},
\label{Interference_Soft}
\end{array}
\end{equation}
where $\Isoft_{M+1}(\GammaS{M+1},\bfPS{M+1}) \triangleq 0$. This interference model is based on the assumption that the channel gain from the SU's transmitter to the PU's receiver is known at the SU's transmitter. This is the case for reciprocal channels when the PR acts as a transmitter and transmits training data to its intended primary transmitter (when it is acting as a receiver)~\cite{Wireless_Comm_Tse}. The ST overhears this training data and estimates the channel from itself to the PR. Moreover, the gain at each channel from the ST to the PR is assumed unity for presentation simplicity. This could be extended easily to the case of non-unity-gain by multiplying the first term in (\ref{Interference_Soft}) by the gain from the ST to the PR at channel $i$. Finally, $\psoft_1(\GammaS{1})$ is the probability of a successful transmission in the current time slot and can be calculated using
\begin{equation}
\begin{array}{ll}
	\psoft_i(\GammaS{i})=&\int_{-\infty}^{\infty}{\int_{\gammaz{i}}^{\infty}{f_{\gamma}(\gamma)} \, d\gamma f_z(z)} \, dz +\\
	&\pskip{i}\psoft_{i+1}(\GammaS{i+1}),
\label{Prob_recursive_Soft}
\end{array}
\end{equation}
$i\in \script{M}$, $\psoft_{M+1}(\GammaS{M+1}) \triangleq 0$. Given this background, the problem is
\begin{equation}
\begin{array}{ll}
\rm{maximize}& \Usoft_1(\GammaS{1},\bfPS{1})\\
\label{Prob_Opt_Pow_Control_Soft}
\rm{subject \; to} &\Isoft_1(\GammaS{1},\bfPS{1}) \leq \Iavg\\
&\psoft_1(\GammaS{1}) \geq \invDmax\\
\rm{variables} & \GammaS{1},\bfPS{1},
\end{array}
\end{equation}
Let $\lambdai$ and $\lambdad$ be the Lagrange multipliers associated with the interference and delay constraints of problem (\ref{Prob_Opt_Pow_Control_Soft}), respectively. 
Problem (\ref{Prob_Opt_Pow_Control_Soft}) is more challenging compared to the overlay case. This is because, unlike in (\ref{Prob_Opt_Pow_Control}), the thresholds in (\ref{Prob_Opt_Pow_Control_Soft}) are functions rather than constants. The KKT conditions for \eqref{Prob_Opt_Pow_Control_Soft} are given by
\begin{align}
\label{Water_Filling_Soft}
	&\PSst{i}=\left( \frac{1}{\lambdaist {\rm Pr}\left[b_i=1 \vert z \right]} - \frac{1}{\gamma}\right)^+ \hspace{0.5cm},\hspace{0.5cm} i\in \script{M}.\\
	\label{gammaS_i}
\nonumber &\gammazst{i}=\\
&\frac{-\lambdaist {\rm Pr}\left[b_i=1 \vert z \right]}{W_0 \left( -\exp \left(-\frac{\left(\Usoftst_{i+1} - \lambdaist \Isoftst_{i+1} - \lambdadst \left(1-\psoftst_{i+1} \right) \right)^+}{c_i}-1\right) \right)}, \hspace{0.25cm} i\in \script{M},
\end{align}
in addition to the primal feasibility, dual feasibility and the complementary slackness equations given in \eqref{Primal_Dual_Feasible}, \eqref{Comp_Slackness_Power} and \eqref{Comp_Slackness_Delay}, where $\Usoftst_{i+1} \triangleq \Usoft_1  \left( \Gammasst{1},\PSst{1} \right)$, $\Isoftst_{i+1} \triangleq \Isoft_1  \left( \Gammasst{1},\PSst{1} \right)$ and $\psoftst_{i+1} \triangleq \psoft_1  \lb \Gammasst{1}\rb$ while ${\rm Pr}\left[b_i=1 \vert z \right]$ is the conditional probability that channel $i$ is busy given $z_i$ and is given by
\begin{equation}
\label{Cond_Prob_i_given_z_i}
{\rm Pr}\left[b_i=1 \vert z \right]=\frac{\left( 1-\theta_i \right) f_{z \vert b} \lb z \vert b_i=1 \rb}{f_z \left( z \right)}.
\end{equation}
Note that $\PSst{i}$ is increasing in $\gamma$ and is upper bounded by the term $1 / \lb \lambdaist {\rm Pr}\left[b_i=1 \vert z \right] \rb$. Hence, as ${\rm Pr}\left[b_i=1 \vert z \right]$ increases, the SU's maximum power becomes more limited, i.e. the maximum power decreases. This is because the PU is more likely to be occupying channel $i$. Thus the power transmitted from the SU should decrease in order to protect the PU.


Algorithm \ref{Alg_lambda_P_lambda_D} can also be used to find $\lambdaist$. Only a single modification is required in the algorithm which is that $S_1^*$ would be replaced by $\Isoftst_1$. Thus the solution of problem (\ref{Prob_Opt_Pow_Control_Soft}) can be summarized on 3 steps: 1) Fix $\lambdadst \in \mathbb{R}^+$ and find the corresponding optimum $\lambdaist$ using the modified version of Algorithm \ref{Alg_lambda_P_lambda_D}. 2) Substitute the pair $(\lambdaist,\lambdadst)$ in equations (\ref{Water_Filling_Soft}) and (\ref{gammaS_i}) to get the power and threshold functions, then evaluate $\Usoftst_1$ from (\ref{Reward_Soft}). 3) Repeat steps 1 and 2 for other values of $\lambdadst$ until reaching the optimum $\lambdadst$ that satisfies $\psoftst_1=\invDmaxline$ and if there are multiple $\lambdadst$'s satisfying $p_1^*=\invDmaxline$, then the optimum one is the one that gives the highest $\Usoftst_1$. This approach yields the optimal solution. Next, Theorem \ref{Thm_Unique_Solution_I} asserts the monotonicity of $\Isoftst_1$ in $\lambdaist$ which allows using the bisection to find $\lambdaist$ given some fixed $\lambdadst$.
\begin{thm}
\label{Thm_Unique_Solution_I}
$\Isoftst_1$ is decreasing in $\lambdaist \in [0,\infty)$ given some fixed $\lambdadst \geq 0$.
\end{thm}
\begin{proof}
We differentiate $\Isoftst_1$ with respect to $\lambdaist$ given that $\PSst{i}$ and $\gammazst{i}$ are given by equations \eqref{Water_Filling_Soft} and \eqref{gammaS_i} respectively, then show that this derivative is negative. The proof is omitted since it follows the same lines of Theorem \ref{Thm_Unique_Solution_S}. 
\end{proof}

Although the interference power constraint is sufficient for the problem to prevent the power functions from going to infinity, in some applications one may have an additional power constraint on the SUs. Hence, problem (\ref{Prob_Opt_Pow_Control_Soft}) can be modified to introduce an average power constraint that is given by $\Ssoft_1(\GammaS{1},\bfPS{1}) \leq \Pavg$ where
\begin{equation}
\begin{array}{ll}
\Ssoft_i(\GammaS{i},\bfPS{i})&=c_i \int_{-\infty}^{\infty}{\int_{\gammaz{i}}^{\infty}{\PS{i} f_{\gamma}(\gamma)} \, d\gamma f_z(z)} \, dz \\
&+ \pskip{i}\Ssoft_{i+1}(\GammaS{i+1},\bfPS{i+1}).
\label{Avg_Pow_Soft}
\end{array}
\end{equation}
It can be easily shown that the solution to the modified problem is similar to that presented in equations (\ref{Water_Filling_Soft}) and (\ref{gammaS_i}) which is
\begin{align}
\label{Water_Filling_Soft_Avg_Pow}
	&\PSst{i}=\left( \frac{1}{\lambdapst +\lambdaist {\rm Pr}\left[b_i=1 \vert z \right]} - \frac{1}{\gamma}\right)^+ \hspace{0.5cm},\\
	\label{gammaS_i_Avg_Pow}
\nonumber &\gammazst{i}=\\
&\frac{-\lb \lambdapst + \lambdaist {\rm Pr}\left[b_i=1 \vert z \right] \rb}{W_0 \left( -\exp \left(-\frac{\left(\Usoftst_{i+1} - \lambdaist \Isoftst_{i+1} -\lambdapst \Ssoftst_{i+1} - \lambdadst \left(1-\psoftst_{i+1} \right) \right)^+}{c_i}-1\right) \right)},
\end{align}
$\forall i\in \script{M}$ where $\Ssoft_i^* \triangleq \Ssoft_i(\Gammasst{i},\PSst{i})$. This solution is more general since it takes into account both the average interference and the average power constraint besides the delay constraint. Moreover, it allows for the case where the power constraint is inactive which happens if the PU has a strict average interference constraint. In this case the optimum solution would result in $\lambdapst=0$ making equations (\ref{Water_Filling_Soft_Avg_Pow}) and (\ref{gammaS_i_Avg_Pow}) identical to equations (\ref{Water_Filling_Soft}) and (\ref{gammaS_i}) respectively.

\section{Multiple Secondary Users}
\label{Multi_SU}
In this section, we show how our single SU framework can be extended to multiple SUs in a multiuser diversity framework without increase in the complexity of the algorithm. We will show that when the number of SUs is high, with slight modifications to the definitions of the throughput, power and probability of success, the single SU optimization problem in \eqref{Prob_Opt_Pow_Control} (or \eqref{Prob_Opt_Pow_Control_Soft}) can capture the multi-SU scenario. Moreover, the proposed solution for the overlay model still works for the multi-SU scenario. Finally, at the end of this section, we show that the proposed algorithm provides a throughput-optimal and delay-optimal solution with even a lower complexity for finding the thresholds compared to the single SU case, if the number of SUs is large.

Consider a CR network with $L$ SUs associated with a centralized secondary base station (BS) in a downlink overlay scenario. Before describing the system model, we would like to note that when we say that channel $i$ will be sensed, this means that each user will independently sense channel $i$ and feedback the sensing outcome to the BS to make a global decision. Although we neglect sensing errors in this section, the analysis will work similarly in the presence of sensing errors by using the underlay model. At the beginning of each time slot the $L$ SUs sense channel 1. If it is free, each SU observes it free with no errors and probes the instantaneous channel gain and feeds it back to the BS. The BS compares the maximum received channel gain among the $L$ received channel gains to $\gammath{1}$. Channel 1 is assigned to the user having the maximum channel gain if this maximum gain is higher than $\gammath{1}$, while the remaining $L-1$ users continue to sense channel 2. On the other hand if the maximum channel gain is less than $\gammath{1}$, channel 1 is skipped and channel 2 is sensed by all $L$ users. 
Unlike the case in the single SU scenario where only a single channel is claimed per time slot, in this multi-SU system, the BS can allocate more than one channel in one time slot such that each SU is not allocated more than one channel and each channel is not allocated to more than one SU. 
Based on this scheme, the expected per-SU throughput $U_1^L$ is calculated from
\begin{align}
\nonumber U_i^l=&\frac{\theta_i c_i}{l} \int_{\gammath{i}}^{\infty}{\log \lb 1+P_{1,i}(\gamma)\gamma \rb f_l(\gamma)} \, d\gamma+\\
 &\theta_i \bar{F}_l(\gammath{i})\left(1-\frac{1}{l} \right) U_{i+1}^{l-1} + \left(1-\theta_i \bar{F}_l(\gammath{i}) \right) U_{i+1}^l
\label{Reward_l_SUs}
\end{align}
$i\in \script{M}$ and $l \in\{L-i+1,...,L\}$ with initialization $U_{M+1}^l=0$. Here $f_l(\gamma)$ represents the density of the maximum gain among $l$ i.i.d. users' gains, while $\bar{F}_l(\gamma)$ is its complementary cumulative distribution function. We study the case where $L \gg M$, thus when a channel is allocated to a user we can assume that the remaining number of users is still $L$. Thus we approximate $l$ with $L$ $\forall l \in\{L-i,...,L\}$ and $\forall i \in \script{M}$. 
Similar to the the throughput derived in \eqref{Reward_l_SUs}, we could write the exact expressions for the per-SU average power and per-SU probability of transmission. And since $L \gg M$, we can approximate $S_i^l$ with $S_i^L$ and $p_i^l$ with $p_i^L$, $\forall l \in\{L-i+1,...,L\}$ and $\forall i \in \script{M}$. The per-SU expected throughput $U_1^L$, the average power $S_1^L$ and the probability of transmission $p_1^L$ can be derived from
\begin{align}
\nonumber U_i^L(\bfGamma(i),\bfP{1,i})=&\frac{\theta_i c_i}{L} \int_{\gammath{i}}^{\infty}{\log \lb 1+P_{1,i}(\gamma)\gamma \rb f_{L}(\gamma)} \, d\gamma +\\
&\left[1-\frac{\theta_i \bar{F}_{L}(\gammath{i})}{L} \right] U_{i+1}^L \lb \bfGamma(i+1),\bfP{i+1} \rb
\label{Reward_L_SU}\\
\nonumber S_i^L(\bfGamma(i),\bfP{1,i})=&\frac{\theta_i c_i}{L} \int_{\gammath{i}}^\infty{P_{1,i}(\gamma) f_{L}(\gamma) \,d\gamma}+ \\
&\left[1-\frac{\theta_i \bar{F}_{L}(\gammath{i})}{L} \right]S_{i+1}^L(\bfGamma(i+1),\bfP{i+1}),
\label{Average_Power_L_SU}\\
\nonumber p_i^L(\bfGamma(i))=&\frac{\theta_i}{L} \bar{F}_{L}(\gammath{i})+ \\
&\left[1-\frac{\theta_i \bar{F}_{L}(\gammath{i})}{L} \right]p_{i+1}^L(\bfGamma(i+1)),
\label{Prob_Trans_L_SU}
\end{align}
$i\in \script{M}$, respectively, with  $U_{M+1}^L=S_{M+1}^L=p_{M+1}^L=0$. To formulate the multi-SU optimization problem, we replace $U_1$, $S_1$ and $p_1$  in (\ref{Prob_Opt_Pow_Control}) with $U_1^L$, $S_1^L$ and $p_1^L$ derived in equations (\ref{Reward_L_SU}), (\ref{Average_Power_L_SU}) and (\ref{Prob_Trans_L_SU}), respectively. Taking the Lagrangian and following the same procedure as in Section \ref{Problem_Formulation}, we reach at the solution for $P_{1,i}^*$ and $\gammathst{i}$ as given by equations (\ref{Water_Filling}) and (\ref{Gamma_Solution_Lambert_W}) respectively. Hence, equations (\ref{Water_Filling}) and (\ref{Gamma_Solution_Lambert_W}) represent the optimal solution for the multi-SU scenario. The details are omitted since they follow those of the single SU case discussed in Section \ref{Problem_Formulation}.

Next we show that this solution is optimal with respect to the delay as well as the throughput when $L$ is large. We show this by studying the system after ignoring the delay constraint and show that the resulting solution of this system (which is what we refer to as the unconstrained problem) is a delay optimal one as well. The solution of the unconstrained problem is given by setting $\lambdadst=0$ in \eqref{Gamma_Solution_Lambert_W} arriving at
\begin{equation}
\gammathst{i}\vert_{\lambdadst=0}=\frac{-\lambdapst}{W_0 \left( -\exp \left(-\frac{\left(U_{i+1}^{L*} - \lambdapst S_{i+1}^{L*}\right)^+}{c_i}-1\right) \right)}.
\label{Gamma_Solution_Lambert_W_Thr_Opt}
\end{equation}
$\forall i \in \script{M}$. As the number of SUs increases, the per-user expected throughput $U_1^L$ decreases since these users share the total throughput. Moreover, $U_i^L$ decreases as well $\forall i \in \script{M}$ decreasing the value of $\gammathst{i}$ (from equation \eqref{Gamma_Solution_Lambert_W_Thr_Opt} until reaching its minimum (i.e. $\gammathst{i}=\lambdapst$) (the right-hand-side of \eqref{Gamma_Solution_Lambert_W_Thr_Opt} is minimum when its denominator is as much negative as possible. That is, when $W_0(x)=-1$ since $W_0(x)\geq -1$, $\forall x\in \mathbb{R}$) as $L \rightarrow \infty$. From \eqref{Prob_Trans_L_SU}, it can be easily shown that $p_1^L(\bfGamma(1))$ is monotonically decreasing in $\gammath{i}$ $\forall i\in \script{M}$. Thus the minimum possible average delay (corresponding to the maximum $p_1^L(\bfGamma(1))$) occurs when $\gammath{i}$ is at its minimum possible value for all $i \in \script{M}$. Consequently, having $\gammathst{i}=\lambdapst$ means that the system is at the optimum delay point. That is, the unconstrained problem cannot achieve any smaller delay with an additional delay constraint. Hence, the multi-SU problem, that is formulated by adding a delay constraint to the unconstrained problem, achieves the optimum delay performance when $L$ is asymptotically large.

Recall that the overall complexity of solution for the single SU case is due to three factors: 1) evaluating the Lambert W function in Algorithm \ref{Alg_lambda_P_lambda_D}, 2) the bisection algorithm in Algorithm \ref{Alg_lambda_P_lambda_D} and 3) the search over $\lambdad$. On the other hand, the complexity of solution for the multi-SU case decreases asymptotically (as $L\rightarrow \infty$). This is because of two reasons: 1) When $L \gg M$, $\gammathst{i} \rightarrow \lambdapst \forall i \in \script{M}$. Which means that we will not have to evaluate the Lambert W function in \eqref{Gamma_Solution_Lambert_W} but instead we set $\gammathst{i}=\lambdapst$, since $L\gg M$. 2) When $\gammathst{i} = \lambdapst$ there will be no need to find $\lambdadst$ since the delay is minimum (we recall that in the single SU case, we need to calculate $\lambdadst$ to substitute it in \eqref{Gamma_Solution_Lambert_W} to evaluate $\gammathst{i}$, but in the multi-SU case $\gammathst{i} = \lambdapst$).

\section{Numerical Results}
\label{Results_Seq}
We show the performance of the proposed solution for the overlay and underlay scenarios. The slot duration is taken to be unity (i.e. all time measurements are taken relative to the time slot duration), while $\tau=0.05\Ts$. Here, we use $M=10$ channels that suffer i.i.d. Rayleigh fading. The availability probability is taken as $\theta_i=0.05i$ throughout the simulations. The power gain $\gamma$ is exponentially distributed as $f_{\gamma} \left( \gamma \right)=\exp \left( \gamma / \bar{\gamma}\right) / \bar{\gamma}$ where $\bar{\gamma}$ is the average channel gain and is set to be $1$ unless otherwise specified.


Fig. \ref{Overlay_Fig_2_Unconst_Const_NoOSR_M10} plots the expected throughput $U_1^*$ for the overlay scenario after solving problem \eqref{Prob_Opt_Pow_Control}. $U_1^*$ is plotted using equation \eqref{Reward} that represents the average number of bits transmitted divided by the average time required to transmit those bits, taking into account the time wasted due to the blocked time slots. We plot $U_1^*$ with $\Dmax=1.02 \Ts$ and with $\Dmax=\infty$ (i.e. neglecting the delay constraint). We refer to the former problem as constrained problem, while to the latter as unconstrained problem. We also compare the performance to the non optimum stopping rule case (No-OSR) where the SU transmits at the first available channel. We expect the No-OSR case to have the best delay and the worst throughput performances. We can see that the unconstrained problem has the best throughput amongst all constrained problems.

Examining the constrained problem for different sensing orders of the channels, we observe that when the channels are sorted in an ascending order of $\theta_i$, the throughput is higher. This is because a channel $i$ has a higher chance of being skipped if put at the beginning of the order compared to the case if put at the end of the order. This is a property of the problem no matter how the channels are ordered, i.e. this property holds even if all channels have equal values of $\theta_i$. Hence, it is more favorable to put the high quality channels at the end of the sensing order so that they are not put in a position of being frequently skipped. However, this is not necessarily optimum order, which is out of the scope of this work and is left as a future work for this delay-constrained optimization problem.

We also plot the expected throughput of a simple stopping rule that we call $K$-out-of-$M$ scheme, where we choose the highest $K$ channels in availability probability and ignore the remaining channels as if they do not exist in the system. The SU senses those $K$ channels sequentially, probes the gain of each free channel, if any, and transmits on the channel with the highest gain. This scheme has a constant fraction $K\tau/\Ts$ of time wasted each slot. Yet it has the advantage of choosing the best channel among multiple available ones. In Fig. \ref{Overlay_Fig_2_Unconst_Const_NoOSR_M10} we can see that the degradation of the throughput when $K=5$ compared to the optimal stopping rule scheme. The reason is two-fold: 1) Due to the constant wasted fraction of time, and 2) Ignoring the remaining channels that could potentially be free with a high gain if they were considered as opposed to the constrained problem.

The delay is shown in Fig. \ref{Overlay_Fig_3_Unconst_Const_NoOSR_M10} for both the constrained and the unconstrained problems. We see that the unconstrained problem suffers around $6\%$ increase in the delay, at $\Pavg=10$, compared to the constrained one.

\begin{figure}
	\centering
	\includegraphics[width=1\columnwidth]{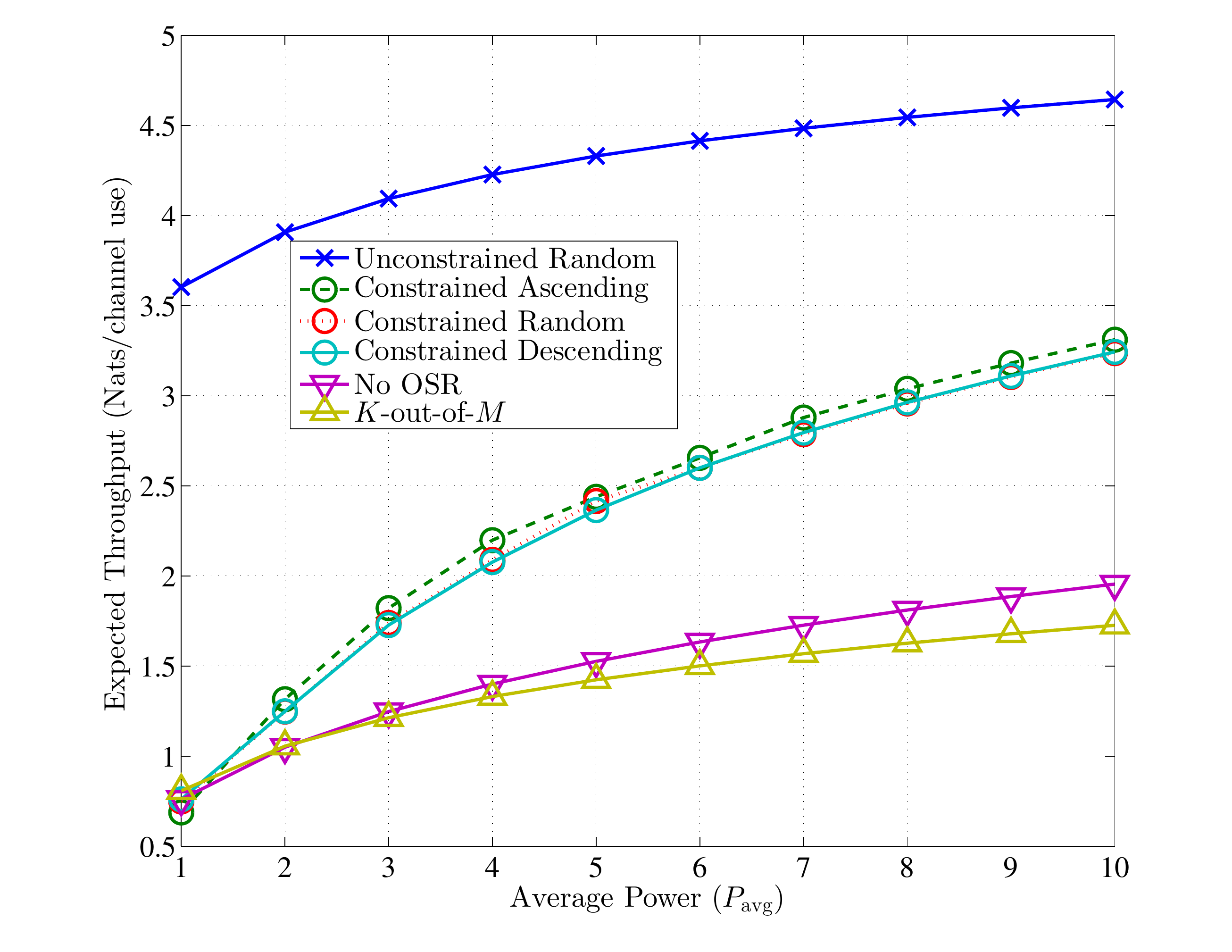}
		\caption{The expected throughput for the overlay scenario for four cases: 1) {\bf Proposed constrained problem}: with average delay constraint for three channel ordering possibilities (ascending ordering of channel availability probabilities, descending ordering, and random ordering), 2) {\bf Unconstrained problem} that ignores the delay constraint, 3) {\bf No optimum stopping rule (No-OSR)} where the SU transmits at the first free channel  and 4) {\bf $K$-out-of-$M$ scheme} where the SU assumes the system has only $K=5$ channels and ignores the remaining $M-K$ channels.}
	\label{Overlay_Fig_2_Unconst_Const_NoOSR_M10}
\end{figure}

\begin{figure}
	\centering
	\includegraphics[width=1\columnwidth]{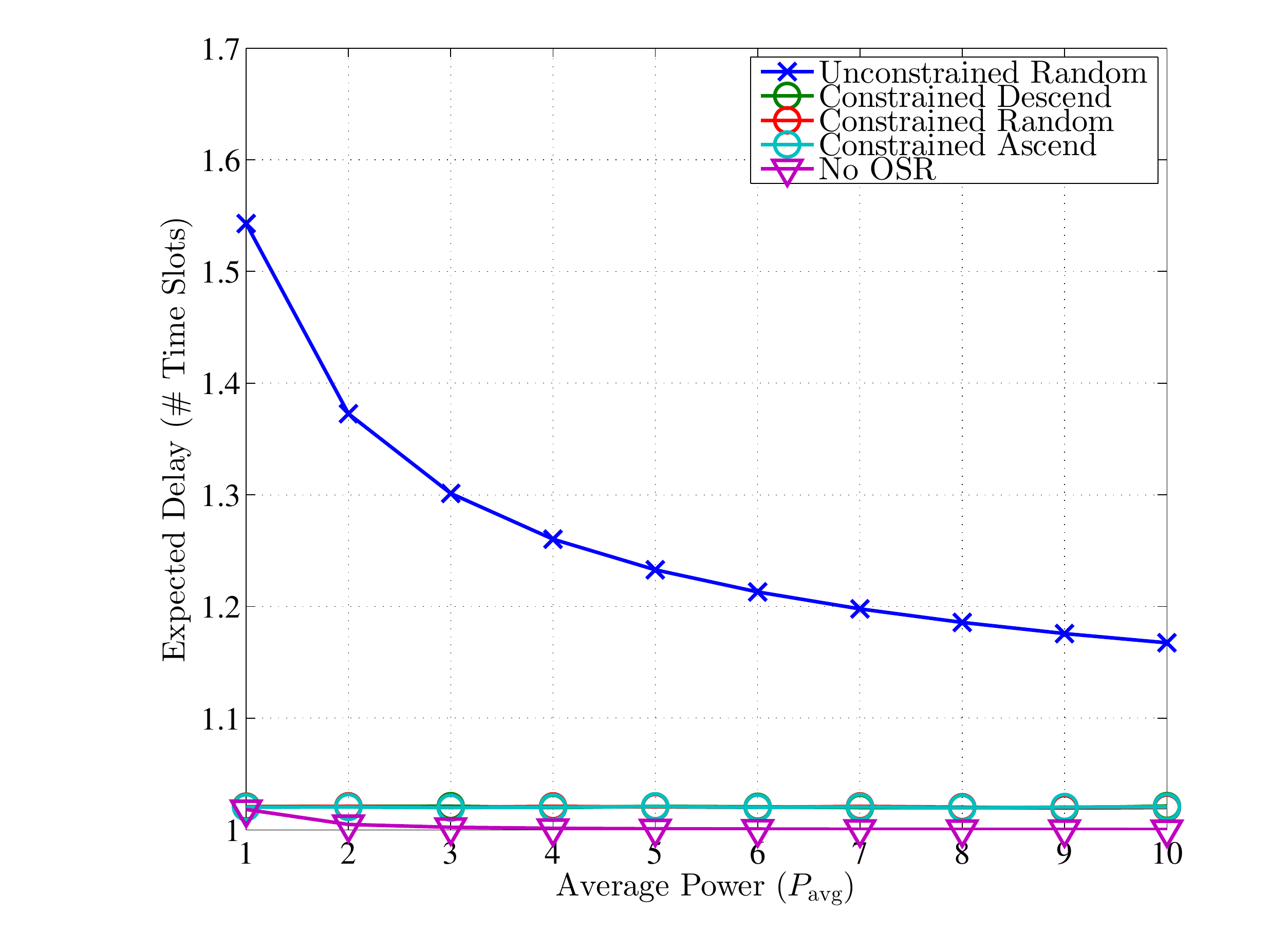}
		\caption{The expected delay for the overlay scenario for problem (\ref{Prob_Opt_Pow_Control}). The unconstrained problem can suffer arbitrary high delay. The constrained problem has a guaranteed average delay for all ordering strategies. The No-OSR scenario, on the other hand, has the best delay performance since the SU uses the first free channel.}
	\label{Overlay_Fig_3_Unconst_Const_NoOSR_M10}
\end{figure}

Studying the system performance under low average channel gain is essential. A low average channel gain represents a SU's channel being in a permanent deep fade or if there is a relatively high interference level at the secondary receiver. Fig. \ref{Overlay_Fig4_Optimum_Threshold} shows $\gammathst{i}$ versus the $\bar{\gamma}$. At low $\bar{\gamma}$, the throughput is expected to be small, hence $\gammathst{i}$ is close to its minimum value $\lambdapst$ so that even if $\gamma_i$ is relatively small, $i$ should not be skipped. In other words, at low average channel gain, the expected throughput is small, thus a relatively low instantaneous gain will be satisfactory for stopping at channel $i$. While when the average channel gain increases, $\gammathst{i}$ should increase so that only high instantaneous gains should lead to stopping at channel $i$. In both cases, high and low $\bar{\gamma}$ there still is a trade-off between choosing a high versus a low value of $\gammathst{i}$.
\begin{figure}[htbp]
	\centering
		\includegraphics[width=1\columnwidth]{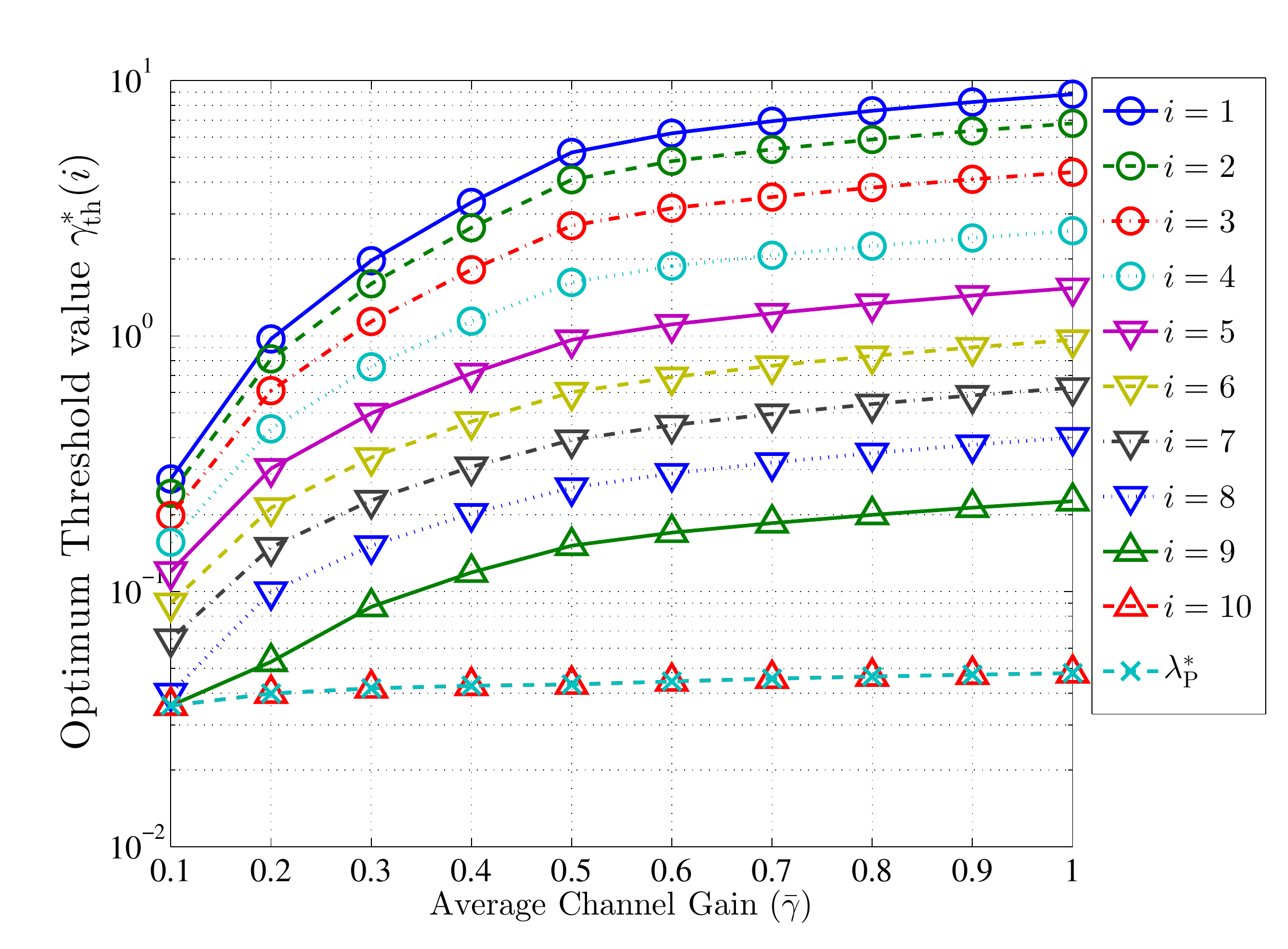}
	\caption{The gap between the optimum threshold $\gammathst{i}$ and its minimum value $\lambdapst$ increases as the average gain increases. This is because as $\bar{\gamma}$ increases, $U_{i+1}$ increases as well. Hence $\gammathst{i}$ increases so that only sufficiently high instantaneous gains should lead to stopping at channel $i$.}
	\label{Overlay_Fig4_Optimum_Threshold}
\end{figure}

The sensing channel (i.e. the channel between the PT and ST over which the ST overhears the PT activity) is modeled as AWGN with unit variance.
 The distributions of the energy detector output $z$ (average energy of N samples sampled from this sensing channel) under the free and busy hypotheses are the Chi-square and a Noncentral Chi-square distributions given by
\begin{align}
\label{X2_Distribution}
&\fzf =\left( \frac{N}{\sigma^2} \right)^N \frac{z^{N-1}}{\left( N-1 \right)!} \exp \left( \frac{-N z}{\sigma^2} \right),\\
\label{NCX2_Distribution}
\nonumber &\fzb =\\
&\left( \frac{N}{\sigma^2} \right) \left( \frac{z}{\cal E} \right)^{\frac{N-1}{2}} \exp \left( \frac{-N \left( z+ {\cal E} \right)}{\sigma^2} \right) \Ibessel{N-1} \left( \frac{2N \sqrt{{\cal E} z}}{\sigma^2} \right),
\end{align}
where $\sigma^2$, which is set to $1$, is the variance of the Gaussian noise of the energy detector, $\cal E$ is the amount of energy received by the ST due to the activity of the PT and is taken as ${\cal E}=2\sigma^2$ throughout the simulations, while $\Ibessel{i}(x)$ is the modified Bessel function of the first kind and $i$th order, and $N=10$.

The main problem we are addressing in this chapter is the optimal stopping rule that dictates for the SU when to stop sensing and start transmitting. As we have seen, this is identified by the threshold vector $\Gammasst{1}$. If the SU does not consider the optimal stopping rule problem and rather transmits as soon as it detects a free channel, then it will be wasting future opportunities of possibly higher throughput. Hence, we expect a degradation in the throughput. We plot the two scenarios in Fig. \ref{Underlay_OSR_vs_noOSR_vs_Iavg_M10} for the underlay system with no delay constraint.

Throughout this chapter, we use bold fonts for vectors and asterisk to denote that $x^*$ is the optimal value of $x$; all logarithms are natural, while the expected value operator is denoted $\mathbb{E}[\cdot]$ and is taken with respect to all the random variables in its argument. Finally, we use $(x)^+ \triangleq  {\rm{max}}(x,0)$ and $\mathbb{R}$ to denote the set of the real numbers.

\begin{figure}[htbp]
	\centering
		\includegraphics[width=1\columnwidth]{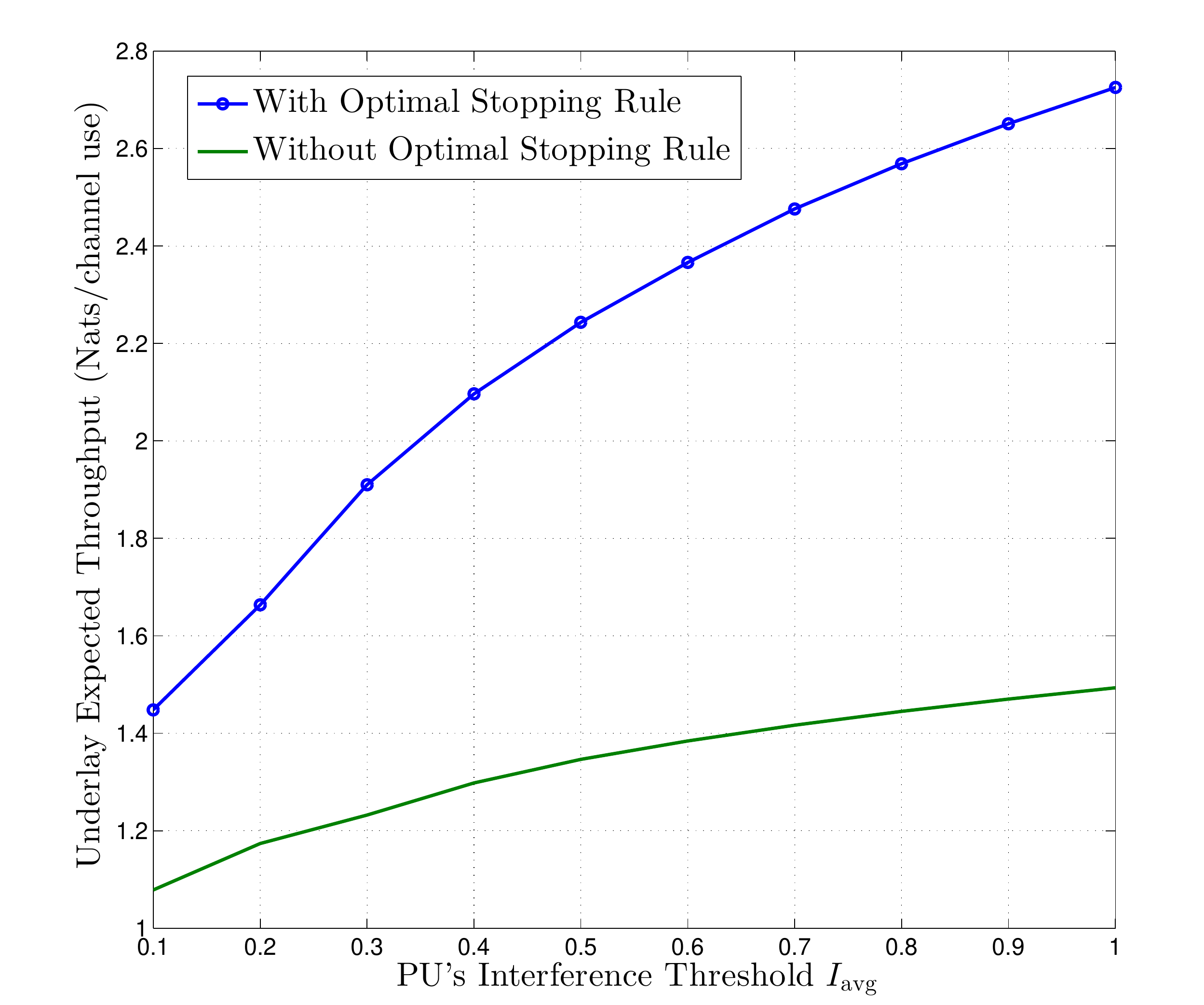}
	\caption{The underlay expected throughput versus the average interference threshold $I_{\rm avg}$. Two scenarios are shown: with and without the optimal stopping rule formulation. In the latter, the SU transmits as soon as a channel is found free.}
	\label{Underlay_OSR_vs_noOSR_vs_Iavg_M10}
\end{figure}

For the multiple SU scenario, the numerical analysis were run for the case of $L=30$ SUs while $M=10$ channels. We assumed the fading channels are i.i.d. among users and among frequency channels. Each channel is exponentially distributed with unity average channel gain. And since $L$ is large, the distribution of the maximum gain among $L$ random gains converges in distribution to the Gumbel distribution \cite{zhang2009asymptotic} having a cumulative distribution function of $\exp \left( -\exp \left( -\gamma/\bar{\gamma}\right)\right)$. The per-user throughput $U_1^{L*}$ is plotted in Fig. \ref{Throughput_MultiSU} where the throughput of the delay-constrained and of the unconstrained optimization problems coincide. This is because when $L\gg M$, the solution of the unconstrained problem is delay optimal as well. Hence, adding a delay constraint does not sacrifice the throughput, when $L$ is large. Moreover, the delay performance shown in Fig. \ref{Delay_MultiSU} shows that the delay does not change with and without considering the average delay constraint since the system is delay- (and throughput-) optimal already.

\begin{figure}%
	\centering
	\includegraphics[width=1\columnwidth]{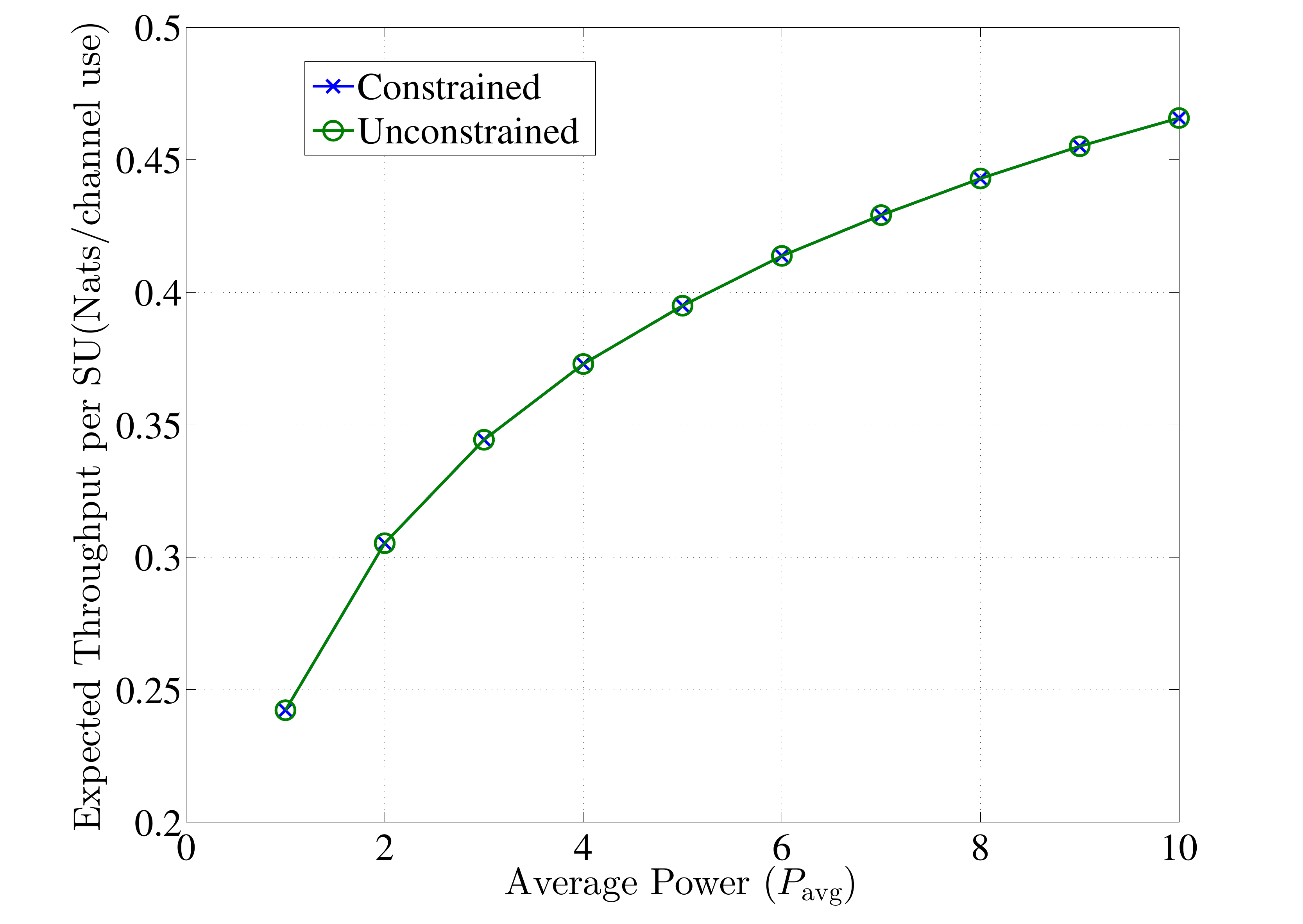}
	\caption{Per user throughput of the system at $L=30$ SUs. The throughput of the constrained and unconstrained problem coincide since the system is throughput (and delay) optimal.}
	\label{Throughput_MultiSU}
	\end{figure}
	\begin{figure}
	\centering
	\includegraphics[width=1\columnwidth]{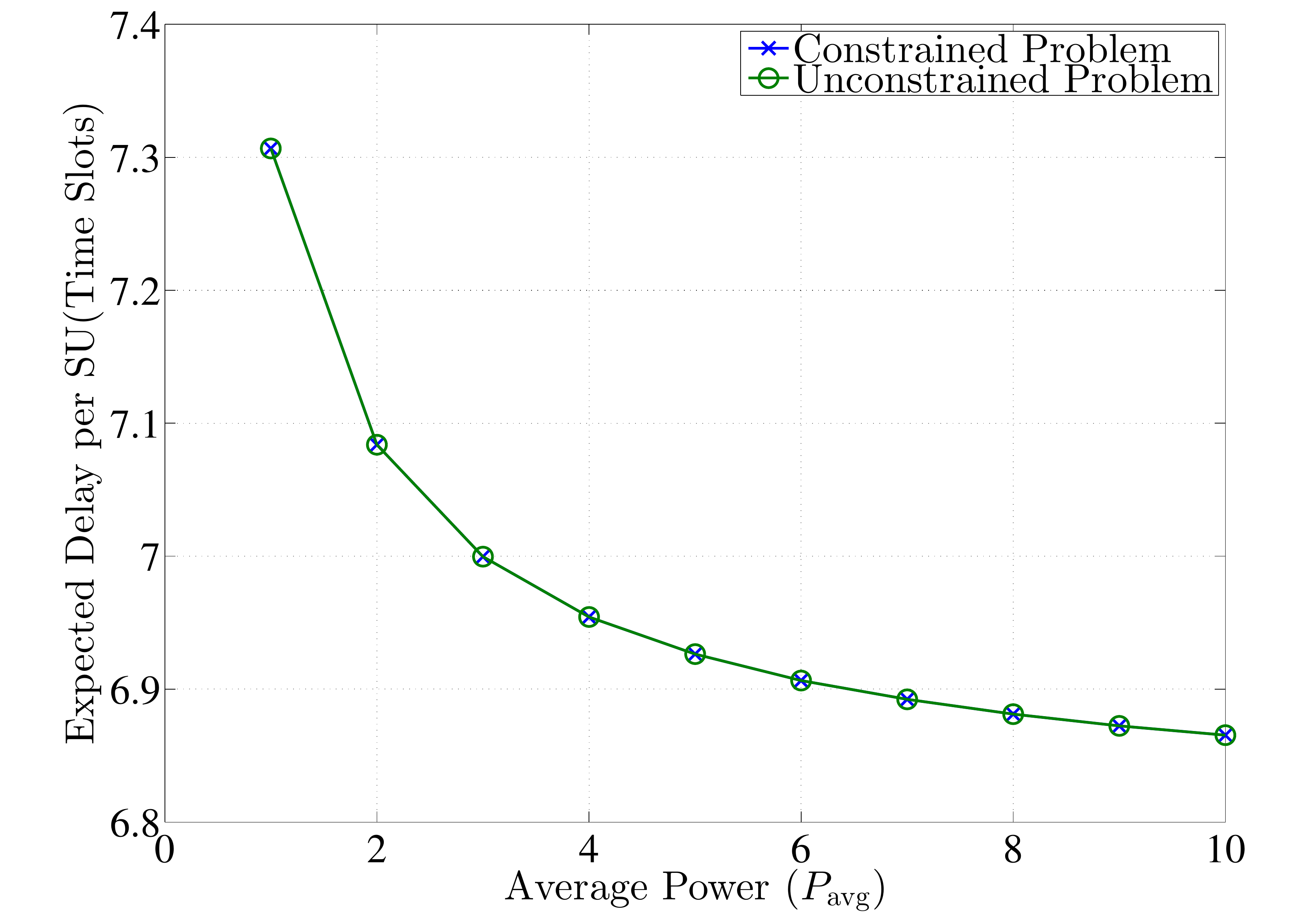}
	\caption{The average delay seen by each user in the system at $L=30$ SUs. The delay of the constrained and unconstrained problems coincide since the system is delay (and throughput) optimal.}
	\label{Delay_MultiSU}
\end{figure}

We have simulated the system for the online algorithm of Section \ref{General_Delay} for $K(1)=2$ packets and $t_f=4$ time slots. We simulated the system at $\rmin=0.95$ which means that the QoS of the SU requires that at least $95\%$ of the frames to be successfully transmitted. Fig. \ref{General_Delay_Overlay_OptPow} shows the improvement in the throughput of the online over the offline adaptation. This is because the SU adapts the power and thresholds at each time slot to allocate the remaining resources (i.e. remaining time slots) according to the remaining number of packets and the desired QoS. This comes at the expense of re-solving the problem at each time slot (i.e. $t_f$ times more).

\begin{figure}%
\centering
\includegraphics[width=1\columnwidth]{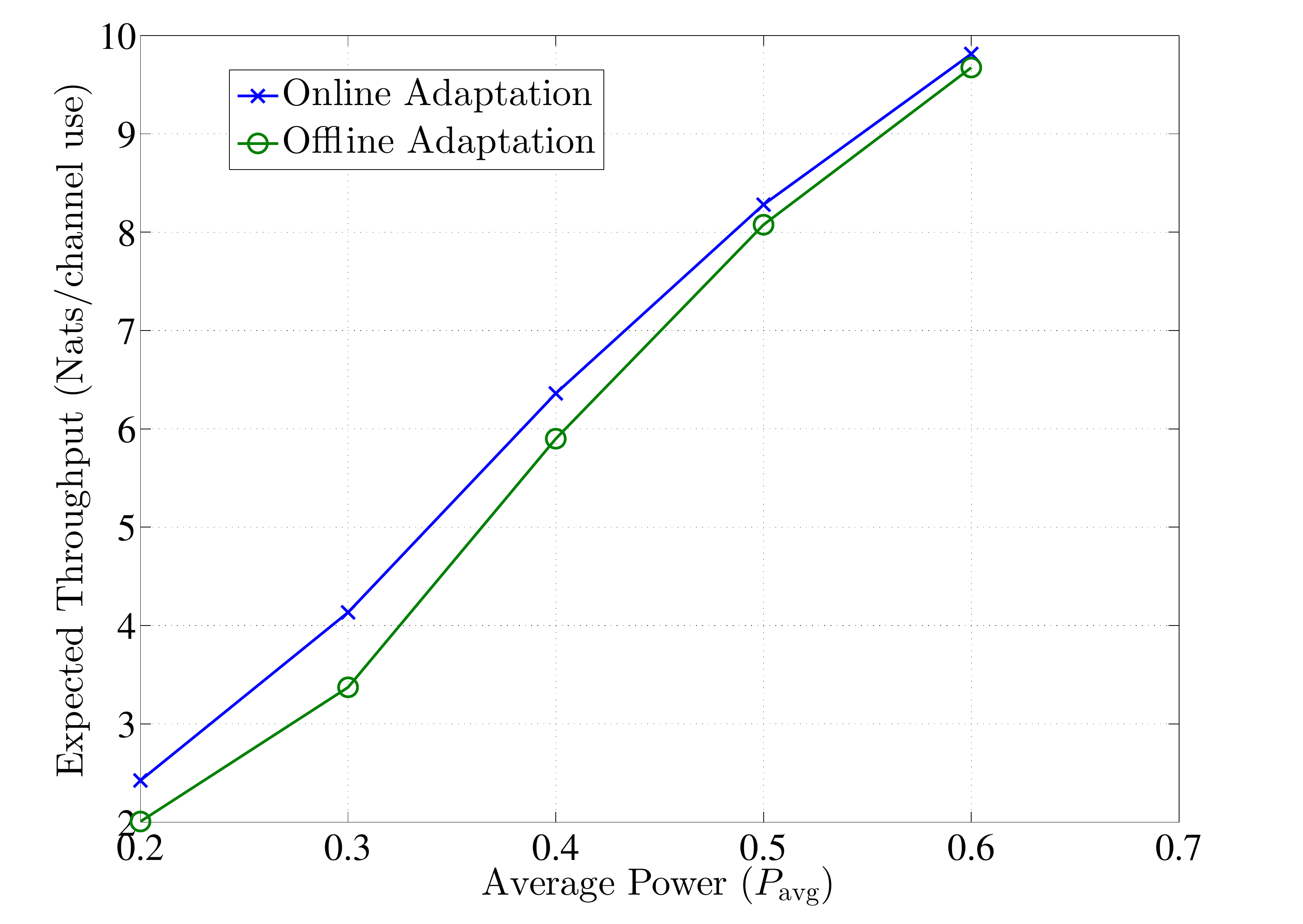}
\caption{The performance of the online adaptation algorithm for the general delay case.}
\label{General_Delay_Overlay_OptPow}
\end{figure}

\chapter{Delay Due to Queue-Waiting Time}
\label{Ch3_Queues}
In this chapter we study the delay resulting from the service time as well as that from queue-waiting time. The service time is affected by the power transmitted by the SU, while the queue-waiting time is affected by the scheduling algorithm. We propose a delay-optimal scheduling-and-power-allocation algorithm that guarantees bounds on the SUs' delays while causing an acceptable interference to the PUs. This algorithm is useful to provide fair delay guarantees to the SUs when delay fairness cannot be achieved due to the heterogeneity in SUs' channel statistics.

\section{Network Model}
\label{Model}
We assume a CR system consisting of a single secondary base station (BS) serving $N$ secondary users (SUs) indexed by the set $\script{N}\triangleq \{1,\cdots N\}$ (Fig. \ref{Cell_Fig}). We are considering the uplink phase where each SU has its own queue buffer for packets that need to be sent to the BS. The SUs share a single frequency channel with a single PU that has licensed access to this channel. The CR system operates in an underlay fashion where the PU is using the channel continuously at all times. SUs are allowed to transmit as long as they do not cause harmful interference to the PU. In this work, we consider two different scenarios where the interference can be considered as harmful. The first is an instantaneous interference constraint where the interference received by the PU at any given slot should not exceed a prespecified threshold $\Iinst$, while the second is an average interference constraint where the interference received by the PU averaged over a large duration of time should not exceed a prespecified threshold $\Iavg$. Moreover, in order for the secondary BS to be able to decode the received signal, no more than one SU at a time slot is to be assigned the channel for transmission.

\begin{figure}%
\centering
\includegraphics[width=0.6\columnwidth]{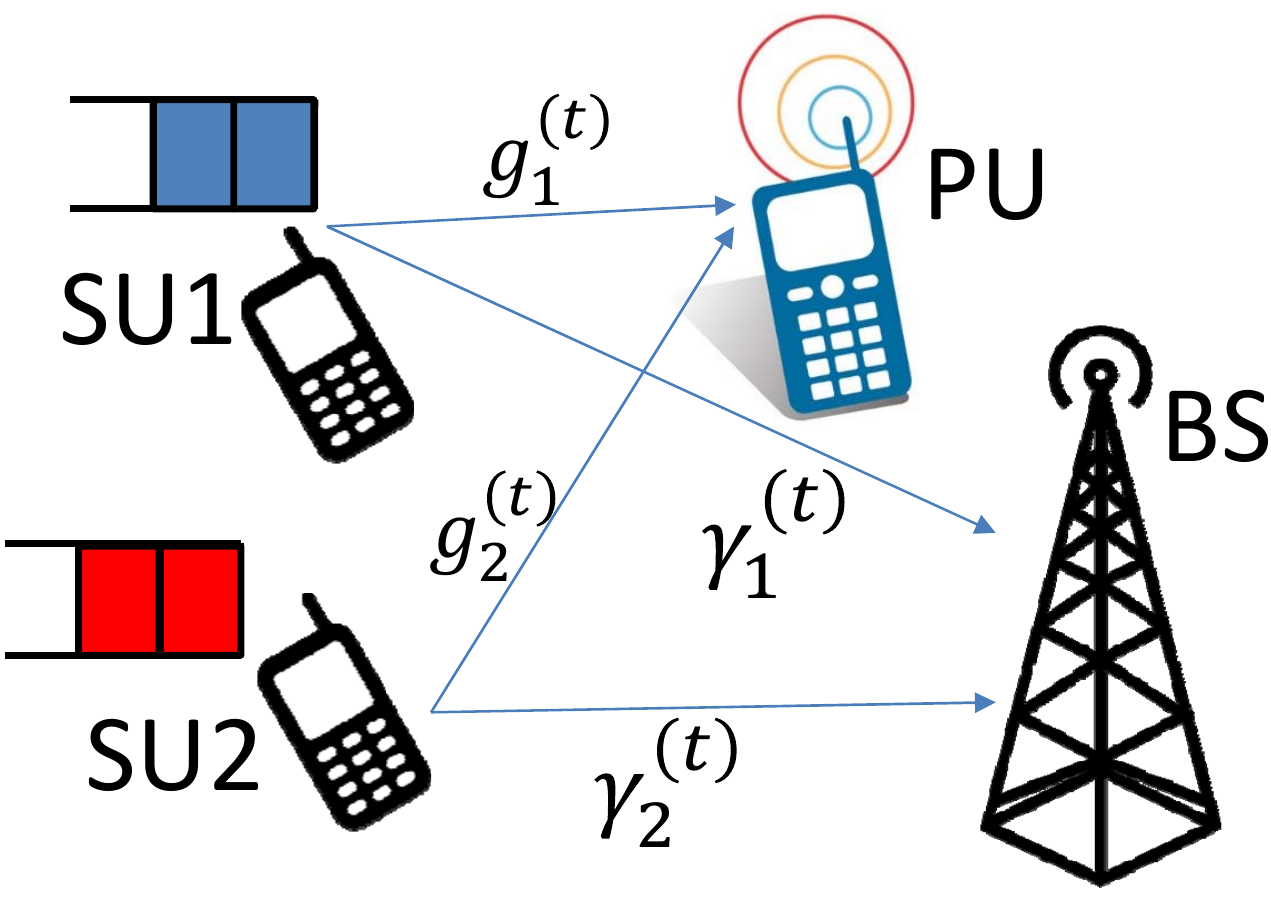}%
\caption{The CR system considered is an uplink one with $N$ SUs (in this figure $N=2$) communicating with their BS. There exists an interference link between each SU and the existing PU. The PU is assumed to be using the channel continuously.}%
\label{Cell_Fig}%
\end{figure}

\subsection{Channel and Interference Model}
\label{Channel_Model}
We assume a time slotted structure where each slot is of duration $\Ts$ seconds, and equal to the coherence time of the channel. The channel between SU $i$ and the BS is block fading with instantaneous power gain $\gamma_i^{(t)}$, at time slot $t$, following the probability mass function $\fgammai(\gamma)$ with mean $\bgamma_i$ and i.i.d. across time slots, and $\gammamax$ is the maximum gain $\gamma_i^{(t)}$ could take. The channel gain is also independent across SUs but not necessary identically distributed allowing heterogeneity among users
. 
SUs use a rate adaptation scheme based on the channel gain $\gamma_i^{(t)}$. The transmission rate of SU $i$ at time slot $t$ is
\begin{equation}
\Rit=\log \lb 1+P_i^{(t)}\gamma_i^{(t)} \rb \hspace{0.1in} \rm{bits},
\label{Tx_Rate}
\end{equation}
where $P_i^{(t)}$ is the power by which SU $i$ transmits its bits at slot $t$. We assume that there exists a finite maximum rate $\Rmax$ that the SU cannot exceed. This rate is dictated by the maximum power $\Pmax$ and the maximum channel gain $\gammamax$.

The PU experiences interference from the SUs through the channel between each SU and the PU. The interference channel between SU $i$ and the PU, at slot $t$, has a power gain $g_i^{(t)}$ following the probability mass function $\fgi(g)$ with mean $\bgi$, and having $\gmax$ as the maximum value that $\git$ could take. These power gains are assumed to be independent among SUs but not identically distributed. We assume that SU $i$ knows the value of $\gamma_i^{(t)}$ as well as $g_i^{(t)}$, at the beginning of slot $t$ through some channel estimation phase \cite{haykin2005cognitive}. The channel estimation to acquire $g_i^{(t)}$ can be done by overhearing the pilots transmitted by the primary receiver, when it is acting as a transmitter, to its intended transmitter \cite[Section VI]{haykin2005cognitive}. The channel estimation phase is out of the scope of this work, however the effect of channel estimation errors will be discussed in Section \ref{Results}.

\subsection{Queuing Model}
\subsubsection{Arrival Process} We assume that packets arrive to the SU $i$'s buffer at the beginning of each slot. The number of packets arriving to SU $i$'s buffer follows a Bernoulli process with a fixed parameter $\lambda_i$ packets per time slot. Following the literature, packets are buffered in infinite-sized buffers \cite[pp. 163]{Bertsekas_Data_Networks} and are served according to the first-come-first-serve discipline. Each packet has a fixed length of $L$ bits that is constant for all users. In this paper, we study the case where $L\gg\Rmax$ which is a typical case for packets with large sizes as video packets \cite[Section 3.1.6.1]{semiconductor2008long}. Due to the randomness in the channels, each packet takes a random number of time slots to be transmitted to the BS. This depends on the rate of transmission $\Rit$ as will be explained next.

\subsubsection{Service Process} When SU $i$ is scheduled for transmission at slot $t$, it transmits $M_i^{(t)}$ bits of the head-of-line (HOL) packet of its queue. The remaining bits of this HOL packet remain in the HOL of SU $i$'s queue until it is reassigned the channel in subsequent time slots. The values $M_i^{(t)}$ and $\HOL_i(t)$ are given by
\begin{align}
M_i^{(t)}&\triangleq\min \lb \Rit,\HOL_i(t) \rb \hspace{0.1in} \rm{bits, and}
\label{Num_Bits}\\
\HOL_i(t+1)&\triangleq \HOL_i(t) - M_i^{(t)},
\label{HOL}
\end{align}
respectively, where $\HOL_i(t)$ is the remaining number of bits of the HOL packet at SU $i$ at the beginning of slot $t$. $\HOL_i(t)$ is initialized by $L$ whenever a packet joins the HOL position of SU $i$'s queue so that it always satisfies $0\leq \HOL_i(t)\leq L$, $\forall t$. A packet is not considered transmitted unless all its $L$ bits are transmitted, i.e. unless $\HOL_i(t)$ becomes zero, at which point SU $i$'s queue decreases by 1 packet. At the beginning of slot $t+1$ the following packet in the buffer, if any, becomes SU $i$'s HOL packet and $\HOL_i(t+1)$ is reset back to $L$ bits. The SU $i$'s queue evolves as follows
\begin{equation}
Q_i^{(t+1)}= \lb Q_i^{(t)} + \vert \script{A}_i^{(t)}\vert - S_i^{(t)} \rb^+,
\label{Queue}
\end{equation}
where $\script{A}_i^{(t)}$ is the set carrying the index of the packet, if any, arriving to SU $i$ at slot $t$, thus $\vert \script{A}_i^{(t)}\vert$ is either $0$ or $1$ since at most one packet per slot can arrive to SU $i$; 
the packet service indicator $S_i^{(t)}=1$ if $\HOL_i(t)$ becomes zero at slot $t$.

The service time $s_i$ of SU $i$ is the number of time slots required to transmit one packet for SU $i$, excluding the service interruptions. It can be shown that the average service time is $L/\EE{\Rit}$ time slots per packet where the expectation is taken over the channel gain $\gamma_i^{(t)}$ as well as over the power $P_i^{(t)}$ when it is channel dependent and random. One example of a random power policy is the \emph{channel inversion} policy as will be discussed later (see equation \eqref{Power_Allocation}). The service time is assumed to follow a general distribution throughout the paper that depends on the distribution of $\Pit\gamma_i^{(t)}$.

We define the delay $W_i^{(j)}$ of a packet $j$ as the total amount of time, in time slots, packet $j$ spends in SU $i$'s buffer from the slot it joined the queue until the slot when its last bit is transmitted. 
The time-average delay experienced by SU $i$'s packets is given by \cite{li2011delay}
\begin{equation}
\bW_i \triangleq \lim_{T \rightarrow \infty} \frac{\EE{\sum_{t=1}^T{\sum_{j\in\script{A}_i^{(t)}}{W_i^{(j)}}}}}{\EE{\sum_{t=1}^T{\vert \script{A}_i^{(t)}\vert}}}
\label{Delay}
\end{equation}
which is the expected total amount of time spent by all packets arriving in a time interval, of a large duration, normalized by the expected number of packets that arrived in this interval.

\subsection{Transmission Process}
At the beginning of each time slot $t$, the BS schedules a SU and broadcasts its index $i^*$ and its power $P_{i^*}^{(t)}$ to all SUs on a common control channel. SU ${i^*}$, in turn, begins transmission of $M_{i^*}^{(t)}$ bits of its HOL packet with a constant power $P_{i^*}^{(t)}$. We assume the BS receives these bits error-free by the end of slot $t$ then a new time slot $t+1$ starts. In this paper, our main goal is the selection of the SU $i^*$ which is a scheduling problem, as well as the choice of the power $P_{i^*}^{(t)}$ which is power allocation. We now elaborate further on this problem.

\section{Problem Statement}
\label{Prob_Statement}
Each SU $i$ has an average delay constraint $\bW_i \leq d_i$ that needs to be satisfied. Moreover, there is an interference constraint that the SU needs to meet in order to coexist with the PU. We discuss the two different constraints and state the problem associated with each constraint.

\subsection{Instantaneous Interference Constraint}
Under the instantaneous interference constraint, the main objective is to solve the following problem
\begin{equation}
\begin{array}{ll}
\underset{\{i^{*(t)}\},\{\bfP{}^{(t)}\}}{\rm{minimize}}& \sum_{i=1}^N \bW_i\\
\label{Problem}
\rm{subject \; to} & \sum_{i=1}^N {P_i^{(t)}g_i^{(t)}} \leq \Iinst \hspace{0.25in} , \hspace{0.25in} \forall t\geq 1\\
& \bW_i \leq d_i\\
& P_i^{(t)} \leq \Pmax \hspace{0.1in} , \hspace{0.1in} \forall i \in \script{N} \rm{\; and \;}\forall t \geq 1,\\
& \sum_{i=1}^N{  \mathds{1} \lb P_i^{(t)}\rb} \leq 1 \hspace{0.25in}, \hspace{0.25in} \forall t \geq 1,
\end{array}
\end{equation}
where $\bfP{}^{(t)} \pardef{P}$, $\{i^{*(t)}\}$ represents the scheduler at each time slot $t\geq1$, 
while $\mathds{1}(x)\triangleq 1$ if $x\neq 0$ and $0$ otherwise. The last constraint indicates that no more than a single SU is to be transmitting at slot $t$.


\subsection{Average and Instantaneous Interference Constraint}
\label{Avg_Interf_Prob}
Let $I$ denote the long-term average interference received by the PU given by
\begin{equation}
I\triangleq \lim_{T\rightarrow \infty} \sum_{i=1}^N{\frac{1}{T}\sum_{t=1}^T{P_i^{(t)} g_i^{(t)}}}.
\label{Interference}
\end{equation}
The following problem is the same as \eqref{Problem} with an additional constraint on the average interference:
\begin{equation}
\begin{array}{ll}
\underset{\{i^{*(t)}\},\{\bfP{}^{(t)}\}}{\rm{minimize}}& \sum_{i=1}^N \bW_i \\
\label{Prob}
\rm{subject \; to} & \sum_{i=1}^N {P_i^{(t)}g_i^{(t)}} \leq \Iinst \hspace{0.25in} , \hspace{0.25in} \forall t\geq 1\\
& I \leq \Iavg\\
& \bW_i \leq d_i\\
& P_i^{(t)} \leq \Pmax \hspace{0.1in} , \hspace{0.1in} \forall i \in \script{N} \rm{\; and \;}\forall t \geq 1,\\
& \sum_{i=1}^N{  \mathds{1} \lb P_i^{(t)}\rb} \leq 1 \hspace{0.25in}, \hspace{0.25in} \forall t \geq 1,
\end{array}
\end{equation}

We notice that problems \eqref{Problem} and \eqref{Prob} are joint power allocation and scheduling problems where the objective function and constraints are expressed in terms of asymptotic time averages and cannot be solved by conventional optimization techniques. The next section proposes low complexity update policies and proves their optimality.

\section{Proposed Power Allocation and Scheduling Algorithm}
\label{Proposed_Algorithm}
We solve problems \eqref{Problem} and \eqref{Prob} by proposing online joint scheduling and power allocation policies that dynamically update the scheduling and the power allocation. We show that these policies have performances that come arbitrarily close to being optimal. That is, we can achieve a sum of the average delays arbitrarily close to its optimal value depending on some control parameter $V$.

We first discuss the idea behind our policies. Then we present the proposed policy for each problem, \eqref{Problem} and \eqref{Prob}, separately.

\subsection{Frame-Based Policy}
\label{Frame_Based_Policy}
The idea behind the policies that solve \eqref{Problem} and \eqref{Prob} is to divide time into frames where frame $k$ consists of a random number $T_k$ time slots and update the power allocation and scheduling at the beginning of each frame. Where each frame begins and ends is specified by idle periods and will be precisely defined later in this section. During frame $k$, SUs are scheduled according to some priority list $\bfpi(k)$ and each SU is assigned some power to be used when it is assigned the channel. The priority list and the power functions are fixed during the entire frame $k$ and are found at the beginning of frame $k$ based on the history of the SUs' time-averaged delays and, in the case of \eqref{Prob}, the PU's suffered interference up to the end of frame $k-1$.

We define $\bfpi(k) \parFdef{\pi}$ where $\pi_j(k)$ is the index of the SU who is given the $j$th priority during frame $k$. Given $\bfpi(k)$, the scheduler becomes a priority scheduler with preemptive-resume priority queuing discipline \cite[pp. 205]{Bertsekas_Data_Networks}. The idea of dividing time into frames and assigning fixed priority lists for each frame was also used in \cite{li2011delay}. Lemma 1 of \cite{li2011delay} proves that restricting the scheduling algorithm to frame-based preemptive-resume priority lists does not result in any loss of optimality.

Frame $k$ consists of $T_k\triangleq\vert \script{F}(k)\vert$ consecutive time-slots, where $\script{F}(k)$ is the set containing the indices of the time slots belonging to frame $k$ (see Fig. \ref{Frame_Structure}). Each frame consists of exactly one \emph{idle period} followed by exactly one \emph{busy period}, both are defined next.

\begin{figure}
\centering
\includegraphics[width=0.8\columnwidth]{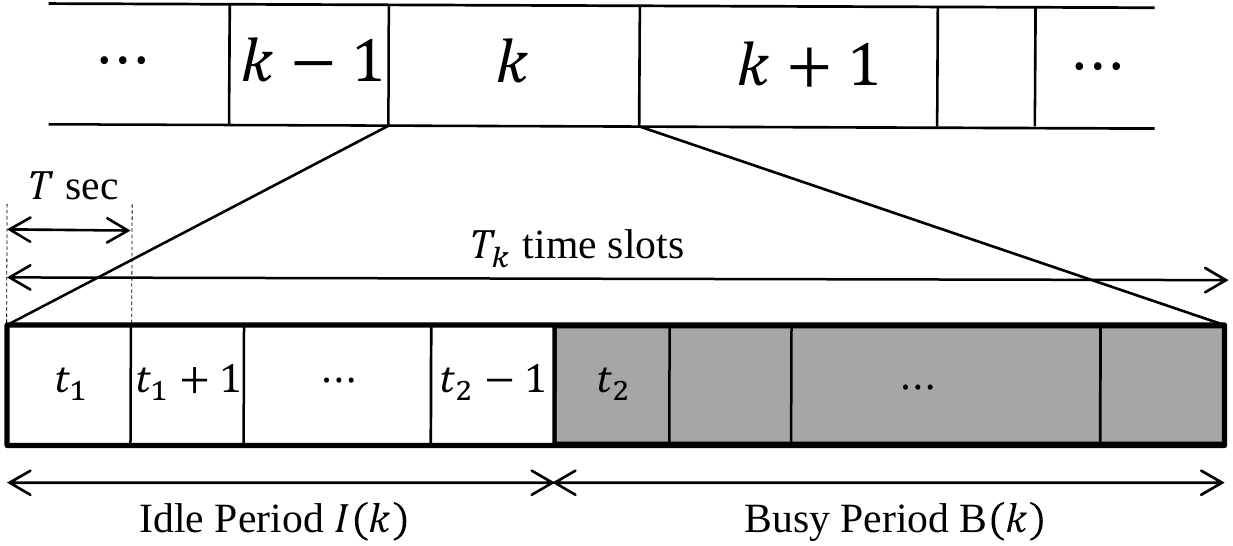}%
\caption{Time is divided into frames. Frame $k$ has $\FDurK\triangleq\vert \script{F}(k)\vert$ slots, each is of duration $\Ts$ seconds. Different frames can have different number of time slots.}%
\label{Frame_Structure}%
\end{figure}

\begin{Def}
\label{Idle_Def}
An idle period is the time interval formed by the consecutive time slots where all SUs have empty buffers. An idle period starts with the time slot $t_1$ following the completion of transmission of the last packet in the system, and ends with a time slot $t_2$ when one or more of the SUs' buffer receives one a new packet to be transmitted (see Fig. \ref{Frame_Structure}). In other words, $t_1$ satisfies $\sum_{i\in\script{N}}Q_i^{(t_1)}=0$ and $\sum_{i\in\script{N}}Q_i^{(t_1-1)}\neq 0$, while $t_2$ satisfies $\sum_{t=t_1}^{t_2-1}\sum_{i\in\script{N}}Q_i^{(t)}=0$ and $\sum_{i\in\script{N}}Q_i^{(t_2)}\neq 0$.
\end{Def}

\begin{Def}
\label{Busy_Def}
Busy period is the time interval between two consecutive idle periods.
\end{Def}
The duration of the idle period $I(k)$ and busy period $B(k)$ of frame $k$ are random variables, thus $T_k=I(k)+B(k)$ is random as well. Since frames do not overlap, if $t\in\script{F}(k_1)$ then $t \notin \script{F}(k_2)$ as long as $k_1\neq k_2$. Our goal in this paper is to choose, at the beginning of each frame $k$, the best priority list $\bfpi(k)$ as well as the best power allocation policy for each SU so that the system has an optimal average delay performance satisfying the constraints in \eqref{Problem} or \eqref{Prob}. An equivalent equation for the average delay equation in \eqref{Delay} is
\begin{equation}
\bW_i \triangleq \lim_{K \rightarrow \infty} \frac{\EE{\sum_{k=0}^K \lb\sum_{j\in \script{A}_i(k)}W_i^{(j)}\rb}}{\EE{\sum_{k=0}^{K}{\vert\script{A}_i(k)\vert}}}
\label{Delay_Frame}
\end{equation}
where $\script{A}_i(k)\triangleq\cup_{t\in\script{F}(k)}\script{A}_i^{(t)}$ is the set of all packets that arrive at SU $i$'s buffer during frame $k$. We note that the long-term average delay $\bW_i$ in \eqref{Delay_Frame} depends on the chosen priority lists as well as the power allocation policy, in all frames $k\geq 0$.

\subsection{Satisfying Delay Constraints}
In order to guarantee a feasible solution satisfying the delay constraints in problems \eqref{Problem} and \eqref{Prob}, we set up a ``virtual queue'' associated with each delay constraint $\bW_i\leq d_i$. The virtual queue for SU $i$ at frame $k$ is given by
\begin{equation}
Y_i(k+1)\triangleq\lb Y_i(k)+\sum_{j\in \script{A}_i(k)}{\lb W_i^{(j)}-r_i(k)\rb} \rb^+
\label{Delay_Q}
\end{equation}
where $r_i(k)\in[0,d_i]$ is an auxiliary random variable, that is to be optimized over and $Y_i(0)\triangleq 0$, $\forall i$. We define $\bfY(k) \parFdef{Y}$. Equation \eqref{Delay_Q} is calculated at the end of frame $k-1$ and represents the amount of delay exceeding the delay bound $d_i$ for SU $i$ up to the beginning of frame $k$. We first give the following definition, then state a lemma that gives a sufficient condition on $Y_i(k)$ for the delay of SU $i$ to satisfy $\bW_i \leq d_i$.
\begin{Def}
\label{Mean_Rate_Def}
A random sequence $\{Y_i(k)\}_{k=0}^\infty$ is mean rate stable if and only if $\lim_{K\rightarrow\infty}\EE{Y_i(K)}/K=0$ holds.
\end{Def}

\begin{lma}
\label{Mean_Rate_Lemma}
If $\{Y_i(k)\}_{k=0}^\infty$ is mean rate stable, then the time-average delay of SU $i$ satisfies $\bW_i \leq d_i$.
\end{lma}
\begin{proof}
Removing the $(\cdot)^+$ sign from equation \eqref{Delay_Q} yields
\begin{equation}
Y_i(k+1) \geq Y_i(k)+\sum_{j\in \script{A}_i(k)}{\lb W_i^{(j)}-r_i(k)\rb}.
\label{Inequality1}
\end{equation}
Summing inequality \eqref{Inequality1} over $k=0,\cdots K-1$ and noting that $Y_i(0)=0$ gives
\begin{equation}
Y_i(K)\geq \sum_{k=0}^{K-1} \lb\sum_{j\in \script{A}_i(k)}W_i^{(j)}\rb-\sum_{k=0}^{K-1} {\lb r_i(k) \vert\script{A}_i(k)\vert\rb}.
\label{Inequality2}
\end{equation}
Taking the $\EE{\cdot}$ then dividing by $\EE{\sum_{k=0}^{K-1}{\vert\script{A}_i(k)\vert}}$ gives
\begin{equation}
\frac{\EE{\sum_{k=0}^{K-1} \lb\sum_{j\in \script{A}_i(k)}W_i^{(j)}\rb}}{\EE{\sum_{k=0}^{K-1}{\vert\script{A}_i(k)\vert}}} \leq \frac{\EE{Y_i(K)}}{K}\frac{K}{\EE{\sum_{k=0}^{K-1}{\vert\script{A}_i(k)\vert}}} + \frac{\sum_{k=0}^{K-1}{\EE{\vert\script{A}_i(k)\vert {r_i(k)}}}}{\sum_{k=0}^{K-1}\EE{\vert\script{A}_i(k)\vert}}.
\label{Wait_r_i}
\end{equation}
Replacing $r_i(k)$ by its upper bound $d_i$, taking the limit as $K\rightarrow\infty$ then using the mean rate stability definition and equation \eqref{Delay_Frame} completes the proof.
\end{proof}
Lemma \ref{Mean_Rate_Lemma} provides a condition on the virtual queue $\Yivq$ so that SU $i$'s average delay constraint $\bW_i\leq d_i$ in \eqref{Problem} and \eqref{Prob} is satisfied. That is, if the proposed joint power allocation and scheduling policy results in a mean rate stable $\Yivq$, then $\bW_i\leq d_i$. For both problems, the proposed policy depends on the Lyapunov optimization where the goal is to choose the joint scheduling and power allocation policy that minimizes the drift-plus-penalty. In Section \ref{Algorithm_DOIC} (Section \ref{Algorithm_DOAC}) we will show that if problem \eqref{Problem} (problem \eqref{Prob}) is feasible, then the proposed policy guarantees mean rate stability for the queues $\Yivq$.

\subsection{Algorithm for Instantaneous Interference Constraint}
\label{Algorithm_DOIC}
We now propose the \emph{Delay Optimal with Instantaneous Interference Constraint} (\emph{DOIC}) policy that solves problem \eqref{Problem}. This policy is executed at the beginning of each frame $k$ for finding $\bfP{}^{(t)}$ as well as the optimum list $\bfpi(k)$, given some prespecified control parameter $V$. Define the random variable $R_i(P)$ as
\begin{equation}
R_i(P)\triangleq\log\lb 1+\min\lb \frac{\Iinst}{g_i^{(t)}},P\rb\gamma_i^{(t)}\rb,
\label{Rate_Explicit}
\end{equation}
where $P$ is some fixed constant argument and define $\mu_i(P)\triangleq\EE{R_i(P)}/L$ where the expectation is taken over $\git$ and $\gamma_i^{(t)}$. The {\DOIC} policy is as follows.\\
\noindent{\bf DOIC Policy} (executed at the beginning of frame $k$):
\begin{enumerate}
	\item The BS sorts the SUs according to the descending order of $Y_i(k)\mu_i(\Pmax)$. The sorted list is denoted by $\bfpi(k)$.
	\item At the beginning of each slot $t\in\script{F}(k)$ the BS schedules SU $i^*$ that has the highest priority in the list $\bfpi(k)$ among those having non-empty buffers.
	\item SU $i^*$, in turn, transmits $M_{i^*}^{(t)}$ packets as dictated by equation \eqref{Num_Bits} where $P_i^{(t)}=0$ $\forall i\neq i^*$ while $P_{i^*}^{(t)}$ is calculated as
\begin{equation}
	P_{i^*}^{(t)}=\min\lb\frac{\Iinst}{g_{i^*}^{(t)}},\Pmax\rb,
\label{Power_Allocation}
\end{equation}
	\item At the end of frame $k$, for all $i\in \script{N}$ the BS updates:
	\begin{enumerate}
		\item $r_i(k)= d_i$ if $V<Y_i(k)\lambda_i$, and $r_i(k)=0$ otherwise, and then
		\item $Y_i(k+1)$ via equation \eqref{Delay_Q}.
	\end{enumerate}
\end{enumerate}
Before we discuss the optimality of the {\DOIC} in Theorem \ref{Optimality}, we define the following quantities. Let $a\triangleq 1-\Pi_{i=1}^N \lb 1-\lambda_i \rb$ denote the probability of receiving a packet from a user or more at a given time slot, while $C_Y\triangleq\sum_{i=1}^N C_{Y_i}$ with $C_{Y_i}\triangleq \sqrt{\EE{A^4}\EE{B^4}} + d_i^2\EE{A^2}$, where $\EE{A^2}$ and $\EE{A^4}$ are bounds on the second and fourth moments of the total number of arrivals $\sum_i\vert \script{A}_i(k)\vert$ during frame $k$, respectively, while $\EE{B^4}$ is a bound on the fourth moment of the busy period $B(k)$. The finiteness of these moments can be shown to hold if the first four moments of the service time are finite. In Appendix \ref{No_Deep_Fade} we show that all the service time moments exist given any distribution for $\Pit\gamma_i^{(t)}$.

\begin{thm}
\label{Optimality}
If problem \eqref{Problem} is feasible, then the proposed {\DOIC} policy results in a time average of the SUs' delays satisfying the following inequality
\begin{equation}
\sum_{i=1}^N{\bW_i} \leq \frac{aC_Y}{V} + \sum_{i=1}^N{\bW_i^*}
\label{Optimality_Equation}
\end{equation}
where $\bW_i^*$ is the optimum value of the delay when solving problem \eqref{Problem}, while $a$ and $C_Y$ are as given above. Moreover, the virtual queues $\Yivq$ are mean rate stable $\forall i \in \script{N}$.
\end{thm}

\begin{proof}
See Appendix \ref{Optimality_Proof_Inst}.
\end{proof}

Theorem \ref{Optimality} says that the objective function of problem \eqref{Problem} is upper bounded by the optimum value $\sum_i\bW_i^*$ plus some constant gap that vanishes as $V\rightarrow\infty$. Having a vanishing gap means that the {\DOIC} policy is asymptotically optimal. Moreover, based on the mean rate stability of the queues $\Yivq$, the set of delay constraints of problem \eqref{Prob} is satisfied. The drawback of setting $V$ very large is that the time needed for the algorithm to converge increases. This increase is linear in $V$ \cite{NeelyPhD}. That is, if the number of frames required for the quantity $\sum_i{Y_i(k)}/(Nk)$ to be less than $\epsilon$ (for some $\epsilon>0$) is $O(K_1)$, then increasing $V$ to $\beta V$ will require $O(\beta K_1)$ frames for it to be less than $\epsilon$, for any $\beta>1$. We note that the complexity of the {\DOIC} policy is $O(N)$ because calculating $\mu_i(\Pmax)$ is of $O(1)$, while the power is closed-form in \eqref{Power_Allocation}. We note that if problem \eqref{Problem} is not feasible, then this is because one of two reasons; either one or more of the constraints is stringent, or otherwise because $\sum_{i=1}^N\lambda_i/\mu_i(\Pmax)\geq1$. If it is the former, then the {\DOIC} policy will result in a point that is as close as possible to the feasible region. On the other hand, if it is the latter, then we could add an admission controller that limits the average number of packets arriving at buffer $i$ to $\lambda_i(1-\epsilon)/\lb \sum_{i=1}^N\lambda_i/\mu_i(\Pmax)\rb$ for some $\epsilon>0$.

\subsection{Algorithm for Average Interference Constraint}
\label{Algorithm_DOAC}
We now propose the Delay-Optimal-with-Average-Interference-Constraint {\DOAC} policy for problem \eqref{Prob}. We first give the following useful definitions. Since the scheduling scheme in frame $k$ is a priority scheduling scheme with preemptive-resume queuing discipline, then given the priority list $\bfpi$ we can write the expected waiting time of all SUs in terms of the average residual time \cite[pp. 206]{Bertsekas_Data_Networks} defined as $T_{\pi_j}^{\rm R}\triangleq \sum_{l=1}^j\lambda_{\pi_l}\EE{s_{\pi_l}^2}/2$, where the expectation is taken over $P_{\pi_l}^{(t)}\gamma_{\pi_l}^{(t)}$. The waiting time of SU $\pi_j$ that is given the $j$th priority is \cite[pp. 206]{Bertsekas_Data_Networks}
\begin{align}
W_{\pi_j}\lb P,\mu_{\pi_j}(P),\rho_{\pi_j}(P),\brho_{\pi_{j-1}},T_{\pi_j}^{\rm R}\rb&=\frac{1}{\lb 1-\brho_{\pi_{j-1}}\rb}\left[\frac{1}{\mu_{\pi_j}(P)} + \frac{T_{\pi_j}^{\rm R}}{\lb 1-\brho_{\pi_{j-1}} - \rho_{\pi_j}(P)\rb}\right]
\label{Priority_Delay}\\
\nonumber&\leq \frac{1}{\lb 1-\brhomax_{\pi_{j-1}}\rb}\left[\frac{1}{\mu_{\pi_j}(P)} + \frac{T^{\rm R}}{\lb 1-\brhomax_{\pi_{j-1}} - \rho_{\pi_j}(P)\rb}\right]\\
&\triangleq \Wup\lb P,\rho_{\pi_j}(P),\brhomax_{\pi_{j-1}},T^{\rm R}\rb,
\label{Priority_Delay_UB}
\end{align}
where we define $\rho_i(P)\triangleq \lambda_i/\mu_i(P)$, $\brho_{\pi_{j-1}} \triangleq \sum_{l=1}^{j-1} \rho_{\pi_l}(P_{\pi_l})$, while $T^{\rm R}$ is an upper bound on $T_{\pi_j}^{\rm R}$ and is given by $T^{\rm R}\triangleq\sum_{i=1}^N\lambda_i\lb L^2+L\lb 1-p_i(\Pminn_i)\rb\rb/p_i^2(\Pminn_i)/2$ with $p_i(P)\triangleq1-\Prob{R_i(P)=0}$ and $\Pminn_i$ is the minimum power satisfying $\rho_i(\Pminn_i)+\sum_{j\neq i}\rho_j(\Pmax)<1$ (see Appendix \ref{No_Deep_Fade} for the derivation of $T^{\rm R}$), while $\brhomax_i$ is some upper bound on $\brho_i$ that will be defined later. We henceforth drop all the arguments of $\Wup(P,\brhomax_{\pi_{j-1}})$ except $P$ and $\brhomax_{\pi_{j-1}}$ and all those from $W_{\pi_j}(P)$ except $P$.

To track the average interference at the PU up to the end of frame $k$ we set up the following virtual queue that is associated with the average interference constraint in problem \eqref{Prob} and is calculated at the BS at the end of frame $k$.
\begin{equation}
X(k+1)\triangleq \lb X(k) + \sum_{i=1}^N{\sum_{t\in\script{F}(k)}{P_i^{(t)} g_i^{(t)}}} -\Iavg T_k\rb^+,
\label{Avg_Interf_Q}
\end{equation}
where the term $\sum_{i=1}^N{\sum_{t\in\script{F}(k)}{ P_i^{(t)}g_i^{(t)}}}$ represents the aggregate amount of interference energy received by the PU due to the transmission of the SUs during frame $k$. Hence, this virtual queue is a measure of how much the SUs have exceeded the interference constraint above the level $\Iavg$ that the PU can tolerate. Lemma \ref{Mean_Rate_Lemma_Avg_Interf} provides a sufficient condition for the interference constraint of problem \eqref{Prob} to be satisfied.

\begin{lma}
\label{Mean_Rate_Lemma_Avg_Interf}
If $\{X(k)\}_{k=0}^\infty$ is mean rate stable, then the time-average interference received by the PU satisfies $I \leq \Iavg$.
\end{lma}
\begin{proof}
The proof is similar to that of Lemma \ref{Mean_Rate_Lemma} and is omitted for brevity.
\end{proof}
Lemma \ref{Mean_Rate_Lemma_Avg_Interf} says that if the power allocation and scheduling algorithm results in mean rate stable $\Xvq$, then the interference constraint of problem \eqref{Prob} is satisfied.

Before presenting the {\DOAC} policy, we first discuss the idea behind it. Theorem \ref{Optimality_Avg} will show that the optimum power allocation for SU $i$ is
	\begin{equation}
	\Pit=\min\lb\frac{\Iinst}{g_i^{(t)}},P_i(k)\rb,
	\label{Pow_Allocation}
	\end{equation}
where $P_i(k)\in[\Pminn_i,\Pmax]$ is a power parameter that is fixed within frame $k$ (i.e. $\forall t\in\script{F}(k)$). Intuitively, a policy that solves problem \eqref{Prob} should allocate SU $i$'s power and assign its priority such that SU $i$'s expected delay and the expected interference to the PU is minimized. The {\DOAC} policy is defined as the policy that selects the power parameter vector $\bfP{}(k)\parFdef{P}$ jointly with the priority list $\bfpi(k)$ that minimizes  $\Psi\triangleq\sum_{j=1}^N\psi_{\pi_j}(P_{\pi_j}(k),\brhomax_{\pi_{j-1}})$ where
\begin{align}
\psi_{\pi_j}(P,\brhomax_{\pi_{j-1}})&\triangleq \psiDel_{\pi_j}(P,\brhomax_{\pi_{j-1}}) + \psiInt_{\pi_j}(P),  \hspace{0.1in} {\rm with}
\label{Optimization_Obj}\\
\psiDel_{\pi_j}(P,\brhomax_{\pi_{j-1}})&\triangleq Y_{\pi_j}(k) \lambda_{\pi_j} \Wup(P,\brhomax_{\pi_{j-1}}), \hspace{0.1in} {\rm while}\\
\psiInt_{\pi_j}(P)&\triangleq X(k)\rho_{\pi_j}(P)P.
\end{align}
The function $\psiDel_{\pi_j}(P,\brhomax_{\pi_{j-1}})$ (and $\psiInt_{\pi_j}(P)$) represents the amount of delay (interference) that SU $\pi_j$ is expected to experience (to cause to the PU) during frame $k$.

To minimize $\Psi$ in an efficient way we need to set the function $\brhomax_{\pi_j}$, that upper bounds $\brho_{\pi_j}$, to some function that does not depend except on the power $P_{\pi_j}(k)$ of user $\pi_j$. Thus, the functions $\psi_{\pi_j}(P_{\pi_j}(k),\brhomax_{\pi_{j-1}})$ become decoupled for all $j\in\script{N}$. That is, each $\psi_{\pi_j}(P_{\pi_j}(k),{\pi_{j-1}})$ is a function in $P_{\pi_j}(k)$ only. This $\brhomax_{\pi_j}$ functions, for all $j\in \script{N}$, are given by
\begin{equation}
\brhomax_{\pi_j}\triangleq \sum_{l=1}^j \rho_{\pi_l}\lb P_{\pi_l}^{\brhomax}\rb,
\label{brhomax}
\end{equation}
where
\begin{equation}
P_{\pi_j}^{\brhomax}\triangleq\arg \min_P \psi_{\pi_j}\lb P,\brhomax_{\pi_{j-1}}\rb.
\label{Local_psi_Inequality}
\end{equation}
Equation \eqref{Local_psi_Inequality} means that in order to find $P_{\pi_j}^{\brhomax}$ we need to find $P_{\pi_{j-1}}^{\brhomax}$. Hence, we find $P_{\pi_j}^{\brhomax}$ recursively starting from $j=1$ at which $\brhomax_{\pi_0}=0$ by definition. We will show that $\brhomax_{\pi_j}$ is an upper bound on $\brho_{\pi_j}$ in the following lemma.
\begin{lma}
\label{Local_psi_Lma}
Given some priority list $\bfpi(k)$, for any user $\pi_j\in\script{N}$ the function $\brho_{\pi_j}$ evaluated at the power vector $\bfP{}^{\brho}$ which is the power vector that minimizes $\sum_{j=1}^N\psi_{\pi_j}(P_{\pi_j}(k),\brho_{\pi_{j-1}})$, is upper bounded by $\brhomax_{\pi_j}$. Namely,
\begin{equation}
\left.\brho_{\pi_j}\right\vert_{\bfP{}^{\brho}}\leq \left.\brhomax_{\pi_j}\right\vert_{\bfP{}^{\brhomax}}
\label{brho_leq_brhomax}
\end{equation}
where
\begin{equation}
\bfP{}^{\brho}\triangleq \arg\min_{\bfP{}}\sum_{j=1}^N\psi_{\pi_j}(P_{\pi_j},\brho_{\pi_{j-1}})
\label{bfP_rho}
\end{equation}
while
\begin{equation}
\bfP{}^{\brhomax}\triangleq \arg\min_{\bfP{}}\sum_{j=1}^N\psi_{\pi_j}(P_{\pi_j},\brhomax_{\pi_{j-1}})
\label{bfP_rhomax}
\end{equation}
\end{lma}
\begin{proof}
We first argue that $P_{\pi_j}^{\brho}\geq P_{\pi_j}^{\brhomax}$ for any $j\in\script{N}$. Then we show that $\brho_{\pi_j}$ is decreasing in $P_{\pi_l}$ for all $l\leq j$ which completes the proof.

From \eqref{bfP_rhomax}, we have
\begin{equation}
P_{\pi_j}^{\brhomax}=\underset{P_{\pi_j}\leq\Pmax}{\arg\min}\sum_{l=1}^N \psi_{\pi_l}\lb P_{\pi_l},\brhomax_{\pi_{l-1}}\rb.
\end{equation}
But since $\brhomax_{\pi_j}$ is not a function in $P_{\pi_l}$ except if $j=l$, then
\begin{equation}
P_{\pi_j}^{\brhomax}=\underset{P_{\pi_j}\leq\Pmax}{\arg\min}\left[\psiInt_{\pi_j}\lb P_{\pi_j}\rb+\psiDel_{\pi_j}\lb P_{\pi_j},\brhomax_{\pi_{j-1}}\rb \right]
\label{P_pi_j_brhomax}
\end{equation}

Moreover, from \eqref{bfP_rho}, we have
\begin{align}
P_{\pi_j}^{\brho}&=\underset{P_{\pi_j}\leq\Pmax}{\arg\min}\sum_{l=1}^N \psi_{\pi_l}\lb P_{\pi_l},\brho_{\pi_{l-1}}\rb\\
&=\underset{P_{\pi_j}\leq\Pmax}{\arg\min}\left[\psiInt_{\pi_j}\lb P_{\pi_j}\rb+\sum_{l=j}^N \psiDel_{\pi_l}\lb P_{\pi_l},\brho_{\pi_{l-1}}\rb \right]
\label{Obj_ToBe_Differentiated}
\end{align}
If $\psiInt_{\pi_j}\lb P_{\pi_j}\rb$ is non increasing in $P_{\pi_j}$ over its entire domain, then the optimum solution for \eqref{Obj_ToBe_Differentiated} is $P_{\pi_j}^{\brho}=\Pmax$, which is the same as the optimum solution of $P_{\pi_j}^{\brhomax}$. Hence, we continue the proof assuming that there exists a region in the domain of $P_{\pi_j}$ where $\psiInt_{\pi_j}\lb P_{\pi_j}\rb$ is increasing.

Since $\Wup(P_{\pi_j},\brho_{\pi_{j-1}})$ is decreasing in $P_{\pi_j}$, $\psiDel_{\pi_j}\lb P_{\pi_j},\brho_{\pi_{j-1}}\rb$ is also decreasing in $P_{\pi_j}$. Hence, the summation in \eqref{Obj_ToBe_Differentiated} is decreasing in $P_{\pi_j}$. At the same time, $\psiInt_{\pi_j}\lb P_{\pi_j}\rb$ is increasing in $P_{\pi_j}$. Thus, there are two forces that determines the optimum value of $P_{\pi_j}^{\brho}$; the first term is in favor of decreasing it, while the summation is in favor of increasing it. We continue the proof by induction on $j$. Setting $j=1$ in \eqref{Obj_ToBe_Differentiated}, we can easily see that if we neglect all the terms in the summation except one, the value of $P_{\pi_1}^{\brho}$ decreases. Namely,
\begin{align}
\label{Neglect_All_Terms_But_One_1}
P_{\pi_1}^{\brho}&\geq\underset{P_{\pi_1}\leq\Pmax}{\arg\min}\left[\psiInt_{\pi_1}\lb P_{\pi_1}\rb+\psiDel_{\pi_1}\lb P_{\pi_1},\brho_{\pi_0}\rb \right]\\
&=P_{\pi_1}^{\brhomax},
\label{P_brho_geq_P_brhomax_1}
\end{align}
where the last equation follows after setting $j=1$ in \eqref{P_pi_j_brhomax}. From \eqref{Rate_Explicit} we can see that $\brho_{\pi_1}$ is decreasing in $P_{\pi_1}$. With this in mind, and after using the inequality in \eqref{P_brho_geq_P_brhomax_1} we get $\brho_{\pi_1}\leq\brhomax_{\pi_1}$. Setting $j=2$ in \eqref{Obj_ToBe_Differentiated} and neglecting all terms in the summation except at $j=2$ yields
\begin{align}
\label{Neglect_All_Terms_But_One_2}
P_{\pi_2}^{\brho}&\geq\underset{P_{\pi_2}\leq\Pmax}{\arg\min}\left[\psiInt_{\pi_2}\lb P_{\pi_2}\rb+\psiDel_{\pi_2}\lb P_{\pi_2},\brho_{\pi_1}\rb \right]\\
&\geq\underset{P_{\pi_2}\leq\Pmax}{\arg\min}\left[\psiInt_{\pi_2}\lb P_{\pi_2}\rb+\psiDel_{\pi_2}\lb P_{\pi_2},\brhomax_{\pi_1}\rb \right]\\
&=P_{\pi_2}^{\brhomax}.
\end{align}
Thus we get $\brho_{\pi_2}\leq\brhomax_{\pi_2}$. Repeating for a general $j\leq N$ and assuming that $\brho_{\pi_{j-1}}\leq\brhomax_{\pi_{j-1}}$, we get $P_{\pi_2}^{\brho}\geq P_{\pi_2}^{\brhomax}$ yielding $\brho_{\pi_j}\leq\brhomax_{\pi_j}$ which completes the proof.
\end{proof}
Lemma \ref{Local_psi_Lma} states that we can replace $\brho_{\pi_j}$ by $\brhomax_{\pi_j}$ as an upper bound in \eqref{Priority_Delay_UB}. $\brhomax_{\pi_j}$ has an advantage over $\brho_{\pi_j}$ (and hence $\psi_{\pi_j}\lb P_{\pi_j},\brhomax_{\pi_{j-1}}\rb$ over $\psi_{\pi_j}\lb P_{\pi_j},\brho_{\pi_{j-1}}\rb$) which is that it is not a function in $P_{\pi_l}$ as long as $l\neq j$. This decouples the power search optimization problem to $N$ one-dimensional searches

To solve $\min_{\bfpi(k),\bfP{}(k)} \Psi$, we use the dynamic programming illustrated in Algorithm \ref{DOACopt}.
\begin{algorithm}
\caption{{\DOACopt}: Optimization-problem-solution algorithm called by the {\DOAC} policy at the beginning of frame $k$ to solve for $\bfPst(k)$ as well as $\bfpist(k)$.}
\begin{algorithmic}[1]
\label{DOACopt}
\STATE Define $\script{S}$ as the set of all sets formed of all subsets of $\script{N}$ and define the auxiliary functions
\begin{align*}
&\tPsi(\cdot,\cdot):\script{N}\times\sS\rightarrow \mathbb{R}^+\\
& \trho(\cdot):\sS\rightarrow [0,1],\\
&\tS(\script{X}):\sS\rightarrow\script{N}^{\vert\script{X}\vert},\\
&\tP(\script{X}):\sS\rightarrow[0,\Pmax]^{\vert\script{X}\vert},\\
&\bP(\cdot,\cdot):\sS\times\script{N}\rightarrow[0,\Pmax].
\end{align*}
\FOR{$i=1,\cdots, N$}
\STATE In stage $i$, the first $i$ priorities have been assigned to $i$ users. The corresponding priority list is denoted $[\pi_1,\cdots,\pi_i]$. In stage $i$ we have $\binom{N}{i}$ states each corresponds to a set $j$ formed from all possible combinations of $i$ elements chosen from the set $\script{N}$. We calculate $\tPsi(i,j)$ associated with each state $j$ in terms of $\tPsi(i-1,\cdot)$ obtained in stage $i-1$ as follows.
\FOR{$j\in$ all possible $i$-element sets}
\STATE At state $j\triangleq \{\pi_1,\cdots,\pi_i\}$, we have $i$ transitions, each connects it to state $j'\triangleq j\backslash l$ in stage $i-1$, $\forall l\in j$. Find the power associated with each transition $l\in j$ denoted
\begin{equation}
\bP(j,l)\triangleq\arg\min_P \psi_l(P,\trho(j\backslash l))
\label{Local_psi_minimization}
\end{equation}
\STATE Set
\begin{align*}
& l^*=\arg\min_{l\in j}\tPsi\lb i-1,j\backslash l\rb+ \psi_l\lb\bP(j,l),\trho(j\backslash l)\rb,\\
&\tPsi(i,j)=\tPsi(i-1,j\backslash l^*)+ \psi_{l^*}\lb\bP(j,l^*),\trho(j\backslash l^*)\rb,\\
&\trho(j)=\trho\lb j\backslash l^*\rb+\rho\lb\bP(j,l^*)\rb,\\
&\tS(j)=\left [\tS\lb j\backslash l^*\rb,l^*\right],\\
&\tP(j)=\left [\tP\lb j\backslash l^*\rb,\bP(j,l^*)\right].
\end{align*}
\ENDFOR
\ENDFOR
\STATE Set $\bfpist(k)= \tS\lb\script{N}\rb$ and $\bfPst(k)=\tP\lb\script{N}\rb$.
\end{algorithmic}
\end{algorithm}
 Its search complexity is of $O(MN2^N)$ where $M$ is the number of iterations in a one-dimensional search. Compared to the complexity of $O(M^N \cdot N!)$ which is that of the $N$-dimensional power search along with the brute-force of all $N!$ permutations of priority list $\bfpi(k)$, this is a large complexity reduction. However, the $O(MN2^N)$ is still high if $N$ was large. Finding an optimal algorithm with a lower complexity is extremely difficult since the scheduling and power control problem are coupled. In other words, in order to find the optimum scheduler we need to know the optimum power vector and vice versa. In Section \ref{Suboptimal} we propose a sub-optimal policy with a very low complexity and little degradation in the delay performance. We now present the {\DOAC} policy that the BS executes at the beginning of frame $k$.

{\bf {\DOAC} Policy} (executed at the beginning of frame $k$):
\begin{enumerate}
	\item The BS executes {\DOACopt} in Algorithm \ref{DOACopt} to find the optimum power parameter vector $\bfPst(k) \parFdef{P^*}$ as well as the optimum priority list $\bfpist(k) \parFdef{\pi^*}$ that will be used during frame $k$.
	\item The BS broadcasts the vector $\bfPst(k)$ to the SUs.
	\item At the beginning of each slot $t\in\script{F}(k)$, the BS schedules SU $i^*$ that has the highest priority in the list $\bfpist(k)$ among those having non-empty buffers.
	\item SU $i^{*(t)}$, in turn, transmits $M_{i^{*(t)}}^{(t)}$ bits as dictated by equation \eqref{Num_Bits} where $P_i^{(t)}=0$ for all $i\neq i^{*(t)}$ while $P_{i^{*(t)}}^{(t)}$ is given by equation \eqref{Pow_Allocation}.
	\item At the end of frame $k$, for all $i\in \script{N}$ the BS updates:
	\begin{enumerate}
		\item $r_i(k)= d_i$ if $V<Y_i(k)\lambda_i$, and $r_i(k)=0$ otherwise,
		\item $X(k+1)$ via equation \eqref{Avg_Interf_Q},
		\item $Y_i(k+1)$ via equation \eqref{Delay_Q}, $\forall i\in \script{N}$.
	\end{enumerate}
\end{enumerate}
Define $C_X\triangleq\lb\Pmax^2\gmax^2+\Iavg^2\rb\lb(1-a)(2+a)+\EE{B^2}+2\EE{B}(a-a^2)\rb/a^2$ and $C\triangleq C_Y+C_X$ where $\EE{B}$ is a bound on the mean of $B(k)$. It can be shown that $\EE{B}$ and $\EE{B^2}$ are finite since the first two moments of the service time are finite (see Appendix \ref{No_Deep_Fade}). Thus, $C_X$ is finite. Next, we state Theorem \ref{Optimality_Avg} that discusses the optimality of the {\DOAC} policy.

\begin{thm}
\label{Optimality_Avg}
When the BS executes the {\DOAC} policy, the time average of the SUs' delays satisfy the following inequality
\begin{equation}
\sum_{i=1}^N{\bW_i} \leq \frac{aC}{V} + \sum_{i=1}^N{\bW_i^*}
\label{Optimality_Equation_Avg}
\end{equation}
where $\bW_i^*$ is the optimum value of the delay when solving problem \eqref{Prob}. Moreover, the virtual queues $\Xvq$ and $\Yivq$ are mean rate stable $\forall i \in \script{N}$.
\end{thm}

\begin{proof}
See Appendix \ref{Optimality_Proof}.
\end{proof}

Theorem \ref{Optimality_Avg} says that the objective function of problem \eqref{Prob} is upper bounded by the optimum value $\sum_i\bW_i^*$ plus some constant gap that vanishes as $V\rightarrow\infty$. Having a vanishing gap means that the {\DOAC} policy is asymptotically optimal. Moreover, based on the mean rate stability of $\Xvq$ and $\Yivq$, the interference and delay constraints of problem \eqref{Prob} are satisfied.


\subsection{Near-Optimal Low Complexity Algorithm for Average Interference Problem}
\label{Suboptimal}
As seen in the {\DOAC} policy, the complexity of finding the optimal power vector and priority list can be high when the number of SUs $N$ is large. This is mainly due to the large complexity of Algorithm \ref{DOACopt}. In this subsection we propose a suboptimal solution with an extreme reduction in complexity and with little degradation in the performance. This solution solves for the power allocation and scheduling algorithm, thus it replaces the Algorithm \ref{DOACopt}.

The challenges in Algorithm \ref{DOACopt} are three-fold. First finding the priority list (scheduling problem) requires the search over $N!$ possibilities. Second, even with a genie-aided knowledge of the optimum list, we still have to carry-out $N$ one-dimensional searches to find $\bfPst(k)$ (power control problem). Third, the scheduling and power control problems are coupled. We tackle the latter two challenges first, by finding a low-complexity power allocation policy that is independent of the scheduling algorithm. Then we use the $c\mu$ rule \cite{c_mu_Rule} to find the priority list. The $c\mu$ rule is a policy that gives the priority list that minimizes the quantity $\sum_{i=1}^N Y_i(k)\lambda_i W_i(P_i(k))$, given some power allocation vector $\bfP{}(k)$.

Define $\Pmin$ to be the minimum power that satisfies $\sum_{j=1}^N\rho_{\pi_j}(\Pmin)<1$. Intuitively, if, for some $\pi_j\in\script{N}$, $X(k)\gg Y_{\pi_j}(k)$ then $P_{\pi_j}^*(k)$ is expected to be close to $\Pmin$ since the interference term $\psiI(P)$ dominates over $\psiD(P)$ in the $\pi_j$th term of the summation in equation \eqref{Optimization_Obj}. On the other hand, if $X(k)\ll Y_{\pi_j}(k)$ then $P_{\pi_j}^*(k)\approx \Pmax$. We propose the following power allocation policy for SU ${\pi_j}\in\script{N}$
\begin{equation}
\hat{P}_{\pi_j}(k)=
\left\{
\begin{array}{lll}
	\Pmin \mbox{ if } X(k)>Y_{\pi_j}(k)\\
	\Pmax \mbox{ otherwise.}
\end{array}
\right.
\label{Subopt_Power}
\end{equation}
We can see that the power allocation policy in \eqref{Subopt_Power} does not depend on the position of SU $i$ in the priority list as opposed to Algorithm \ref{DOACopt} which requires the knowledge of SU $\pi_j$'s priority position. In other words, $\hat{P}_{\pi_j}(k)$ is a function of $\pi_j$ but it is not a function of $j$. 
Before proposing the scheduling policy, we note the following two properties based on the knowledge of the power $\bfPst(k)$. First, when $X(k)=0$, the solution to the minimization problem $\min_\bfpi \Obj$ is given by the $c\mu$ rule \cite{c_mu_Rule} that sorts the SUs according to the descending order of $Y_{\pi_j}(k)\mu_{\pi_j}(\hat{P}_{\pi_j}(k))$. Second, when $Y_{\pi_j}(k)=0$ $\forall \pi_j \in \script{N}$, any sorting order would not affect the objective function $\Obj$.


The two-step scheduling and power allocation algorithm that we propose is 1) allocate the power vector $\bfP{}(k)$ according to \eqref{Subopt_Power}, then 2) assign priorities to the SUs in a descending order of $Y_{\pi_j}(k)\mu_{\pi_j}(\hat{P}_{\pi_j}(k))$ (the $c\mu$ rule). The complexity of this algorithm is that of sorting $N$ numbers, namely $O(N\log(N))$. This is a very low complexity if compared to that of the {\DOAC} policy of $O(MN \cdot N!)$. In Section \ref{Results} we will demonstrate that this huge reduction of complexity causes little degradation to the delay performance.

\section{Achievable Rate Region of the {\DOAC}}
We have shown that the {\DOAC} policy is delay optimal. In this section we show how much of the capacity region this policy achieves. We also present different scenarios where the {\DOAC} policy achieves the whole capacity region, hence becoming both throughput optimal and delay optimal at the same time.

Theorem 1 in \cite[pp. 52] {NeelyPhD} explicitly states the capacity region in the case of an instantaneous power constraint. 
%
%
In general this capacity region is strictly convex. A simple example of this capacity region is shown in Fig. \ref{Capacity_Region} for a 2-user case with channel gains $\gammait\in\{0,1\}$ while $\git=0$, for all $i=1,2$. The next lemma presents the rate region that the {\DOAC} achieves.

\begin{lma}
\label{Rate_Region}
Under the {\DOAC} policy, the queues of all users will be stable if and only if the arrival rate vector satisfies $\sum_{i\in\script{N}}\rho_i(\Pmax)<1$ with strict inequality.
\end{lma}

\begin{proof}
If $\sum_{i\in\script{N}}\rho_i(\Pmax)\geq 1$, then for any $\pi\in\script{P}$ we will have $\bW_{\pi_N}=\infty$ from \eqref{Delay}. Thus, the queue of at least one of the users will build up. Moreover, if the inequality $\sum_{i\in\script{N}}\rho_i(\Pmax)<1$ holds, we will have $\bW_i<\infty$ for all $i\in\script{N}$. Little's law completes the proof.
\end{proof}

The achievable region provided in Lemma \ref{Rate_Region} is shown in Fig. \ref{Capacity_Region} for the 2-user case. This is a straight line intersecting the two axes at $\lb\mu_1(\Pmax),0\rb$ and $\lb0,\mu_2(\Pmax)\rb$, respectively. Although, in general, this rate region lies strictly inside the capacity region, there are cases where the two regions coincide. Before presenting two of these examples, we note that in these cases the {\DOAC} is delay optimal and throughput optimal at the same time.

{\it Example 1 (Unknown channel gain):} If all SUs are not able to estimate the gain of their direct channel to their BS, then each SU would be transmitting with a fixed rate that corresponds to the minimum non-zero channel gain $\gammamini\triangleq \min_{\gamma_i^{(t)}\neq 0}\gamma_i^{(t)}$. Hence the capacity region shrinks \cite[pp. 115]{srikant2013communication} to be the region bounded by the hyper plane intersecting the $i$th axis at the point $[0,\cdots, \mu_i^{\rm min},0, \cdots]^T$ where
\begin{equation}
\mu_i^{\rm min}\triangleq \frac{\log\lb 1+\Pmax\gammamini\rb\lb 1- \Prob{\gammait=0}\rb}{L},
\label{mu_min}
\end{equation}
thus coinciding with the {\DOAC} achievable rate region.

{\it Example 2 (Non-fading channel):} When we have a non-fading channel, each SU transmits with a fixed rate equals $\log(1+\Pmax)$ bits per slot. Hence the capacity region becomes the region in the first quadrant that is bounded by the hyper plane intersecting the each axis at $\log(1+\Pmax)/L$, thus coinciding with the {\DOAC} achievable rate region.

\begin{figure}%
\centering
\includegraphics[width=\widthn\columnwidth]{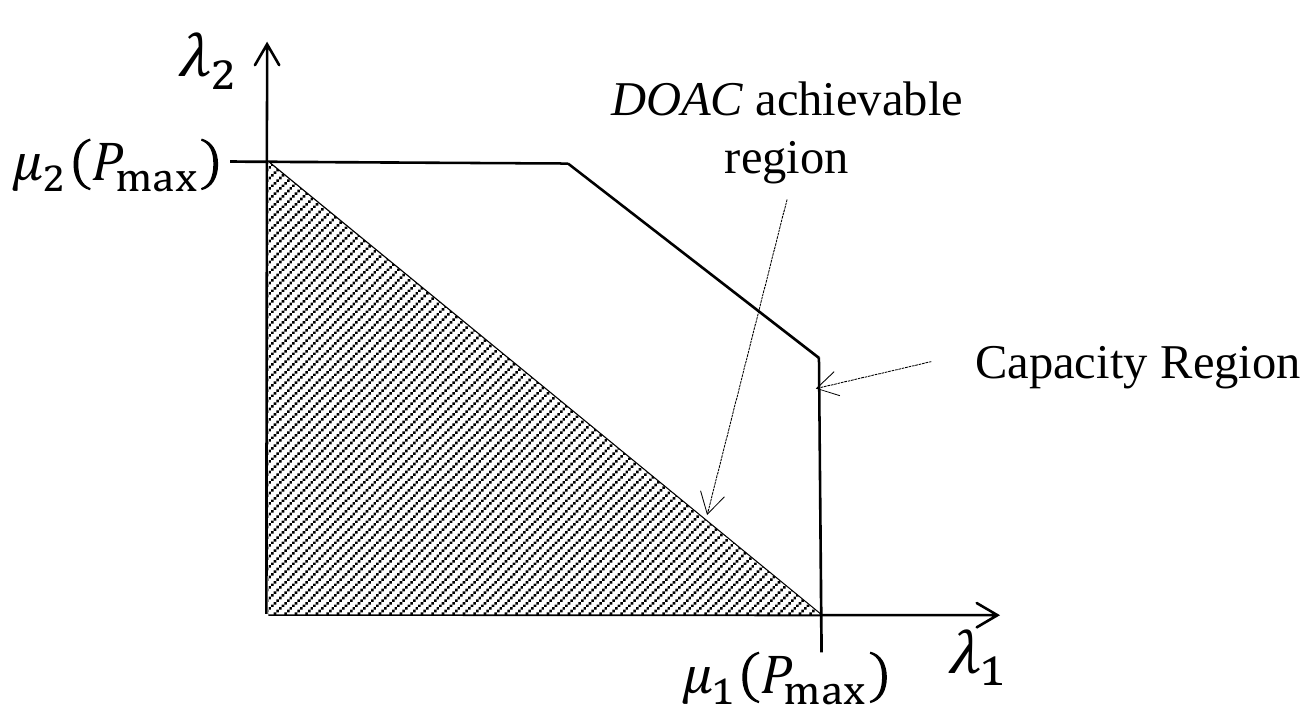}%
\caption{Capacity region for a 2-user case.}%
\label{Capacity_Region}%
\end{figure}

\section{Performance Under Channel Estimation Errors}
\label{Errors}
In this section, we present the solution of the system under channel estimation errors. 
We assume that SU $i$ estimates $\gamma_i^{(t)}$ and $g_i^{(t)}$ with $\alpha \%$ error relative to its actual value, where $\alpha>0$ represents the percentage of maximum deviation from the true value. This is a good model when the source of channel errors is mainly due to quantization. The observed values satisfy
\begin{align}
\lb 1-\frac{\alpha}{2}\rb\gamma_i^{(t)}\leq\gammaerr&\leq\lb 1+\frac{\alpha}{2}\rb\gamma_i^{(t)},
\label{gamma_Err}\\
\lb 1-\frac{\alpha}{2}\rb g_i^{(t)}\leq\gerr&\leq\lb 1+\frac{\alpha}{2}\rb g_i^{(t)},
\label{g_Err}
\end{align}
From equation \eqref{Tx_Rate}, in order to prevent outage, we need to consider the worst case scenario for $\gamma_i^{(t)}$. Therefore, we estimate $\gamma_i^{(t)}$ to be
\begin{equation}
\gamma_i^{(t)}=\frac{\gammaerr}{1+\frac{\alpha}{2}}.
\label{gamma_Estimate}
\end{equation}
Although this is a worst case estimation of $\gamma_i^{(t)}$, we will show through simulations that the reduction in performance is not high even with a relatively high value of $\alpha$. Similarly, instantaneous interference constraint in equation \eqref{Problem} is satisfied using a worst-case estimate of $g_i^{(t)}$ as
\begin{equation}
g_i^{(t)}=\frac{\gerr}{1-\frac{\alpha}{2}}.
\label{g_Estimate}
\end{equation}

With the estimated CSI values given by equations \eqref{gamma_Estimate} and \eqref{g_Estimate}, the two problems of instantaneous and average interference constraint, namely problems \eqref{Problem} and \eqref{Prob}, become functions of the observed CSI values as well as the parameter $\alpha$. Hence, the two policies {\DOIC} and {\DOAC} can be used to to solve problems \eqref{Problem} and \eqref{Prob}, respectively, under estimation errors. Section \ref{Results} simulates this system and shows the performance under this error model.

\section{Simulation Results}
\label{Results}
We simulated a system of $N=2$ SUs. Table \ref{Parameters} lists all parameter values for both scenarios; the instantaneous as well as the average interference constraint. We expect SU $1$ to have higher average delay in both scenarios. This is because it has a lower average channel gain and higher interference channel gain compared to those of SU $2$. However, the {\DOIC} policy can guarantee a bound on this delay using the constraint $\bW_1\leq d_1$, so that the QoS requirement of SU $1$ is satisfied. In our simulations we set $d_1=\dOne\Ts$ unless otherwise specified.
\begin{table}
	\centering
		\caption{Simulation Parameter Values}
		\label{Parameters}
		\begin{tabular}{|c|c||c|c|}
			\cline{1-4}
			Parameter & Value & Parameter & Value \\
			\cline{1-4}
			$L$ & $1000$ bits per packet & $\Iinst$ & 50 \\
			$\Rmax$ & $82$ bits per slot & $\Pmax$ & 100 \\
			$\lambda_1=\lambda_2=\lambda$ & $\lambda \in \{1,\cdots 10\}\times 10^{-3}$ packets/slot & $N$ & $2$ SUs \\
			$\fgammai(\gamma)$ & $\exp{\lb-\gamma/\bgamma_i\rb}/\bgamma_i$ & $\alpha$ & 0.1 \\
			$\fgi(g)$ & $\exp{\lb-g/\overline{g}_i\rb}/\overline{g}_i$ & $\epsilon$ & $0.1$ \\
			$(\bgamma_1,\bgamma_2)$ & $(2,4)$ & $V$ & $10$ \\
			$(\overline{g}_1,\overline{g}_2)$ & $(0.4,0.2)$ & $d_2$ & $\dTwo\Ts$
			\\ \cline{1-4}
			\end{tabular}
\end{table}

\subsection{Instantaneous Interference}
In Figures \ref{PerUser_Delay_Inst} and \ref{Avg_Delay_Inst} we consider problem \eqref{Problem} and assumed perfect knowledge of the direct and interference channel state information (CSI), namely $g_i^{(t)}$ and $\gamma_i^{(t)}$. Fig. \ref{PerUser_Delay_Inst} plots the average per-SU delay $\bW_i$, from equation \eqref{Delay}, for two cases; the first being the constrained optimization problem where $d_1=\dOne\Ts$ while setting $d_2$ to any arbitrarily high value (we set $d_2=\dTwo\Ts$), while the second is the unconstrained optimization problem where both $d_1$ and $d_2$ are set arbitrarily high (we set $d_1=d_2=\dTwo\Ts$). We call it the unconstrained problem because the average delay of both SUs is strictly below $\dTwo\Ts$, thus both delay constraints are inactive. The X-axis is the probability of a packet arrival per time slot $\lambda$, where $\lambda \triangleq \lambda_1=\lambda_2$. From Fig. \ref{PerUser_Delay_Inst} we can see a gap, in the unconstrained problem, between the average delay of SU $1$ and that of SU $2$. Hence, SU $1$ suffers from high delay. While for the constrained problem, the {\DOIC} policy has forced $\bW_1$ to be smaller than $\dOne\Ts$ for all $\lambda$ values. This comes at the cost of SU $2$'s delay. We conclude that the delay constraints in problem \eqref{Problem} can force the delay vector of the SUs to take any value as long as it is feasible.

\begin{figure}%
\centering
\includegraphics[width=\widthn\columnwidth]{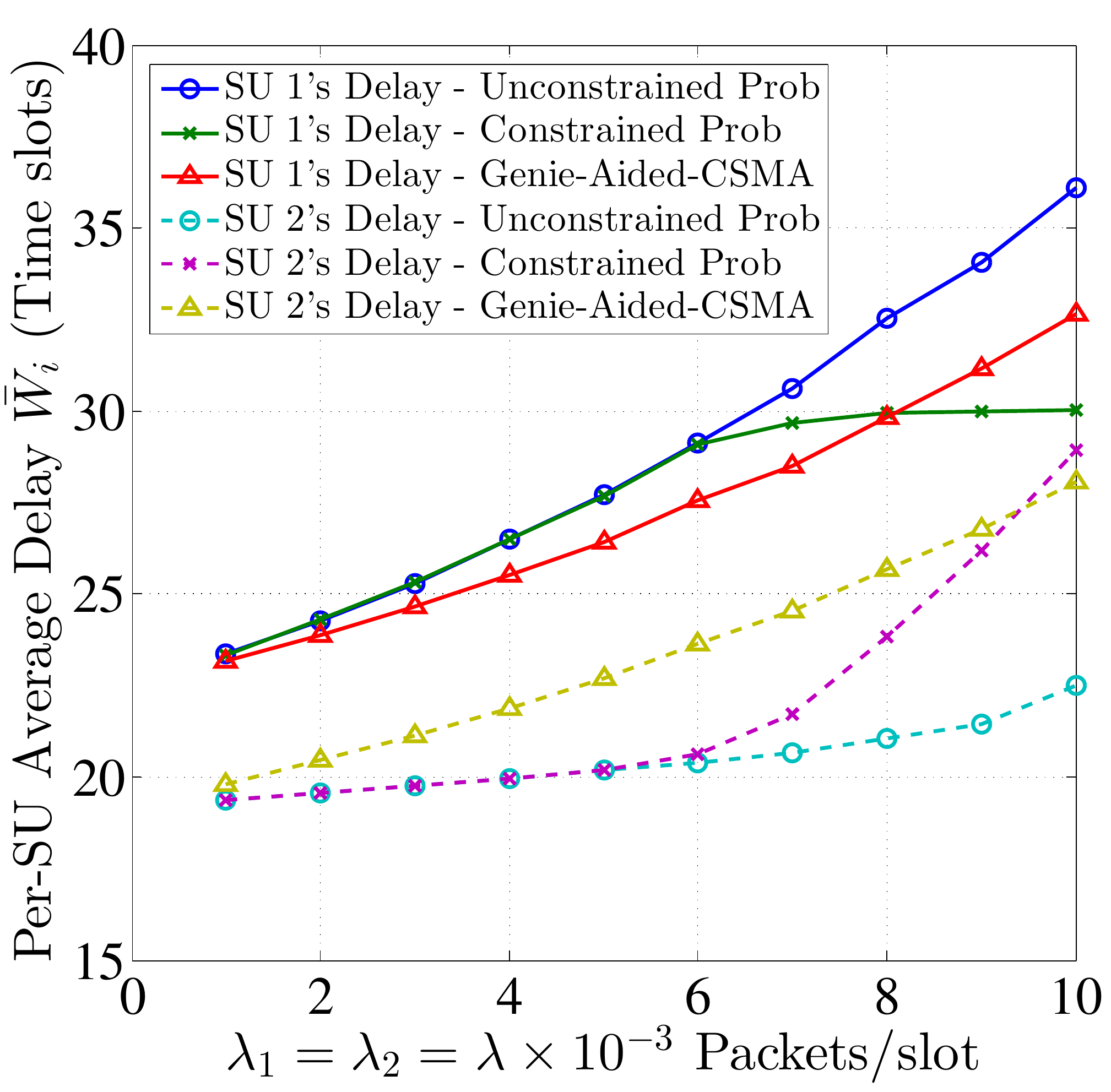}%
\caption{Average per-user delay for both the constrained and unconstrained optimization problems}%
\label{PerUser_Delay_Inst}%
\end{figure}

\begin{figure}%
\centering
\includegraphics[width=\widthn\columnwidth]{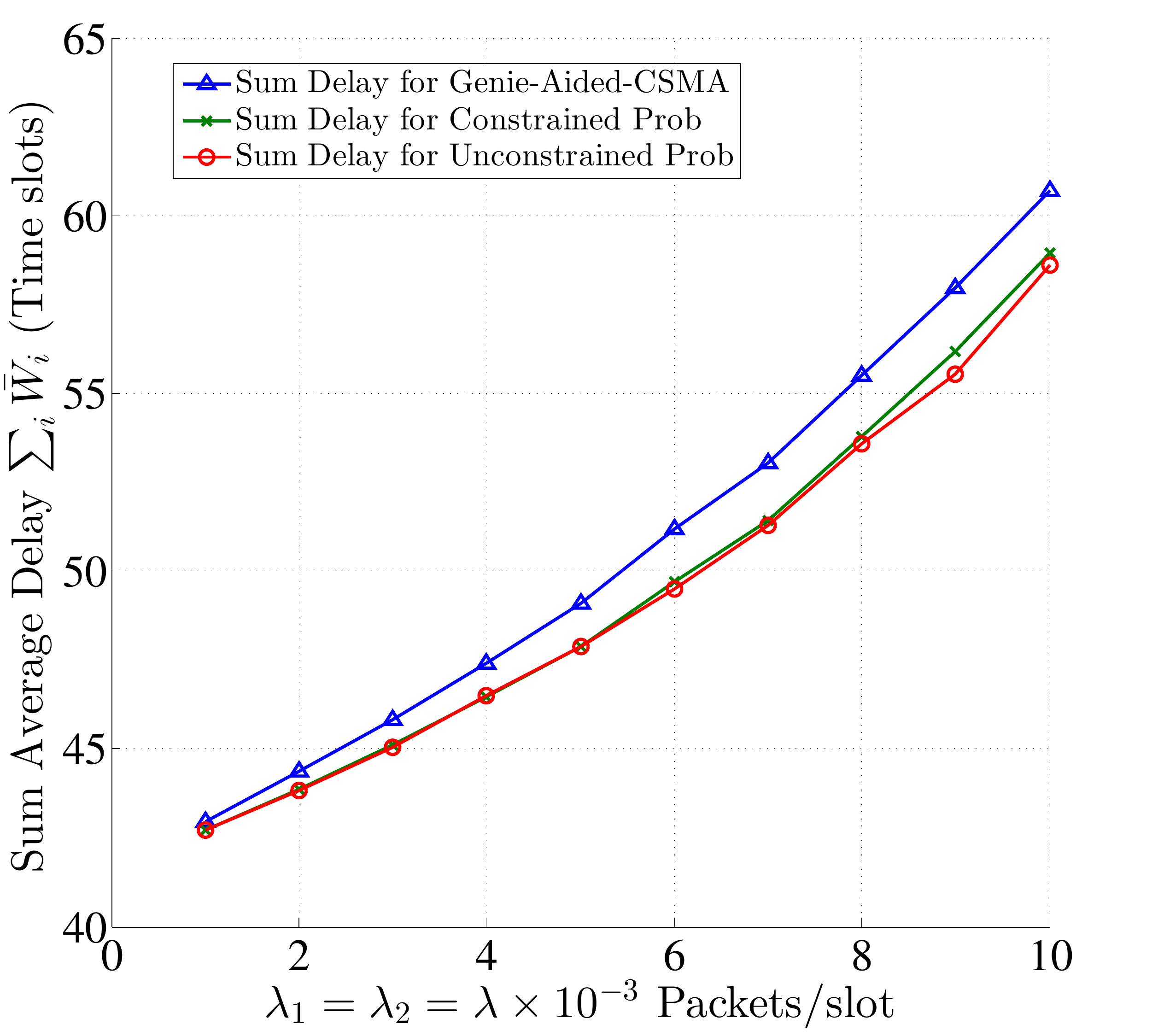}%
\caption{Sum of cost functions for the perfect CSI estimates for the {\DOIC} policy to solve problem \eqref{Problem}.}%
\label{Avg_Delay_Inst}%
\end{figure}

\subsection{Average Interference}
Problem \eqref{Prob} differs than problem \eqref{Problem} by an additional average interference constraint. This comes at the cost of the sum of average delays of SUs. We simulated the system with $d_1=\dTwo\Ts$ and compared it to the performance of the {\DOIC} policy with $d_1=\dTwo\Ts$ as well. The sum of average delays of the two SUs is plotted in Fig. \ref{Avg_Delay_Avg_Inst_CSMA} for both algorithms. The increase in the average delay for the {\DOAC} policy is due to adding an additional average interference constraint. However, when comparing the {\DOAC} policy to a Carrier-Sense-Multiple-Access (CSMA) scheduling policy we find it to have a lower average delay performance. This is because the CSMA allocates the channels randomly uniformly among users and does not prioritize the users based on their delay requirement $d_i$. On the other hand, the {\DOAC} allocates the channels based on the objective of minimizing the sum of average delays. We note that the CSMA policy plotted in Fig. \ref{Avg_Delay_Avg_Inst_CSMA} uses a ``genie-aided'' power allocation policy obtained from Algorithm \ref{DOACopt}. Thus, even when the two algorithms, the CSMA policy and the {\DOAC} policy, have the same power allocation policy, the {\DOAC} scheduling policy has an improved delay performance over the CSMA policy.

\begin{figure}%
\centering
\includegraphics[width=\widthn\columnwidth]{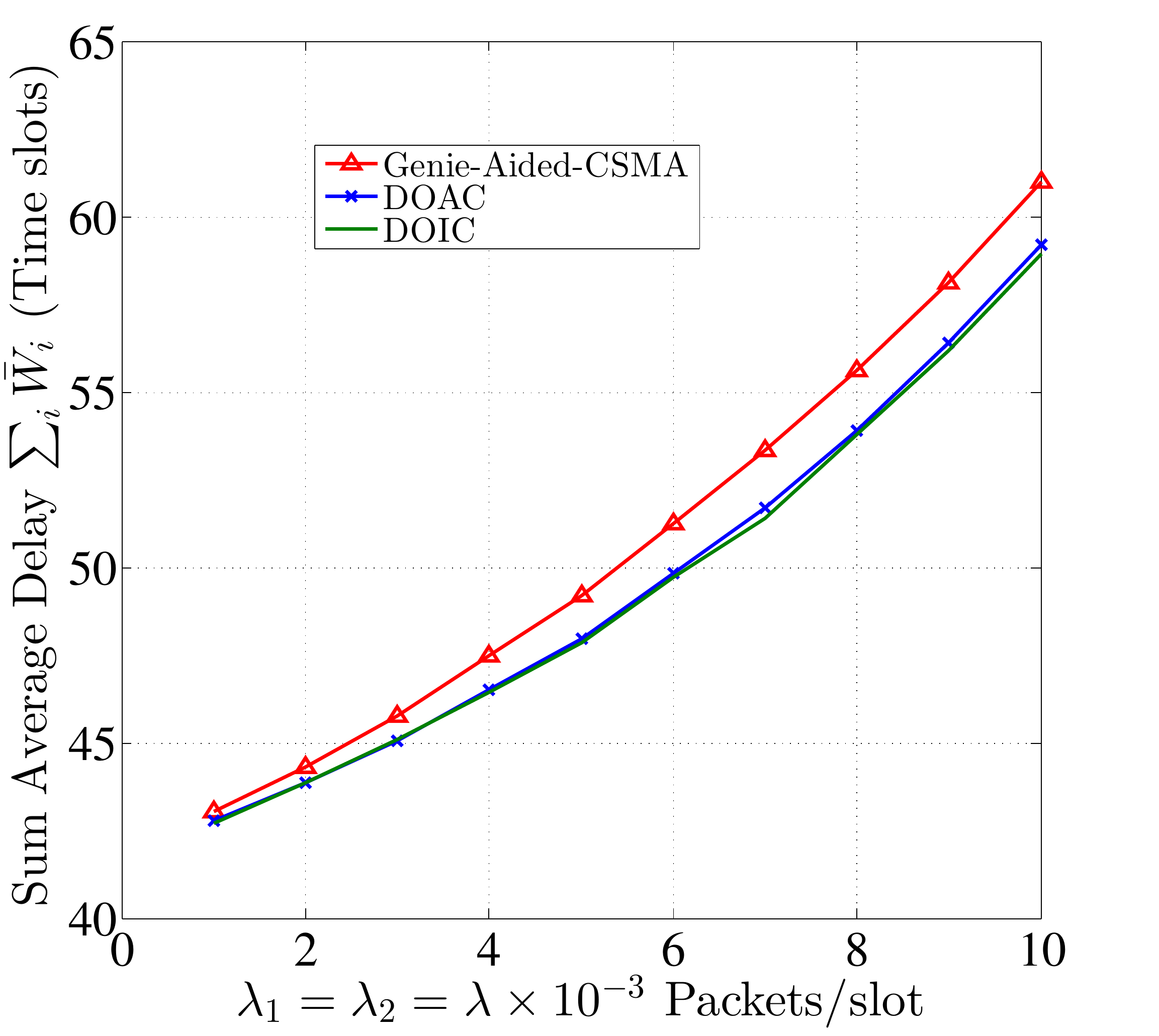}%
\caption{Comparing the CSMA policy with the {\DOIC} policy and the {\DOAC} policy. The power allocation scheme used for the {\DOAC} policy is the one used for the CSMA, hence the term genie-aided. However, the genie-aided CSMA policy has a worse delay performance compared to the {\DOAC} policy.}%
\label{Avg_Delay_Avg_Inst_CSMA}%
\end{figure}

\subsection{Low-Complexity Algorithm Performance}
When implementing the suboptimal algorithm proposed in Section \ref{Suboptimal} we find that the sum of the average delay across SUs is very close to its optimal value found via Algorithm \ref{DOACopt}. This is demonstrated in Fig. \ref{Avg_Delay_Avg_vs_Sub} where the error doesn't exceed $0.37\%$ at $\lambda=0.01$
\begin{figure}%
\centering
\includegraphics[width=\widthn\columnwidth]{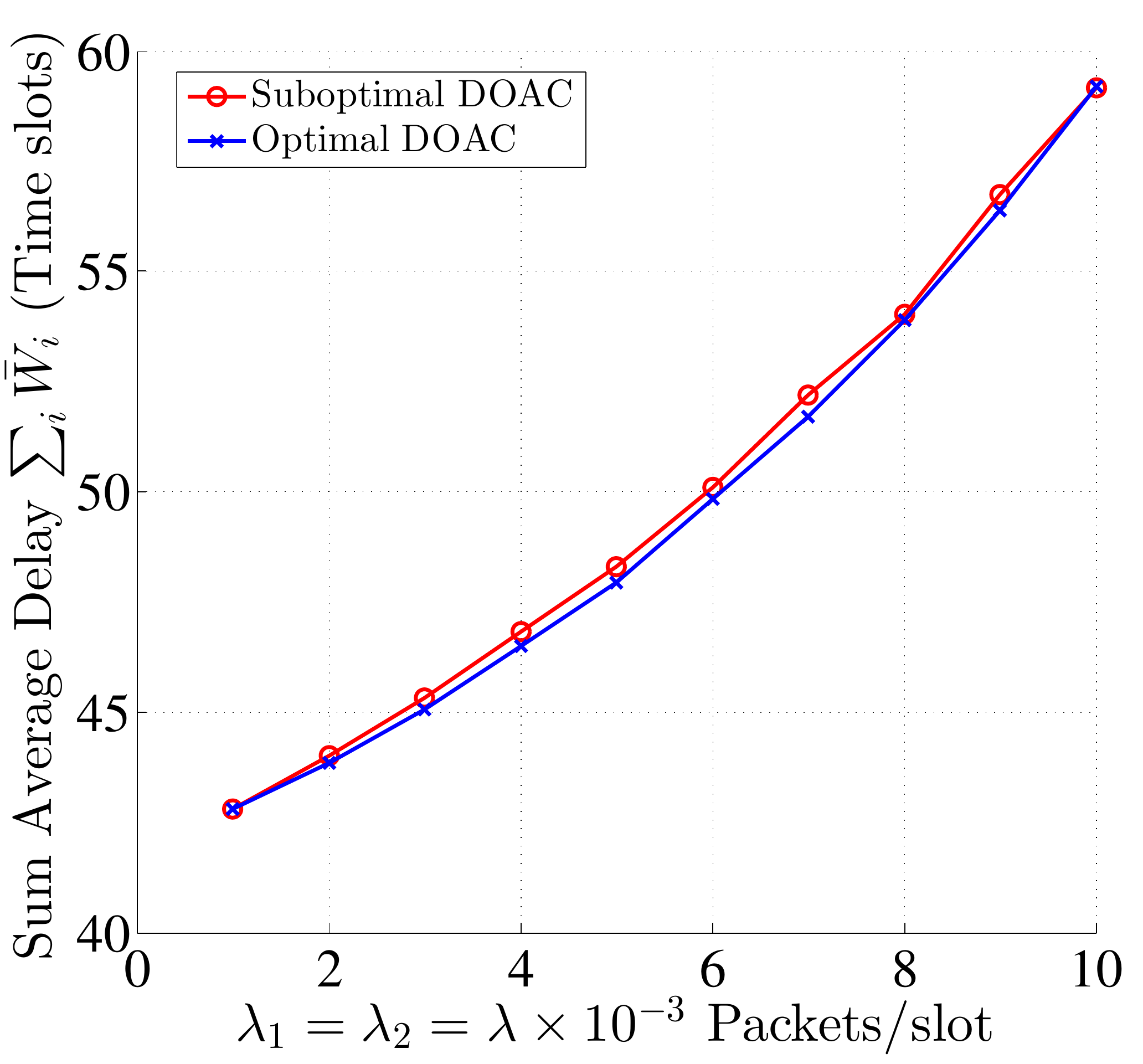}%
\caption{The low-complexity algorithm proposed in Section \ref{Suboptimal} has a close-to-optimal average delay performance with a maximum error of $0.37\%$.}%
\label{Avg_Delay_Avg_vs_Sub}%
\end{figure}

\subsection{CSI Estimation Errors}
For the imperfect CSI case, we assumed that each SU has an error of $\alpha=10\%$ in estimating each of $g_i^{(t)}$ and $\gamma_i^{(t)}$ and simulated the system with $d_1=32\Ts$. In order to avoid outage we substitute by equation \eqref{gamma_Estimate} in \eqref{Tx_Rate}. To guarantee protection to the PU from interference, we substitute equation \eqref{g_Estimate} in \eqref{Pow_Allocation} for the {\DOAC} policy, and in \eqref{Power_Allocation} for the {\DOIC} policy. From Fig. \ref{Avg_Delay_Avg_CSI_Err} we see that the performance difference between the perfect and the imperfect CSI problem, for the {\DOAC} policy, ranges between $2.4\%$ at $\lambda=10^{-3}$, and $9.5\%$ at $\lambda=10^{-2}$. We note that this performance difference represents an upper bound on the actual difference since the $10\%$ is an upper bound on the actual estimation error.

\begin{figure}%
\centering
\includegraphics[width=\widthn\columnwidth]{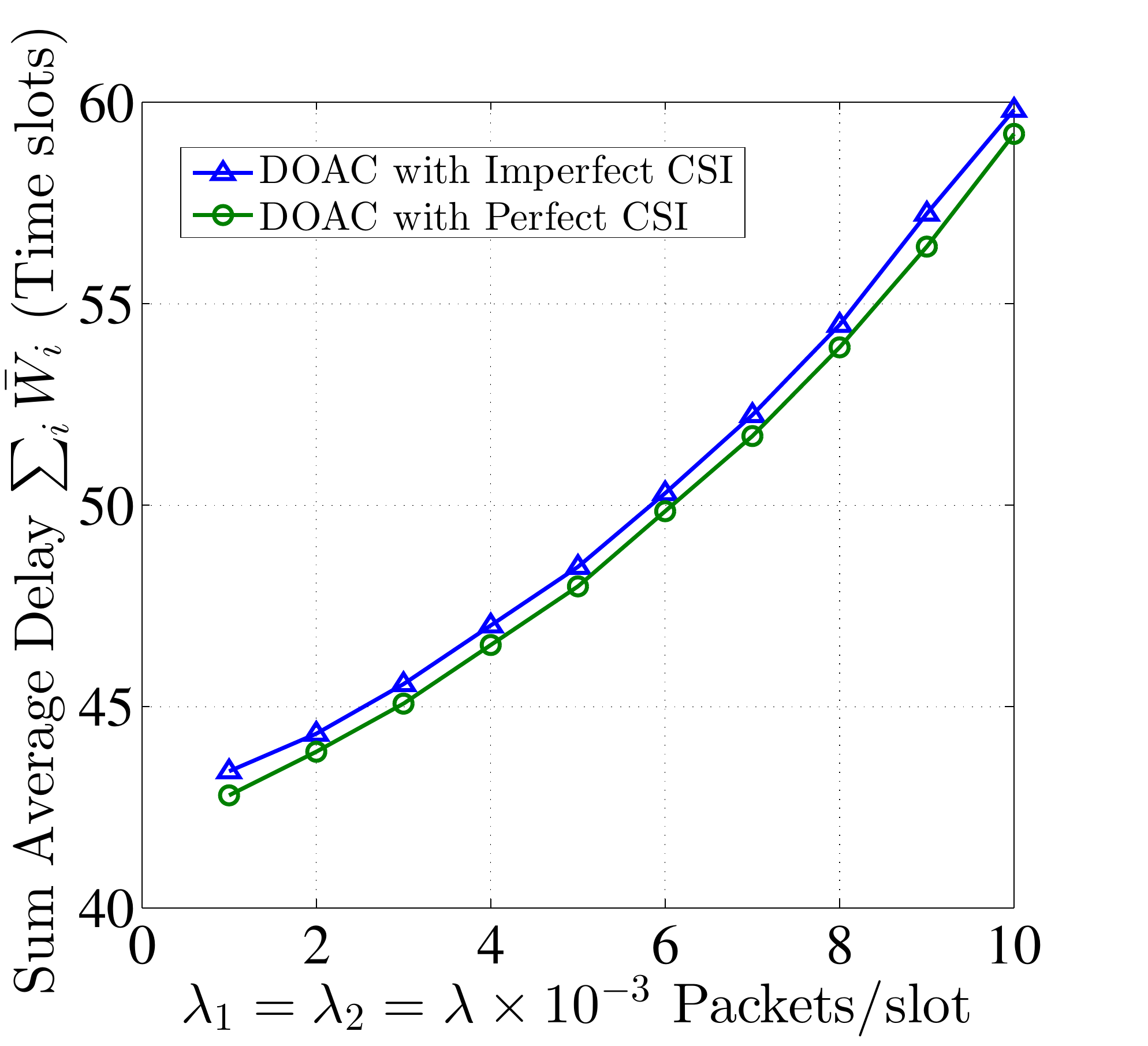}%
\caption{Sum of cost functions for the perfect as well as the imperfect channel sensing for the {\DOAC} policy to solve the constrained optimization problem \eqref{Prob}.}
\label{Avg_Delay_Avg_CSI_Err}
\end{figure}

\chapter{Proposed Future Work}
\label{Future_Work}
The problems presented in this report discuss the average delay of SUs where we propose algorithms to guarantee the desired QoS. However, the analysis depends on two main assumptions: 1) the SUs have a long-term average delay constraint which might not be realistic in some applications as online gaming where packets have a strict deadline that needs to be met and not a long-term delay constraint; 2) the arrival process of the SUs is stationary where the arrival rate vector is fixed. In this section we propose a possible solution to address each of these problems.

\section{Hard Deadline Guarantees}
\label{Hard_Deadline}
Some real-time applications require that packets arrive before a specified deadline. These applications are referred to as time-sensitive applications as online data streaming and online gaming. If the packet arrives after this deadline, it is considered useless and is not counted as a successfully-arrived packet. Hence, having a constraint on the average delay in the form of $\bW_i\leq d_i$ does not necessary guarantee that packets will be transmitted by this threshold $d_i$. Thus, we need the delay constraint to be formulated differently to capture this hard deadline nature. However, since we could not guarantee that all packets will be delivered by this deadline we will allow for a small percent of the packets to miss their deadlines as long as this percent is below a prespecified threshold.

The problem of hard deadlines has been considered in the literature for non CR systems (see \cite{Adaptive_NC_Deadline,A_Theory_of_QoS,kang2013performance,hou2010scheduling} and references therein). The work in \cite{Adaptive_NC_Deadline} assumes multiple packets sharing the same hard deadline and finds the optimal way to network-code these packets to maximize the throughput of a downlink system. The advantages of network coding are exploited when having a multicast system where all packets are to be sent to all users, as assumed in their work. However, when the system is a unicast one, network coding does not have any advantages \cite{ahlswede2000network}. In \cite{kang2013performance}, an algorithm was proposed to guarantee at least a fraction of the system's throughput capacity is met, under the constraint that a prespecified percentage of packets meet their deadlines. However, fading was out of the scope of their work. The work in \cite{hou2010scheduling} proposes a scheduling policy to guarantee that the percentage of packets that are not transmitted by their deadline is within a tolerable threshold. Although fading is considered in their model, the algorithm is valid for non CR systems where interference does not exist.

We now propose a possible direction to find a scheduling-and-power-allocation algorithm that guarantees packets to meet their deadlines as well as protecting the PU from interference. Assume that time is divided into time slots and that slots are grouped into frames of a fixed duration of $\DF$ slots each. We assume that a single packet arrives at the beginning of the frame for SU $i$ with probability $\lambda_i$ packets per frame that is i.i.d across frames and users. All packets have a fixed length of $L$ bits. However, different packets might need different number of time slots to be fully transmitted to the BS due to the fading nature of the wireless channels. We adopt the same channel model as in Chapter \ref{Channel_Model}. Channel gains are assumed independent across time slots and users, thus $\gamma_i^{(t)}$ is independent of $\gamma_j^{(t)}$ for all $i\neq j$ and all $t$. If SU $i$ is assigned the channel at slot $t$, then it transmits the $M_i^{(t)}$ bits of the HOL packet according to \eqref{Num_Bits}.

The goal is to find a scheduling algorithm that maximizes the sum rate of the SUs such that $q_i\%$ of the packets for user $i$ are fully transmitted by the deadline $\tau_i$. We call this constraint the delivery ratio requirement constraint. We assume that the deadline $\tau_i$ is within the frame where the packet has arrived. Moreover, we require that the average interference received by the PU does not exceed $\Iavg$.

In order to find an algorithm that satisfies the delivery ratio requirement, we set up the virtual queue calculated at the end of frame $k$ to represent the amount of undelivered packets for user $i$ up to the end of frame $k$
\begin{equation}
Z_i(k+1)=\lb Z_i(k) +\vert \script{A}_i(k)\vert-S_i(k)\rb^+,
\label{Delivery_VQ}
\end{equation}
where $S_i(k)=1$ if SU $i$ successfully transmits its packet before the deadline during frame $k$, and $0$ otherwise.

Given this formulation we believe that the algorithm can be deduced from the Lyapunov analysis. This is done by finding an upper bound on the Lyapunov drift. The policy would thus be the one that minimizes this upper bound. Moreover, we will compare the derived algorithm to a greedy algorithm. The greedy algorithm is nothing but the {\DOAC} policy with the value of $d_i$ set to $q_i\tau_i$. Although, according to the Markov inequality, this ensures that the probability of exceeding the deadline $\tau_i$ is below $q_i$, we believe that this policy will perform worse than the one we are seeking. This is because, as opposed to the {\DOAC}, this policy dynamically updates the scheduling according to the percentage of packets that missed their deadlines not according to the average delay experienced by these packets.

\section{Throughput Optimality}
\label{Throughput_Optimality}
We have seen that the {\DOAC} policy has an optimum delay performance as long as the arrival rate satisfy $\sum_i\rho_i(\Pmax)<1$. The region formed by this inequality is strictly smaller than the system's capacity region found in \cite{NeelyPhD}. The main reason for this loss in capacity region stems from the priority lists dictated by the proposed algorithm. As a future work, we are interested in studying some scheduling policies that could potentially achieve higher rate regions than that of the {\DOAC}.

One example of such policies is the Max-Weight-$\alpha$ algorithm that assigns the channel to the user with the largest $q_i^{\alpha_i}\Rit$, where $q_i$ is the queue length of user $i$ at the beginning of slot $t$, while $\alpha_i>0$ is some fixed constant. This policy is shown in \cite{Max_Weight_alpha} to be throughput optimal. This policy has the advantage of giving more weight to users with higher $\alpha_i$ values. Hence, users with higher $\alpha_i$ experience better delay performance. Moreover, this class of policies takes the advantage of multi-user diversity in the system since the rate is incorporated in the scheduling algorithm. On the other hand, the {\DOAC} policy uses a strict priority scheduling policy and neglects the channel gain information that potentially increases the rate due to the inherited multi-user diversity.

%
%

\chapter{Conclusion}
\label{Conclusion}
In this work we have studied the joint scheduling and power allocation problem of an uplink multi-SU CR system. The main goal of this work is to present low complexity algorithms that yield delay guarantees to the SUs while protecting the PU at the same time. We have shown that the delay mainly consists of two main factors; the service time and the queue-waiting time. The service time is mainly affected by the transmission power of the SU. While the queue-waiting time depends on the transmission power as well as the scheduling algorithm.

We studied the delay due to the service time in a single SU multichannel setup. In particular, to find the optimal power that maximizes the SU's throughput as well as guaranteeing that the SU's service time is below some prespecified threshold, we formulated the problem as an optimal stopping rule problem. The SU senses the channels sequentially and stops to transmit at the first channel that gives the highest reward. A closed-form expression for the threshold of stopping was given and proven to guarantee that the PU is protected from harmful interference. The performance of the proposed algorithm was compared to numerous baseline algorithms.

To address the delay due to the queue-waiting time, we formulated the problem as a delay minimization problem in the presence of average and instantaneous interference constraints to the PU, as well as an average delay constraint for each SU that needs to be met. Most of the existing literature that studies this problem either assume on-off fading channels or does not provide a delay-optimal algorithms which is essential for real-time applications. We proposed the {\DOAC} policy that dynamically updates the power allocation of the SUs as well as finding the optimal scheduling policy. The scheduling policy is found by dynamically updating a priority list based on the statistics as well as the history of the arrivals, departures and channel fading realizations. The proposed algorithm updates the priority list on a per-frame basis while controlling the power on a per-slot basis. 
We showed, through the Lyapunov optimization, that the proposed {\DOAC} policy is asymptotically delay optimal. That is, it minimizes the sum of any convex increasing function of the average delays of the SUs as well as satisfying the average interference and average delay constraints.

However, it is found that when the number of SUs $N$ in the system is large, the complexity of the {\DOAC} policy scales as $O(MN2^N)$, where $M$ is the number of points required to solve a one-dimensional search. Hence, we propose a suboptimal policy with a complexity of $O(N\log(N))$ that does not sacrifice the performance. Extensive simulation results showed the robustness of the {\DOAC} policy against CSI estimation errors.

Finally we have proposed two different directions as a possible future work. The first is studying the system in the presence of hard deadlines where packets are dropped if they are not transmitted by these deadlines. We presented a potential solution approach to this problem that depends on the Lyapunov analysis. Moreover, we proposed to study other scheduling policies that are throughput optimal in addition to being delay optimal.


\appendices
\chapter{Proof of Lemma \ref{Lma_Stochastic_Dominance}}
\label{Apdx_Stochastic_Dominance}
\begin{proof}
We carry out the proof by contradiction. Assume, for some $i$, that $\gamma_{\rm th}^*(i)<\lambdapst$. Thus the reward starting from channel $i$, $U_i \left([\gamma_{\rm th}^*(i), \gamma_{\rm th}^*(i+1),...,\gamma_{\rm th}^*(M)]^T,\Pst{i} \right)$, becomes
\begin{align}
\nonumber & \theta_i c_i \int_{\gamma_{\rm th}^*(i)}^{\infty}{\log(1+P_{1,i}^*\gamma) f_{\gamma}(\gamma)} \, d\gamma + \theta_i U_{i+1}^* \int_0^{\gamma_{\rm th}^*(i)}{f_{\gamma}(\gamma)} \, d\gamma \\
&\hspace{1in} + (1-\theta_i)U_{i+1}^*
\label{Stoch_Dom_Definition}
%
%
\\\nonumber \leq & \theta_i c_i \int_{\lambdapst}^{\infty}\hspace{0.2cm}{\log(1+P_{1,i}^*\gamma) f_{\gamma}(\gamma)} \, d\gamma + \theta_i U_{i+1}^* \int_0^{\lambdapst}\hspace{0.3cm}{f_{\gamma}(\gamma)} \, d\gamma \\
&\hspace{1in }+ (1-\theta_i)U_{i+1}^*
\label{Stoch_Dom_Inequality}
\\
\label{Stoch_Dom_w_lambda}
=&U_i \left([\lambdapst, \gamma_{\rm th}^*(i+1),...,\gamma_{\rm th}^*(M)]^T,\Pst{i} \right).
\end{align}
Where inequality (\ref{Stoch_Dom_Inequality}) follows by adding the term $\theta_i \lb \int_{\gammathst{i}}^{\lambdapst}{\pdf(\gamma)} \, d\gamma \rb U_{i+1}^*$ to (\ref{Stoch_Dom_Definition}) while (\ref{Stoch_Dom_w_lambda}) follows by the definition of the right-hand-side of (\ref{Stoch_Dom_Inequality}).
Using equation (\ref{Reward}), we can calculate the reward $U_{i-1}$ for both the left-hand-side and right-hand-side of the previous inequality. Thus the following inequality holds
\begin{align}
\nonumber &U_{i-1} \left([\gamma_{\rm th}^*(i-1),\gamma_{\rm th}^*(i),...,\gamma_{\rm th}^*(M)]^T,\Pst{i-1} \right) \leq \\
& U_{i-1} \left([\gamma_{\rm th}^*(i-1), \lambdapst,...,\gamma_{\rm th}^*(M)]^T,\Pst{i-1} \right).
\label{Stoch_Dom_Recursion}
\end{align}
Carrying out the last step recursively $i-2$ more times, we find the relation
\begin{align}
\nonumber &U_1 \left([\gamma_{\rm th}^*(1),...,\gamma_{\rm th}^*(i-1),\gamma_{\rm th}^*(i),...,\gamma_{\rm th}^*(M)]^T,\Pst{1} \right) \leq \\
& U_1 \left([\gamma_{\rm th}^*(1),...,\gamma_{\rm th}^*(i-1),\lambdapst ,...,\gamma_{\rm th}^*(M)]^T,\Pst{1} \right),
\label{Stoch_Dom_Contradiction}
\end{align}
which contradicts with the fact that $\gamma_{\rm th}^*(i)$ is optimal.
\end{proof}

\chapter{Proof of Theorem \ref{Thm_Unique_Solution_S}}
\begin{proof}
\label{Apdx_Unique_Solution_S}
We first get $S_i^*$, $U_i^*$ and $p_i^*$ by substituting by equations $\gammathst{i}$ and $P_{1,i}^*(\gamma)$ in the three equations (\ref{Average_Power}), (\ref{Reward}) and (\ref{Prob_recursive}), respectively. Then we differentiate with respect to $\lambdapst$, treating $\lambdadst$ as a constant, yielding
\begin{align}
\label{Deriv_S_i}
\nonumber \dSdlamp{i}=&-\theta_i \pdf(\gammathst{i}) \dgamdlamp{i} \left( c_i \Pist{i} - S_{i+1}^*\right) -\\
& \theta_i c_i \frac{\ccdf(\gammathst{i})}{\lb \lambdapst \rb^2} + \left( 1 - \theta_i \ccdf(\gammathst{i}) \right) \dSdlamp{i+1},\\
\label{Deriv_U_i}
\nonumber \dUdlamp{i}=&-\theta_i \pdf(\gammathst{i}) \dgamdlamp{i} \times \\
\nonumber &\left[ \lambdapst \lb c_i \Pist{i} - S_{i+1}^* \rb - \lambdadst \lb 1 - p_{i+1}^* \rb \right] - \\
 & \theta_i c_i \frac{\ccdf(\gammathst{i})}{\lambdapst}+ \left( 1 - \theta_i \ccdf(\gammathst{i}) \right) \dUdlamp{i+1},\\
\label{Deriv_p_i}
\dpdlamp{i}=&-\theta_i \pdf(\gammathst{i}) \dgamdlamp{i} \left( 1 - p_{i+1}^* \right) + \\
&\left( 1 - \theta_i \ccdf(\gammathst{i}) \right) \dpdlamp{i+1},
\end{align}
respectively. Multiplying equation (\ref{Deriv_S_i}) by $-\lambdapst$ and equation (\ref{Deriv_p_i}) by $\lambdadst$ then adding them to equation (\ref{Deriv_U_i}) we can easily show that, for all $i \in \script{M}$,
\begin{equation}
\dUSpdlamp{i}=0.
\label{Deriv_U_S_p_i_eq_zero}
\end{equation}

We now find the derivative of $\gammathst{i}$ with respect to $\lambdapst$ by 
differentiating both sides of equation (\ref{gamma_i_Equation}) with respect to $\lambdapst$, while treating $\lambdadst$ as a constant,then using equation (\ref{Deriv_U_S_p_i_eq_zero}), then rearranging we get
\begin{equation}
\dgamdlamp{i} = \frac{c_i \Pist{i}-S_{i+1}^*}{c_i  \frac{\lambdapst}{\gammathst{i}} \Pist{i}}.
\label{gammai_Deriv_Explicit}
\end{equation}
Substituting by equation (\ref{gammai_Deriv_Explicit}) in (\ref{Deriv_S_i}) we get
\begin{align}
\nonumber \dSdlamp{i}=&- \alpha_i \left[ c_i \Pist{i} - S_{i+1}^*\right]^2 -\theta_i c_i \frac{\ccdf(\gammathst{i})}{\lb \lambdapst \rb^2} + \\
&\left( 1 - \theta_i \ccdf(\gammathst{i}) \right) \dSdlamp{i+1},
\label{Deriv_S_i_recursive}
\end{align}
where $\alpha_i$ is given by
\begin{equation}
\alpha_i=\frac{\theta_i \pdf(\gammathst{i})}{c_i  \frac{\lambdapst}{\gammathst{i}} \Pist{i}} \geq 0,
\label{alpha}
\end{equation}
Now evaluating (\ref{Deriv_S_i_recursive}) at $i=M$ and $i=M-1$ we get
\begin{align}
\label{Deriv_S_M}
&\dSdlamp{M}=- \alpha_M \left[ c_M \Pist{M} \right]^2 -\theta_M c_M \frac{\ccdf(\gammathst{M})}{\lb \lambdapst \rb^2},\\
&{\mbox{and }}\nonumber \dSdlamp{M-1} =- \alpha_{M-1} \left[ c_{M-1} \Pist{M-1} -S_M^*\right]^2 \\
\nonumber &\hspace{0.85in}-\theta_{M-1} c_{M-1} \frac{\ccdf(\gammathst{M-1})}{\lb \lambdapst \rb^2}\\
&\hspace{0.85in}+\left( 1 - \theta_{M-1} \ccdf(\gammathst{M-1}) \right) \dSdlamp{M},
\label{Deriv_S_M__1}
\end{align}
respectively. We can see that $\dSdlamp{M} < 0$, hence $\dSdlamp{M-1} < 0$. By induction, let's assume that $\dSdlamp{i+1} < 0$. From (\ref{Deriv_S_i_recursive}) we get that
\begin{align}
\nonumber \dSdlamp{i} =&- \alpha_i \left( c_i \Pist{i} - S_{i+1}^*\right)^2 -\theta_i c_i \frac{\ccdf(\gammathst{i})}{\lb \lambdapst \rb^2} + \\
&\left( 1 - \theta_i \ccdf(\gammathst{i}) \right) \dSdlamp{i+1}<0
\label{Deriv_S_i_induction}
\end{align}
since all its terms are negative. Finally we find that $\dSdlamp{1}<0$ indicating that $S_1^*$ is monotonically decreasing in $\lambdapst$ given any fixed $\lambdadst \geq 0$.

Now, to get an upper bound on $\lambdapst$, we know that
\begin{equation}
S_i^*=\theta_i c_i \int_{\gammathst{i}}^\infty{ \lb \frac{1}{\lambdapst} - \frac{1}{\gamma} \rb \pdf(\gamma) \,d\gamma}+ \left[1-\theta_i \ccdf(\gammathst{i}) \right]S_{i+1}^*.
\label{Average_Power_Opt}
\end{equation}
We can upper bound the first term in (\ref{Average_Power_Opt}) by $\theta_i c_i / \lambdapst$, while $\left[1-\theta_i \ccdf(\gammathst{i}) \right]<1$. Using these two bounds we can write $S_1^* < \sum_{i=1}^M \theta_i c_i / \lambdapst$. But since $S_1^*=\Pavg$, the upper bound on $\lambdapst$, mentioned in Theorem \ref{Thm_Unique_Solution_S}, follows.
\end{proof}

\chapter{Proof of Lemma \ref{Lma_lambdadst_Bound}}
\label{Apdx_Lma_lambdadst_Bound}
\begin{proof}
We provide a proof sketch for this bound. We know that at the optimal point $p_1^*=\invDmax$ and that $p_1^*=\theta_1 \ccdf \lb\gammathst{1}\rb + \lb 1-\theta_1 \ccdf \lb\gammathst{1}\rb \rb p_2^*$. But since the second term in the latter equation is always positive, then
\begin{align}
\theta_1 \ccdf \lb\gammathst{1}\rb < \invDmax.
\label{Prob_Success_Opt_Inequality}
\end{align}
Substituting by \eqref{Gamma_Solution_Lambert_W} in \eqref{Prob_Success_Opt_Inequality} and rearranging we can upper bound $\lambdadst$ by
\begin{equation}
\nonumber\frac{c_1 \lb {\log {\lb \frac{\lambdapst}{\ccdf^{-1}\lb \frac{1}{\theta_1 \Dmax}\rb}\rb} - {\frac{\lambdapst}{\ccdf^{-1}\lb \frac{1}{\theta_1 \Dmax}\rb}} +1} \rb + U_2^*-\lambdapst S_2^*}{1-p_2^*}
\label{LambdaD_Bound_1}
\end{equation}
We can easily upper bound $\log {\lb \lambdapst/\ccdf^{-1}\lb 1/ \lb\theta_1 \Dmax \rb \rb\rb} - {\lambdapst/\ccdf^{-1}\lb 1/\lb\theta_1 \Dmax \rb\rb}$ by substituting $\lambdapmax$ for $\lambdapst$ when $\lambdapst<\ccdf^{-1}\lb 1/\lb\theta_1 \Dmax \rb\rb$ and by $1$ otherwise. Moreover, it can also be shown that $U_2^*<\Utwomax$, $p_2^*<\ptwomax$ and that $\lambdapst S_2^*>0$ and from Theorem \ref{Thm_Unique_Solution_S} we have $\lambdapst<\lambdapmax$, the proof then follows.
\end{proof}

\chapter{Proof of Theorem \ref{Optimality}}
\label{Optimality_Proof_Inst}
\begin{proof}
In this proof, we show that the drift-plus-penalty under this algorithm is upper bounded by some constant, which indicates that the virtual queues are mean rate stable \cite{georgiadis2006resource,urgaonkar2011optimal}.

We define the Lyapunov function as $L(k) \triangleq\frac{1}{2}\sum_{i=1}^N Y_i^2(k)$ and Lyapunov drift to be
\begin{equation}
\Delta (k) \triangleq \EEY{L(k+1) - L(k)},
\label{Drift_Def}
\end{equation}
Squaring equation \eqref{Delay_Q} then taking the conditional expectation we can write the following bound
\begin{equation}
\frac{1}{2}\E_{\bfY(k)} \left[ Y_i^2(k+1)-Y_i^2(k)\right] \leq Y_i(k) \EEY{\FDurK}\lambda_i \lb \EEY{W_i^{(j)}}-r_i(k)\rb + C_{Y_i}.
\label{Delay_Q_Sq2}
\end{equation}
where we use the bound $\EEY{\lb \sum_{j\in \script{A}_i(k)} W_i^{(j)}\rb^2}+\EEY{\lb\sum_{j\in \script{A}_(k)}r_i(k)\rb^2}<C_{Y_i}$. The derivation is similar to that in \cite[Lemma7]{li2011delay}. Given some fixed control parameter $V>0$, we add the penalty term $V\sum_i \EEY{r_i(k)\FDurK}$ to both sides of \eqref{Drift_Def}. Using the bound in \eqref{Delay_Q_Sq2} the drift-plus-penalty term becomes bounded by
\begin{align}
&\Delta \lb \bfU(k)\rb + V\sum_{i=1}^N \EEY{r_i(k)\FDurK}\leq C_Y+\EEY{\FDurK} \Phi \hspace{0.2in} {\rm where}
\label{Drift_Plus_Penalty1}\\
&\Phi\triangleq \sum_{i=1}^N \lb V -Y_i(k) \lambda_i\rb r_i(k)+ \sum_{j=1}^N Y_{\pi_j}(k) \lambda_{\pi_j} \EEY{W_{\pi_j}^{(j)}},
\end{align}
We define the {\DOIC} policy to be the policy that finds the values of $\bfpi(k)$, $\{\bfP{}^{(t)}\}$ and $\bfr(k)$ vector that minimize $\Phi$ subject to the instantaneous interference, the maximum power and the single-SU-per-time-slot constraints in problem \eqref{Problem}. We can observe that the variables $\bfr(k)$, $\{\bfP{}^{(t)}\}$ and $\bfpi(k)$ can be chosen independently from each other. Step 4.a in the {\DOIC} policy finds the optimum value of $r_i(k)$, $\forall i\in\script{N}$. Moreover, since $\EEY{W_i^{(j)}}$ is decreasing in $\Pit$ $\forall t\in\script{F}(k)$, the optimum value for $\Pit$ is equation \eqref{Power_Allocation}. Finally, from \cite{c_mu_Rule} the $c\mu$-rule can be applied to find the optimum priority list $\bfpi(k)$ which is given by Step 1 in the {\DOIC} policy.

Now, since the proposed {\DOIC} policy minimizes $\Phi$, this gives a lower bound on $\Phi$ compared to any other policy including the optimal policy that solves \eqref{Problem}. Hence, we now evaluate $\Phi$ at the optimal policy that solves \eqref{Problem} with the help of a genie-aided knowledge of $r_i(k)=\bW_i^*$ yielding $\Phi^{\rm opt}=V\sum_{i=1}^N \bW_i^*$, where we use $\EEY{W_i^{(j)}}=\bW_i^*$. Substituting by $\Phi^{\rm opt}$ in the right-hand-side (r.h.s.) of \eqref{Drift_Plus_Penalty1} gives an upper bound on the drift-plus-penalty when evaluated at the {\DOIC} policy. Namely
\begin{equation}
\Delta \lb \bfY(k)\rb + V\sum_{i=1}^N \EEY{r_i(k)\FDurK}\leq C_Y + V\sum_{i=1}^N \bW_i^*\EEY{\FDurK}
\label{DOIC_Genie}
\end{equation}
Taking $\EE{\cdot}$, summing over $k=0,\cdots,K-1$, denoting $\bfY_i(0)\triangleq 0$ for all $i\in\script{N}$, and dividing by $V\sum_{k=0}^{K-1} \EE{\FDurK}$ we get
\begin{equation}
\sum_{i=1}^N \frac{\EE{Y_i^2(K)}}{\sum_{k=0}^{K-1} \EE{\FDurK}}+ \sum_{i=1}^N \frac{\sum_{k=0}^{K-1}\EE{r_i(k)\FDurK}}{\sum_{k=0}^{K-1}\EE{\FDurK}} \overset{(a)}{\leq} \frac{aC_Y}{V} + \sum_{i=1}^N \bW_i^*\triangleq C_1.
\label{Optimal_Eq}
\end{equation}
where in the r.h.s. of inequality (a) we used $\EE{\FDurK}\geq \EE{I(k)}=1/a$, and $C_1$ is some constant that is not a function in $K$. To prove the mean rate stability of the sequence $\{Y_i(k)\}_{k=0}^\infty$ for any $i\in\script{N}$, we remove the first and third terms in the left-side of \eqref{Optimal_Eq} as well as the summation operator from the second term to obtain $\EE{Y_i^2(K)}/K \leq C_1$ $\forall i\in\script{N}$. Using Jensen's inequality we note that
\begin{equation}
\frac{\EE{Y_i(K)}}{K} \leq \sqrt{\frac{\EE{Y_i^2(K)}}{K^2}} \leq \sqrt{\frac{C_1}{K}}.
\label{Jensens}
\end{equation}
Finally, taking the limit when $K\rightarrow \infty$ completes the mean rate stability proof. On the other hand, to prove the upper bound in Theorem \ref{Optimality}, we use the fact that $r_i(k)$ and $\vert \script{A}_i(k) \vert$ are independent random variables (see step 4-a in {\DOIC}) to replace $\EE{\vert \script{A}_i(k) \vert {r_i(k)}}$ by $\lambda_i\EE{\FDurK r_i(k)}$ in equation \eqref{Wait_r_i}, then we take the limit of \eqref{Wait_r_i} as $K\rightarrow \infty$, use the mean rate stability theorem and sum over $i\in\script{N}$ to get
\begin{equation}
\sum_{i=1}^N \frac{\EE{\sum_{k=0}^{K-1} \lb\sum_{j\in \script{A}_i(k)}W_i^{(j)}\rb}}{\EE{\sum_{k=0}^{K-1}{\vert\script{A}_i(k)\vert}}} \leq \sum_{i=1}^N \frac{\sum_{k=0}^{K-1}\EE{r_i(k)\FDurK}}{\sum_{k=0}^{K-1}\EE{\FDurK}}\overset{(b)}{\leq} \frac{aC_Y}{V} + \sum_{i=1}^N \bW_i^*,
\label{Optimality_Eq2}
\end{equation}
where inequality (b) comes from removing the first summation in the left-side of \eqref{Optimal_Eq}. Taking the limit when $K\rightarrow \infty$ and using equation \eqref{Delay_Frame} completes the proof.
\end{proof}

\chapter{Existence of The Service Time Moments}
\label{No_Deep_Fade}
\begin{lma}
\label{No_Deep_Fade_Lemma}
Given any distribution for $\Pit\gamma_i^{(t)}$ the inequality $\EE{s_i^n}<\infty$ holds $\forall n\geq 1$. Moreover, when the power is given by $\Pit=\min\lb \Iinst/\git,P\rb$ for some fixed parameter $P\in[\Pminn_i,\Pmax]$, the inequality $\EE{s_i^2}\leq \lb L^2+L\lb 1-p_i(\Pminn_i)\rb\rb/p_i^2(\Pminn_i)$ holds with $p_i(P)\triangleq 1-\Prob{R_i(P)=0}$.
\end{lma}
\begin{proof}
We carry out the proof by bounding the moments of $s_i$ by the respective moments of the random variable $s_{{\rm B},i}$ which is the service time for a system with a binary transmission rate, i.e. a system with $\Rit\in\{0,1\}$. The proof of the first part of the lemma follows by showing that all the moments of $s_{{\rm B},i}$ are finite. The second part of the lemma is a special case where we set $\Pit=\min\lb \Iinst/\git,P\rb$. 

In this proof we drop the index $i$ for simplicity whenever it is clear from the context. Given some, possibly random, power allocation policy $\Pit$ define the i.i.d. random process $R_{{\rm B},i}^{(t)}\in\{0,1\}$, $t\geq 1$, with $\Prob{R_{{\rm B},i}^{(t)}=0}=\Prob{\Rit=0}$. Dropping the index $i$, the following inequality holds for any $x\geq 1$
\begin{equation}
\Prob{\sum_{t=1}^x \RBt \leq L}\geq\Prob{\sum_{t=1}^x \Rt \geq L},
\label{CDFs_Ordered}
\end{equation}
which says that the probability of transmitting $L$ bits or more in $x$ time slots is higher if the transmission process is $\Rt$ compared to the binary transmission process $\RBt$. Defining $\sB\triangleq \{\min x : \sum_{t=1}^x \RBt \geq L\}$ as the binary service time which is the number of time slots required to transmit $L$ bits given that the transmission process is $\RBt$, we can write
\begin{align}
\Prob{\sB\leq x}&=\Prob{\sum_{t=1}^x \RBt \geq L}\\
&\geq\Prob{\sum_{t=1}^x \Rt \geq L}\\
&=\Prob{s\leq x}
\label{CDFs_Ordered2}
\end{align}
According to the theory of stochastic ordering, when two random variables have ordered cumulative distribution functions, their respective moments are ordered \cite[equation (2.14) pp. 16]{rajan2014ordering}. In other words, if $\Prob{s\leq x}\geq\Prob{\sB\leq x}$, then $\EE{s^n}\leq\EE{\sB^n}$, $\forall n\geq 1$. It suffices to show that the moments of $\sB$ are finite.

Define $\sNB$ as a random variable following the negative binomial distribution \cite[pp. 297]{degroot2011probability} with success probability $1-\Prob{\RBt=0}$ while the number of successes equals $L$. $\sNB$ refers to the number of time slots having $\RBt=0$ before transmitting the $L$th bit. We can see that $\sB=\sNB+L$. Thus we have
\begin{align}
\EE{\sB^n}=\sum_{j=0}^n {n\choose j} \EE{\sNB^j}L^{n-j}< \infty,
\label{Binary_Moments}
\end{align}
where the inequality follows since all the moments of the negative binomial distribution exist \cite[pp. 297]{degroot2011probability}. The first part of the lemma holds.
 
For the second part of the lemma, we set $\Pit=\min\lb \Iinst/\git,P\rb$ for some deterministic parameter $P\geq\Pminn_i$ and define $p_i(P)\triangleq 1-\Prob{R_i(P)=0}$ with $R_i(P)$ defined in \eqref{Rate_Explicit}. Given the moment generating function of $\sNB$ as \cite[pp. 894]{degroot2011probability}
\begin{equation}
\EE{e^{x\sNB}}=\frac{p_i^L(P)}{\lb1- \lb 1-p_i(P)e^x\rb\rb^L},
\label{MGF_sNB}
\end{equation}
the first two moments of $\sNB$ can be derived as
\begin{align}
\EE{\sNB}&=\frac{\lb1-p_i(P)\rb L}{p_i(P)} , \hspace{1in}{\rm and}\\
\EE{\sNB^2}&=\frac{\lb 1-p_i(P)\rb^2L^2+\lb1-p_i(P)\rb L}{p_i^2(P)}.
\label{NB_Two_Moments}
\end{align}
These two moments can be shown to be decreasing in $p_i(P)$. The proof of the second part of the lemma follows using the bound $p_i(P)\geq p_i(\Pminn_i)$ and the inequality $\EE{s^2}\leq\EE{\sB^2}=\EE{\sNB^2}+2L\EE{\sNB}+L^2$.
\end{proof}

\chapter{Proof of Theorem \ref{Optimality_Avg}}
\label{Optimality_Proof}
\begin{proof}
This proof is similar to that in Appendix \ref{Optimality_Proof_Inst}. We define $\bfU(k)\triangleq [X(k) , \bfY(k)]^T$, the Lyapunov function as $L(k) \triangleq \frac{1}{2}X^2(k)+\frac{1}{2}\sum_{i=1}^N Y_i^2(k)$ and Lyapunov drift to be
\begin{equation}
\Delta (k) \triangleq \EEU{L(k+1) - L(k)}.
\label{Drift_Def_Avg}
\end{equation}
Squaring equation \eqref{Avg_Interf_Q} then taking the conditional expectation we can get the bound
\begin{equation}
\frac{\E_{\bfU(k)} \left[X^2(k+1)-X^2(k)\right]}{2} \leq C_X+X(k)\lb\EEU{\sum_{t\in\script{F}(k)}\Pit \git}-\Iavg\EEU{\FDurK}\rb,
\label{Interf_Q_Sq1}
\end{equation}
where we use the bound $\EEU{\lb\sum_{i=1}^N\sum_{t\in\script{F}(k)}\Pit \git\rb^2+\lb\Iavg \FDurK\rb^2}<C_X$ in equation \eqref{Interf_Q_Sq1} and omit the derivation of this bound. Given some fixed control parameter $V>0$, we add the penalty term $V\sum_i \EEU{r_i(k)\FDurK}$ to both sides of \eqref{Drift_Def_Avg}. Using the bounds in \eqref{Delay_Q_Sq2} and \eqref{Interf_Q_Sq1}, the drift-plus-penalty term becomes bounded by
\begin{align}
\Delta &\lb \bfU(k)\rb + V\sum_{i=1}^N \EEU{r_i(k)\FDurK}\leq C+\EEU{\FDurK}\chi(k),
\label{Drift_Plus_Penalty_Avg}\\
&{\rm where}\hspace{0.1in}\chi(k)\triangleq \sum_{i=1}^N \lb V -Y_i(k) \lambda_i\rb r_i(k)+\phi
\label{chi}\\
&{\rm with}\hspace{0.1in}\phi\triangleq\sum_{l=1}^N \lb Y_{\pi_l}(k) \lambda_{\pi_l} \EEU{W_{\pi_l}^{(j)}} + X(k)\lb\frac{\EEU{\sum_{t\in\script{F}(k)}P_{\pi_l}^{(t)} g_{\pi_l}^{(t)}}}{\EEU{\FDurK}} -\Iavg\rb\rb,
\end{align}
We define the {\DOAC} policy to be the policy that jointly finds $\bfr(k)$, $\{\bfP{}^{(t)}\}$ and $\bfpi(k)$ that minimize $\chi(k)$ subject to the instantaneous interference, the maximum power and the single-SU-per-time-slot constraints in problem \eqref{Prob}. Step 5-a in the {\DOAC} policy minimizes the first summation of $\chi(k)$. For $\{\bfP{}^{(t)}\}$ and $\bfpi(k)$, we can see that $\phi$ is the only term in the right side of equation \eqref{chi} that is a function of the power allocation policy $\{\bfP{}^{(t)}\}$, $\forall t\in\script{F}(k)$. For a fixed priority list $\bfpi(k)$, using the Lagrange optimization to find the optimum power allocation policy that minimizes $\phi$ subject to the aforementioned constraints yields equation \eqref{Pow_Allocation}, where $P_{\pi_j}(k)$, $\forall i\in\script{N}$, is some fixed power parameter that minimizes $\Psi$ subject to the maximum power constraint only. Substituting by equation \eqref{Pow_Allocation} in $\phi$ and using the bound $\EEU{W_{\pi_l}^{(j)}}=W_{\pi_l}(P_{\pi_l}(k))\leq W_{\pi_l}^{\rm up}(P_{\pi_l}(k))$ we get $\Psi$ that is defined before equation \eqref{Optimization_Obj}. Consequently, $\bfPst(k)$ and $\bfpist(k)$, the optimum values for $\bfP{}(k)$ and $\bfpi(k)$ respectively, are ones that minimize $\Psi$ as given by Algorithm \ref{DOACopt}.

Since the optimum policy that solves \eqref{Prob} satisfies the interference constraint, i.e. satisfies $\EEU{\sum_{t\in\script{F}(k)}P_{\pi_l}^{(t)} g_{\pi_l}^{(t)}} \leq\EEU{\FDurK}\Iavg$, we can evaluate $\chi(k)$ at this optimum policy with a genie-aided knowledge of $r_i(k)=\bW_i^*$ to get $\chi^{\rm opt}\triangleq V\sum_{i=1}^N\bW_i^*$. Replacing $\chi(k)$ with $\chi^{\rm opt}$ in the r.h.s. of \eqref{Drift_Plus_Penalty_Avg} we get the bound $\Delta \lb \bfU(k)\rb + V\sum_{i=1}^N \EEU{r_i(k)\FDurK}\leq C+\EEU{\FDurK}V\sum_{i=1}^N\bW_i^*$. Taking $\EE{\cdot}$ over this inequality, summing over $k=0,\cdots,K-1$, denoting $X(0)\triangleq \bfY_i(0)\triangleq 0$ for all $i\in\script{N}$, and dividing by $V\sum_{k=0}^{K-1} \EE{\FDurK}$ we get
\begin{equation}
\frac{\EE{X^2(K)}}{\sum_{k=0}^{K-1}\EE{\FDurK}}+\sum_{i=1}^N \frac{\EE{Y_i^2(K)}}{\sum_{k=0}^{K-1} \EE{\FDurK}}+ \sum_{i=1}^N \frac{\sum_{k=0}^{K-1}\EE{r_i(k)\FDurK}}{\sum_{k=0}^{K-1}\EE{\FDurK}} \leq \frac{CK}{V\sum_{k=0}^{K-1}\EE{\FDurK}} + \sum_{i=1}^N \bW_i^*.
\label{Optimal_Eq_Avg}
\end{equation}
Similar steps to those in Appendix \ref{Optimality_Proof_Inst} can be followed to prove the mean rate stability of $\{X(k)\}_{k=0}^\infty$ and $\{Y_i(k)\}_{k=0}^\infty$ as well as the bound in Theorem \ref{Optimality_Avg}, and thus are omitted here.
\end{proof}



\startsinglespace




\bibliographystyle{ieeebib}
\bibliography{MyLib}

\end{document}